%% file: usenix-main24.tex
\documentclass[letterpaper,twocolumn,10pt]{article}
\usepackage{usenix2019_v3}
\usepackage{amsthm}
\usepackage{tikz}
\usepackage{bm}
\usepackage{ulem}

\def\X{\mathbf{X}}
\newcommand{\Tr}{Tr}
\def\x{\mathbf{x}}

\def\XX{\bm{\Xi}}

\def\A{\mathbf{A}}
\def\S{\bm{\Sigma}}
\def\P{\bm{\Psi}}

\def\R{\mathbb{R}}
\def\L{\mathcal{L}}

\def\L{\mathcal{L}}

\newcommand*{\defeq}{\stackrel{\text{def}}{=}}
\makeatletter
\newcommand*{\rom}[1]{\expandafter\@slowromancap\romannumeral #1@}
\makeatother

\usepackage{multirow}

\newtheorem{definition}{Definition}
\newtheorem{theorem}{Theorem}
\newtheorem{proposition}{Proposition}
\newtheorem{property}{Property}
\newtheorem{lemma}{Lemma}
\newtheorem{corollary}{Corollary}
\newtheorem{remark}{Remark}

\usepackage{amsfonts,amssymb}
\usepackage{stfloats}
\usepackage{mathtools}

\usepackage{xcolor}
\usepackage{soul,color}
\usepackage{amsmath}
\newcommand\blfootnote[1]{%
  \begingroup
  \renewcommand\thefootnote{}\footnote{#1}%
  \addtocounter{footnote}{-1}%
  \endgroup
}

\begin{document}

\date{}

\title{\Large \bf Less is More: Revisiting the Gaussian Mechanism for Differential Privacy}
\author{
{\rm Tianxi Ji}\\
Texas Tech University\\
tiji@ttu.edu
\and
{\rm Pan Li}\\
Case Western Reserve University\\
lipan@case.edu
} 

\maketitle

\begin{abstract}
Differential privacy (DP) via output perturbation has been a \textit{de facto} standard for releasing query or computation results on sensitive data. Different variants of the classic Gaussian mechanism  have been developed to reduce the  magnitude  of the   noise and improve the utility of   sanitized  query results.  However, we  identify that all existing Gaussian mechanisms suffer from  \textbf{the curse of full-rank covariance matrices}, and hence the expected accuracy losses of these mechanisms   equal   the trace of the covariance matrix of the noise. Particularly, for query results in $\R^M$ (or $\R^{M\times N}$ in a matrix form), in order to achieve $(\epsilon, \delta)$-DP,   the expected accuracy loss of the classic Gaussian mechanism, that of the analytic Gaussian mechanism, and that of the  Matrix-Variate Gaussian (MVG) mechanism are lower bounded by $C_C  (\Delta_2 f)^2$, $C_A  (\Delta_2 f)^2$, and $C_M  (\Delta_2 f)^2$, respectively, where $C_C =  \frac{2\ln\left(\frac{1.25}{\delta}\right)}{\epsilon^2} M$, $C_A = \frac{\left(\Phi^{-1}(\delta)\right)^2+ \epsilon}{ \epsilon^2} M$,   $C_M =  \frac{ \frac{5}{4}H_r+\frac{1}{4}H_{r,\frac{1}{2}}}{2\epsilon} MN$,    $\Delta_2 f$ is the $l_2$ sensitivity of the query function $f$, $\Phi(\cdot)^{-1}$ is the quantile function of the standard normal distribution, $H_r$ is the $r$th harmonic number, and $H_{r,\frac{1}{2}}$ is the $r$th harmonic number of order   $\frac{1}{2}$.

To lift this curse, we  design a \textbf{R}ank-1 \textbf{S}ingular \textbf{M}ultivariate \textbf{G}aussian  (R1SMG) mechanism. It achieves $(\epsilon,\delta)$-DP   on   query results in $\mathbb{R}^M$ by perturbing the results with  noise following a singular multivariate Gaussian distribution, whose covariance matrix  is a \textbf{randomly} generated rank-1 positive semi-definite matrix.  In contrast, the classic Gaussian mechanism and its variants all consider \textbf{deterministic} full-rank covariance matrices. Our idea  is motivated by a clue from Dwork et al.'s seminal work on the classic Gaussian mechanism that has been ignored in the literature: when projecting    multivariate  Gaussian noise with a full-rank covariance matrix onto a set of orthonormal basis in $\mathbb{R}^M$, only the coefficient of a single basis can  contribute to the privacy  guarantee.\blfootnote{A preliminary version of this paper appears in the proceedings of USENIX Security Symposium 2024. This is a full version.}

\

This paper makes the following technical contributions.

\textbf{(i)} The R1SMG mechanisms achieves    $(\epsilon,\delta)$-DP guarantee on query results in $\R^M$, while  its expected accuracy loss is lower bounded by  $C_R(\Delta_2f)^2$, where $C_R = \frac{2}{\epsilon \psi}$ and $\psi = \Big(\frac{\delta\Gamma(\frac{M-1}{2})}{\sqrt{\pi}\Gamma(\frac{M}{2})}\Big)^{\frac{2}{M-2}}$. We show that $C_R$ is on a lower order of magnitude by at least $M$ or $MN$ compared with $C_C$, $C_A$, and $C_M$. Therefore, the expected accuracy loss of the R1SMG mechanism is on a much lower order compared with   that of the classic Gaussian mechanism, of the analytic Gaussian mechanism, and of the MVG mechanism.

 \textbf{(ii)}   
 Compared with other mechanisms, the R1SMG mechanism is  more stable and  less likely  to generate noise with  large magnitude that overwhelms the query results, because  the kurtosis and skewness of the nondeterministic accuracy loss introduced by    this mechanism  is larger   than that introduced by other mechanisms.
\end{abstract}


\input{sections/intro}

\input{sections/preliminiaries}

\input{sections/identify-curse}

\input{sections/hidden_evidence}

\input{sections/mian_tech}

\input{sections/utility_compare}

\input{sections/experiments}

\input{sections/related_work}

\input{sections/conclusions}

\bibliographystyle{plain}
\bibliography{bibfile1.bib}

\input{sections/appendix}

\end{document}

%% file: sections/intro.tex
\section{Introduction}
\label{sec: int}

Differential privacy (DP) \cite{dwork2006our,dwork2014algorithmic} has been increasingly recognized as the fundamental building block   for privacy-preserving database query, sharing, and analysis. In this area, one important privacy guarantee is   $(\epsilon,\delta)$-DP. 
It means that, except with a  small probability of $\delta$, 
the presence or absence of any  data record in a dataset cannot change the probability of observing   a specific query result of the dataset   by   more than a multiplicative factor of $e^{\epsilon}$.

The classic Gaussian mechanism,   proposed by Dwork et al. \cite{dwork2006our}, is an essential tool to achieve $(\epsilon,\delta)$-DP. 
 Particularly,  when the output of a  query function  is a high dimensional vector, \textbf{the classic Gaussian mechanism will add independent and identically distributed (i.i.d.) noise to each component of the result} (see Definition \ref{def:def_GM} and        \cite[p. 261]{dwork2014algorithmic},        \cite[p. 3]{dwork2006our}). 
 
This  standard practice 
has spread  widely into many fields. For example, a differentially private mechanism returns noisy answers to    a set of   queries  by adding i.i.d. noise to each query result \cite{nikolov2013geometry}. Privacy-preserving   learning employs differentially private stochastic gradient descent through perturbing each component of the gradient with i.i.d. Gaussian noise \cite{abadi2016deep}. Unfortunately, as we will show later, the  adoption of the classic  Gaussian mechanism greatly hinders the utility (accuracy) of the sanitized   query results  because the expected accuracy loss (Definition \ref{def:utility_cost}) of the classic Gaussian mechanism   is on the order of $M(\Delta_2 f)^2$, where $\Delta_2 f$ is the $l_2$ sensitivity of the query function and $M$ is the dimension of the query results. We refer to this as the \textbf{curse of full-rank covariance matrices}.

\begin{definition} \textbf{Accuracy Loss~\cite{nikolov2013geometry,hardt2010geometry}.}\label{def:utility_cost} The accuracy loss of a differentially private  output perturbation mechanism ($\mathcal{M}$)   is  
\begin{equation}\label{l}
    \L  = ||\mathcal{M}(f(\bm{x}))  - f(\bm{x})||_2^2 = ||\mathbf{n}||_2^2,
\end{equation}
where $\bm{x}$ denotes a dataset, 
$f(\bm{x})$ is the query result on dataset $\bm{x}$, and $\mathbf{n}$ is the noise vector added to the query result. 
\end{definition}

We can see from (\ref{l}) that the accuracy loss is equal to the magnitude of the additive noise. Due to the randomness involved in the     noise, $\L$ is nondeterministic. Thus,   we investigate the   expected value of accuracy loss, i.e., $\mathbb{E}[\L]$.

\subsection{The Curse of Full-Rank Covariance Matrices}\label{sec:curse_intro}
Here, we   provide a high-level description on the identified curse of full-rank covariance matrices, while  the details are deferred to Section~\ref{sec:curse}.

\begin{proposition}\label{prop:curse} \textbf{The Identified Curse.} Let $\bm{x}$ be a dataset, $f(\bm{x})\in \R^M$   the queried results, and $\mathbf{n}\in \R^M$   the perturbation noises  introduced by the classic Gaussian mechanism (or its variants, e.g., the analytic Gaussian mechanism~\cite{balle2018improving}, and the  Matrix-Variate
Gaussian (MVG) mechanism~\cite{Chanyaswad2018}), respectively. Then, the expected accuracy loss is equal to the trace of the covariance matrix of the Gaussian noise $\mathbf{n}$, i.e.,  $\mathbb{E}[\L] = \Tr[Cov(\mathbf{n})]$,
where $Cov(\mathbf{n})$ denotes the covariance matrix of the noise $\mathbf{n}$.
\end{proposition}
Since existing Gaussian mechanisms always introduce    noise that has a full-rank covariance matrix, i.e., $rank(Cov(\mathbf{n}))=M$, where  $Cov(\mathbf{n})$ is  symmetric and positive definite \cite{casella2002statistical}, $\mathbb{E}[\L]$  is   equivalent to the summation of $M$ positive   eigenvalues of $Cov(\mathbf{n})$ and hence generally high.   In   Section~\ref{sec:curse}, we will prove Proposition \ref{prop:curse} when $f(\bm{x})$ is in a vector form and in a matrix form under    different variants of the classic Gaussian mechanism.

\subsection{Lifting the Curse: Main  Contributions}\label{sec:contribution}


To lift the  identified curse, we   resort  to  the  \textbf{singular multivariate Gaussian distribution} \cite{khatri1968some,cramer1999mathematical,rao1973linear,srivastava2009introduction,kwong1996singular} (Definition \ref{singular_pdf}) with a rank-1 covariance matrix and develop  the \textbf{R}ank-1  \textbf{S}ingular \textbf{M}ultivariate \textbf{G}aussian (R1SMG) mechanism. Our idea 
is motivated by an  ignored clue in Dwork and Roth's proof for the DP guarantee of the Gaussian mechanism  \cite{dwork2014algorithmic}.    \uline{Specifically,  Dwork and Roth proved that   by projecting the additive noise $\mathbf{n}\in\mathbb{R}^M$ onto a fixed set of orthonormal basis vectors (e.g., $\bm{b}_1, \bm{b}_2, \cdots, \bm{b}_M$), only the coefficient of a single basis vector (e.g., $\mathbf{n}^T\bm{b}_1$)    contributes to the privacy  guarantee (more details are deferred to Section \ref{sec:evidence}). However, this important clue   has been     ignored  in the literature.} To the best of our knowledge, our work is the first to discover and leverage this ignored clue.

The main contributions  of this paper are summarized as follows.

\noindent\textbf{(1) A new DP Mechanism and Privacy Guarantee (Definition \ref{def:rsmgm} and Theorem 
\ref{thm_multivariate}).} We develop the   R1SMG mechanism that can   achieve $(\epsilon,\delta)$-DP guarantee on high dimensional  query results $f(\bm{x})\in\R^M$ ($M> 2$).  The covariance matrix, denoted as $\bm{\Pi}$, of the noise introduced by the R1SMG mechanism is a random rank-1 symmetric positive semi-definite matrix, whose eigenvector is randomly sampled from the unit sphere embedded in $\R^M$.  Let $\sigma_*$ be  the eigenvalue\footnote{The eigenvalues and singular values are identical for symmetric positive (semi)-definite real-valued matrices. Thus, we use ``eigenvalues'' and ``singular values'' interchangeably for such matrices like $\bm{\Pi}$ in this paper.} of the rank-1 covariance matrix $\bm{\Pi}$. 
We prove that a sufficient condition for the R1SMG mechanism to achieve $(\epsilon,\delta)$-DP is $\sigma_{*} \geq  \frac{2(\Delta_2f)^2}{\epsilon \psi}$, where  $\Delta_2f$ is the $l_2$ sensitivity of the query function $f$, $\psi = \Big(\frac{\delta\Gamma(\frac{M-1}{2})}{\sqrt{\pi}\Gamma(\frac{M}{2})}\Big)^{\frac{2}{M-2}}$,  and   $\Gamma(\cdot)$ is Gamma  function. Note that similar to that in the classic Gaussian mechanism where $0<\epsilon<1$, we have $0<\epsilon<1/M$ in the R1SMG mechanism. In other words, The R1SMG mechanism can work  well under very strict privacy budget.

\noindent\textbf{(2) Expected accuracy loss on a much lower order than that of existing Gaussian mechanisms.} We prove that the expected accuracy loss of $f(\bm{x})$ caused  by the   R1SMG mechanism (i.e., $||{R1SMG}(f(\bm{x}))  - f(\bm{x})||_2^2$) is  lower bounded  by $C_R (\Delta_2f)^2$, where for any fixed feasible $\epsilon$, $C_R = \frac{2}{\epsilon \psi}$ has a decreasing trend as $M$ increases  and converges to $\frac{2}{\epsilon}$ as $M$ goes large (Theorem \ref{theorem:r1smg_asymptotic}).  In contrast, the classic Gaussian, analytic Gaussian, and MVG mechanisms result in  expected accuracy loss that is lower bounded  by  $C_C  (\Delta_2 f)^2$, $C_A  (\Delta_2 f)^2$, and $C_M  (\Delta_2 f)^2$, respectively, where $C_C =  \frac{2\ln\left(\frac{1.25}{\delta}\right)}{\epsilon^2} M$, $C_A = \frac{\left(\Phi^{-1}(\delta)\right)^2+ \epsilon}{ \epsilon^2} M$, and   $C_M =  \frac{ \frac{5}{4}H_r+\frac{1}{4}H_{r,\frac{1}{2}} }{ 2\epsilon } MN$  (for query results in $\R^{M\times N}$, i.e., in a matrix form). Therefore, by using the R1SMG mechanism, \textbf{less is more} in the sense that  noise of a much lower order of magnitude is needed compared with that of existing Gaussian mechanisms.

\noindent\textbf{(3) Higher stability of the perturbed    results (Theorem \ref{kk} and \ref{ss}).} We theoretically    show that the accuracy of the perturbed $f(\bm{x})$ obtained by the   R1SMG mechanism is more stable than the other mechanisms. By ``stable'', we mean that it is less likely for the R1SMG mechanism to generate noise with large magnitude that overwhelms  $f(\bm{x})$. 

\noindent\textbf{(4) Applications.} We perform three case studies, including differentially private 2D histogram query, principal component analysis, and  deep learning, to   validate that   the   R1SMG mechanism can achieve better data utility in various applications by introducing less   noise, i.e., less is more.

Table \ref{table:all-compare} summarizes and compares our proposed R1SMG mechanism with the classic Gaussian mechanism and its variants from the perspective of   expected accuracy loss parameterized by the noise parameters and utility stability.

\begin{table*}[htb]
\centering
\hspace{-30mm}
\resizebox{0.8\textwidth}{!}{\begin{minipage}{\textwidth}
\begin{tabular}{|p{6cm}|p{11.5cm}|p{3cm}|} 
\hline
Mechanisms ($\mathcal{M}$)&  Expected accuracy loss $\mathbb{E}_{\mathcal{M}}[\mathcal{L}]$ & Stability (Section \ref{utility})\\
\hline
The classic Gaussian mechanism \cite{dwork2006our} (noise decided by   variance $\sigma^2$ and dimension $f(\bm{x})\in\R^M$)&   $\mathbb{E}_{classic}[\mathcal{L}] = \Tr[\sigma^2\mathbf{I}_{M\times M}] \geq C_C  (\Delta_2 f)^2$, \newline where $C_C =  \frac{2\ln\left(\frac{1.25}{\delta}\right)}{\epsilon^2} M$.  (Section \ref{sec:curse-vector})&  \textbf{unstable}, i.e., likely to generate noise with large magnitude\\
\hline
The analytic Gaussian mechanism  \cite{balle2018improving} (noise decided by   variance $\sigma_A^2$ and dimension $f(\bm{x})\in\R^M$) &  $\mathbb{E}_{analytic}[\mathcal{L}] = \Tr[\sigma_A^2\mathbf{I}_{M\times M}] \geq C_A  (\Delta_2 f)^2$, \newline  where $C_A = \frac{\left(\Phi^{-1}(\delta)\right)^2+ \epsilon}{ \epsilon^2} M$. (Section \ref{sec:curse-vector})   & \textbf{unstable}\\
\hline
The MVG mechanism \cite{Chanyaswad2018} (noise decided by covariance matrices $\S$, $\P$, and dimension $f(\bm{x})\in\R^{M\times N}$)&  $\mathbb{E}_{MVG}[\mathcal{L}] = \Tr[\S\otimes \P] \geq C_M  (\Delta_2 f)^2$, \newline  where $C_M =  \frac{ \frac{5}{4}H_r+\frac{1}{4}H_{r,\frac{1}{2}} }{ 2\epsilon }  MN$.  (Section \ref{sec:matrix-mvg-curse}) & \textbf{unstable}\\
 \hline \hline
The R1SMG mechanism (noise decided by a random rank-1 singular covariance matrix with only one nonzero eigenvalue $\sigma_*$)& $\mathbb{E}_{R1SMG}[\L] = \Tr[\bm{\Pi}] = \sigma_* \geq C_R(\Delta_2f)^2$, where \newline   $C_R = \frac{2}{\epsilon \psi}$, $\psi = \Big(\frac{\delta\Gamma(\frac{M-1}{2})}{\sqrt{\pi}\Gamma(\frac{M}{2})}\Big)^{\frac{2}{M-2}}$,  and $C_R= \Theta(\frac{\epsilon C_C}{M}) = \Theta(\frac{C_A}{M}) = \Theta(\frac{C_M}{MN})$.  \newline $C_R$ is on a lower order of magnitude by at least $M$ or $MN$ compared with $C_C$, $C_A$, \newline and $C_M$, and converges to  $\frac{2}{\epsilon}$ as $M$ goes large for any fixed feasible $\epsilon$. (Section \ref{sec:discussion_r1smg}) & \textbf{stable}, i.e., unlikely to generate noise
with large magnitude \\
\hline
\end{tabular}
\end{minipage}}
\vspace{-1mm}
\caption{Comparisons between the R1SMG mechanism and the classic Gaussian mechanism and its variants.} 
\label{table:all-compare}
\vspace{-3mm}
\end{table*}

\textbf{Roadmap.} We   review the preliminaries for this study in Section \ref{sec:pre}. Next, we identify  the curse of full-rank covariance matrices in various types  of Gaussian mechanisms
 in Section \ref{sec:curse}. 
 In Section \ref{sec:evidence}, we revisit Dwork and Roth's  proof, 
 and identify a previously ignored clue that motivates our work.  After that, we formally develop the R1SMG mechanism in Section \ref{sec:redfrog}. In Section \ref{utility}, we analyze the  utility stability achieved by our mechanism and theoretically compare it with other mechanisms. 
  Section \ref{sec:exp}  presents the case studies. 
 Section \ref{rw} discusses the  related works. 
Finally, Section \ref{sec:con} concludes the paper.

%% file: sections/preliminiaries.tex
\section{Preliminaries}
\label{sec:pre}

In this section,   we   review $(\epsilon,\delta)$-DP, and     revisit the definitions and privacy guarantees of  the classic Gaussian mechanism and its variants.  Table \ref{table:notations} lists the frequently used notations.

\begin{table}[htb]
\begin{center}
\resizebox{0.85\textwidth}{!}{\begin{minipage}{\textwidth}
\begin{tabular}{|c |c | } 
\hline
  $\epsilon$ and $\delta$ & privacy budget of the mechanisms  \\ 
  \hline
 $\Delta_2f$ & $l_2$ sensitivity of a query function $f$ \\\hline
 $\sigma^2$ & variance of  noise in the classic Gaussian mechanism\\\hline
  $\sigma_A^2$ & variance of  noise in the  analytic Gaussian mechanism\\\hline
 $\S$ and $\P$ & row-, column-wise covariance matrices in MVG \\\hline
 $\bm{\Pi}$ & singular covariance matrix in   the     R1SMG mechanism\\ \hline
 $\sigma_*$ & the only nonzero eigenvalue of $\bm{\Pi}$\\\hline
 $\mathbb{PSD}$ & set of positive semi-definitive matrices\\\hline
  $\mathbb{PD}$ & set of positive  definitive matrices\\\hline
$\mathcal{U}(\mathbb{V}_{1,M})$ & {\color{black}uniform distribution on Stiefel manifold $\mathbb{V}_{1,M}$}\\ 
&(i.e., the unit sphere $\mathbb{S}^{M-1}$)\\
\hline
$\chi^2(M)$& Chi-squared  distribution with $M$ degrees of freedom\\\hline
 $\L$ &   non-deterministic accuracy loss\\
 \hline
\end{tabular}
\end{minipage}}
\caption{Frequently used notations in the paper.} 
\label{table:notations}
\end{center}
\vspace{-4mm}
\end{table}

Following the convention of \cite{dwork2006our,dwork2014algorithmic},  a database $\bm{x}$ is represented  by its histogram in a universe $\mathcal{X}$:  $\bm{x}\in\mathbb{N}^{|\mathcal{X}|}$, where each entry $\bm{x}_i$ is the number of elements in the database $\bm{x}$ of $type\ i\in\mathcal{X}$  \cite[p. 17]{dwork2014algorithmic}. $\bm{x}\sim\bm{x}'$ denotes a pair of neighboring  databases that differ by   one data record, and $\Delta_2f$ is the $l_2$ sensitivity of a query function $f$, i.e., $\Delta_2f = \sup_{\bm{x}\sim\bm{x}'}||f(\bm{x})-f(\bm{x}')||_2$.

\begin{definition} \textbf{$(\epsilon, \delta)$-DP \cite{dwork2014algorithmic}.}
\label{def_dp}
A randomized mechanism $\mathcal{M}$ 
satisfies $(\epsilon,\delta)$-DP if for any two neighboring datasets, $\bm{x},\bm{x}'\in\mathbb{N}^{|\mathcal{X}|}$, 
$\epsilon>0$ and $0< \delta <1$, the following holds
\begin{equation*}
\Pr[\mathcal{M}(D_1)\in S] \leq e^{\epsilon}\Pr[\mathcal{M}(D_2)\in S]+\delta.
\end{equation*}
\end{definition}
Dwork et al.  \cite{dwork2016concentrated}  also give an alternative definition of $(\epsilon, \delta)$-DP, i.e., for an outcome $s$,   $\ln\left(\frac{\Pr[\mathcal{M}(\bm{x})=s]}{\Pr[\mathcal{M}(\bm{x}')=s]}\right)\leq \epsilon$ holds with all but $\delta$ probability. The log-ratio of the probability densities is called the privacy loss random variable~\cite{balle2018improving,dwork2016concentrated} (Definition \ref{def_plrm}). 

\begin{definition} \textbf{The classic Gaussian Mechanism.} 
Let $f:\mathbb{N}^{|\mathcal{X}|}\rightarrow \mathbb{R}^M$ be an arbitrary $M$-dimensional function. 
The Gaussian mechanism with parameter $\sigma$ adds random Gaussian noise following $\mathcal{N}(0,\sigma^2)$ to \underline{each of the $M$ components of the output}.
\label{def:def_GM}
\end{definition}

\begin{theorem}\textbf{DP Guarantee of the  classic Gaussian mechanism \cite[p. 261]{dwork2014algorithmic}}
Let $\epsilon\in(0,1)$ be arbitrary. If $\sigma\geq  {c\Delta_2f}/{\epsilon}$, where $c^2 > 2\ln(\frac{1.25}{\delta})$, then the classic Gaussian Mechanism achieves ($\epsilon,\delta$)-DP.
\label{Gaussian_privacyguarantee}
\end{theorem}

The analytic Gaussian mechanism  \cite{balle2018improving} follows   the same procedure of the classic Gaussian mechanism (discussed in Definition \ref{def:def_GM}).  It improves the classic Gaussian mechanism by directly calibrating the variance ($\sigma_A^2$) of the Gaussian noise via   solving (\ref{eq:analytic-equation}) using  binary search. The privacy guarantee of this mechanism is stated as follows.

\begin{theorem} \textbf{DP Guarantee of the analytic Gaussian mechanism \cite{balle2018improving}.}\label{def_AG}
    The analytic Gaussian mechanism  achieves $(\epsilon,\delta)$-DP guarantee on the query result $f(\bm{x})$ if 
\begin{equation}\label{eq:analytic-equation}
    \Phi\Big(\frac{\Delta_2f}{2\sigma_A}-\frac{\epsilon \sigma_A}{\Delta_2f}\Big)-e^\epsilon \Phi\Big(-\frac{\Delta_2f}{2\sigma_A}-\frac{\epsilon \sigma_A}{\Delta_2f}\Big)\leq \delta,
\end{equation}
where $\Phi(t) = \frac{1}{\sqrt{2\pi}}\int_{-\infty}^te^{-y^2/2}dy$ is the cumulative distribution function   of the standard univariate Gaussian distribution.
\end{theorem}

The MVG mechanism \cite{Chanyaswad2018} 
 is an extension of the Gaussian mechanism and focuses on  matrix-valued queries.

\begin{definition} \textbf{The MVG mechanism   \cite{Chanyaswad2018}.}
Given a matrix-valued query function $f(\bm{x})\in\mathbb{R}^{M\times N}$,  the MVG mechanism   is defined as $\mathcal{M}_{MVG}(f(\bm{x}))  = f(\bm{x})+\mathbf{N}$, 
where $\mathbf{N}\sim\mathcal{N}_{M,N}(\mathbf{0},\S,\P)$ represents the matrix-variate Gaussian distribution 
with   PDF 
\begin{equation}\label{eq:matrix_gaussian_distribution}
f({\X}) = \frac{\exp\Big(-\frac{1}{2}\Tr[\P^{-1}(\X-\mathbf{M})^T\S^{-1}(\X-\mathbf{M})]\Big)}{(2\pi)^{MN/2}{\rm{det}}(\P)^{M/2}{\rm{det}}(\S)^{N/2}}.
\end{equation}
$\mathbf{M}\in\mathcal{R}^{M\times N}$ is the mean. $\S\in\mathbb{PD}^{M\times M}$ and $\P\in\mathbb{PD}^{N\times N}$ are the row-wise and column-wise full-rank  covariance matrix,  which are also symmetric.
\label{def_MVG}
\end{definition}

\begin{theorem}\textbf{DP Guarantee of  MVG    \cite{Chanyaswad2018}.}   
$\sigma(\S^{-1})$ 
and $\sigma(\P^{-1})$ 
are the vectors consisting  of singular values of $\S^{-1}$ and $\P^{-1}$.   The  MVG mechanism   achieves  $(\epsilon,\delta)$-DP if $||\sigma(\S^{-1})||_2||\sigma(\P^{-1})||_2\leq  {\big(-\beta+\sqrt{\beta^2+8\alpha\epsilon}\big)^2}\big/{4\alpha^2}$, where  $\alpha  = 2H_r\gamma \Delta_2f+(H_r+H_{r,\frac{1}{2}})\gamma^2$, $\beta = 2(MN)^{1/4}H_r\Delta_2f\tau(\delta)$, $H_r$ is the $r$th harmonic number, $H_{r,\frac{1}{2}}$ is the $r$th harmonic number of order   $\frac{1}{2}$, $\gamma = \sup_{\bm{x}}||f(\bm{x})||_F$, and $\tau(\delta) = 2\sqrt{-MN\ln \delta}-2\ln \delta+MN$.
\label{mvg_privacyguarantee}
\end{theorem}

In our study, we also identify an approach to improve the MVG mechanism. Please see   Appendix \ref{app:mvg-improve} for   details.

%% file: sections/identify-curse.tex
\section{Identifying The Curse}\label{sec:curse}

We identify the curse of full-rank covariance matrices by showing that the expected value of the accuracy loss  (Definition \ref{def:utility_cost}) introduced by the classic Gaussian mechanism  and its variants all are on the order of $M(\Delta_2 f)^2$ for query results in $\R^M$  and $MN(\Delta_2 f)^2$ for query results in $\R^{M\times N}$.

\subsection{Query Result  being a Vector}\label{sec:curse-vector}

When $f(\bm{x})\in\R^M$, the classic Gaussian mechanism    perturbs each element of $f(\bm{x})$ using i.i.d. Gaussian noise drawn from $\mathcal{N}(0,\sigma^2)$ (cf. Definition \ref{def:def_GM}). Thus, we   derive $\mathbb{E}[\L]$ under the framework of the classic Gaussian mechanism as 
\begin{small}
\begin{equation*}
\begin{aligned}
   & \mathbb{E}_{\rm{classic}}[\L] =  \mathbb{E}_{\mathbf{n}_{i}\sim \mathcal{N}(0,\sigma^2)}  \Big[{\sum}_{i=1}^{M}\mathbf{n}_{i}^2\Big]= \mathbb{E}_{z_{i}\sim \mathcal{N}(0,1)}  \Big[\sigma^2{\sum}_{i=1}^{M} z_{i}^2\Big]\\
&   \stackrel{(a)} =  \sigma^2\mathbb{E}_{U\sim \chi^2(M)}[U] 
  \stackrel{(b)}   = \sigma^2 M  =  \Tr[\sigma^2\mathbf{I}_{M\times M}] = \Tr[Cov(\mathbf{n})],  
    \end{aligned}
    \label{eq:gaussian_df}
\end{equation*} 
\end{small}\ignorespacesafterend
where $(a)$ and $(b)$ hold because  the squared sum of $M$ independent standard Gaussian variables follows the chi-squared distribution with $M$ degrees of freedom, i.e., $\chi^2(M)$, whose  expected value  is $M$.

Based on Theorem \ref{Gaussian_privacyguarantee}, since the classic Gaussian mechanism has $\sigma \geq c\Delta_2 f/\epsilon$, we get $\mathbb{E}_{classic}[\L] \geq C_C (\Delta_2f)^2$, where $C_C =  \frac{2\ln (1.25/\delta)}{\epsilon^2}M$.

The same   analysis also applies to the analytic Gaussian mechanism (discussed in Section \ref{sec:pre}), because it still perturbs each element of $f(\bm{x})$ using i.i.d. Gaussian noise   with variance $\sigma_A^2$. Thus, we can have $\mathbb{E}_{Analytic}[\L] = \Tr[\sigma^2_A\mathbf{I}_{M\times M}] = \sigma^2_AM$.

Since there is no closed-form solution of $\sigma_A$  (\cite{balle2018improving} solves for $\sigma_A$ using a binary search scheme), we derive the lower bound of $\mathbb{E}_{analytic}[\L]$ by investigating the sufficient condition for (\ref{eq:analytic-equation}). In what follows $A\Leftarrow B$ means $B$ is the sufficient condition for $A$, and $A\Leftrightarrow B$ means $A$ and $B$ are equivalent. Then, we have
\begin{small}
\begin{equation*}
    \begin{aligned}
        (\ref{eq:analytic-equation})& \stackrel{\mathrm{suff.\ cond.}}\Leftarrow \Phi\Big(\frac{\Delta_2f}{2\sigma_A}-\frac{\epsilon \sigma_A}{\Delta_2f}\Big)\leq \delta
        \Leftrightarrow \frac{\Delta_2f}{2\sigma_A}-\frac{\epsilon \sigma_A}{\Delta_2f} \leq \Phi^{-1}(\delta)\\
        & \Leftrightarrow 2\epsilon\sigma_A^2 + 2\Delta_2f\Phi^{-1}(\delta) \sigma_A - (\Delta_2f)^2\geq 0\\
        & \Leftrightarrow \resizebox{0.75\linewidth}{!}{$\sigma_A\geq \frac{-2\Delta_2f\Phi^{-1}(\delta)+\sqrt{4(\Delta_2f)^2(\Phi^{-1}(\delta))^2+8\epsilon (\Delta_2f)^2}}{4\epsilon}$} \\
        & \resizebox{0.9\linewidth}{!}{$\stackrel{\mathrm{suff.\ cond.}}\Leftarrow \sigma_A^2 \geq \frac{\left(-2\Delta_2f\Phi^{-1}(\delta)+\sqrt{4(\Delta_2f)^2(\Phi^{-1}(\delta))^2+8\epsilon (\Delta_2f)^2}\right)^2}{16\epsilon^2}$} \\
        & \resizebox{0.9\linewidth}{!}{$\stackrel{\mathrm{suff.\ cond.}}\Leftarrow \sigma_A^2 \geq \frac{2\left(4(\Delta_2f)^2(\Phi^{-1}(\delta))^2+4(\Delta_2f)^2(\Phi^{-1}(\delta))^2+8\epsilon (\Delta_2f)^2\right)}{16\epsilon^2}$}\\
       & \Leftrightarrow  \sigma_A^2 \geq \frac{\left(\Phi^{-1}(\delta)\right)^2+\epsilon}{\epsilon^2} (\Delta_2f)^2.
    \end{aligned}
\end{equation*}
\end{small}\ignorespacesafterend
Therefore, we can get $\mathbb{E}_{analytic}[\L]    \geq  C_A (\Delta_2f)^2$, where $C_A = \frac{\left(\Phi^{-1}(\delta)\right)^2+ \epsilon}{ \epsilon^2} M$.

\subsection{Query Result  being  a Matrix}\label{sec:matrix-mvg-curse}
When $f(\bm{x})\in\R^{M\times N}$, instead of adding i.i.d. Gaussian noise,  the MVG mechanism by  Chanyaswad et al.    \cite{Chanyaswad2018}   
perturbs the  matrix-valued result by adding a matrix noise  
attributed to the  matrix-variate Gaussian distribution  shown in (\ref{eq:matrix_gaussian_distribution}).

It is well-known that the vectorization of the matrix-variate Gaussian noise can equivalently   be sampled from    a multivariate Gaussian distribution with a new full-rank   covariance matrix, i.e., $\bm{\Xi}=\S\otimes\P\in\mathbb{PD}^{MN\times MN}$, which is also symmetric \cite{gupta2018matrix} ($\otimes$ is the Kronecker product).  In general,  $\bm{\Xi}$ has nonzero off-diagonal entries, which means the elements in the additive noise are not mutually independent. Thus, we need another tool to analyze its expected accuracy loss. 

First, we observe that  $\L =||\mathbf{n}||_2^2= \mathbf{n}^T\mathbf{n}$ is a quadratic form in random variable  defined as follows.
\begin{definition}\textbf{Quadratic Form in Random Variable \cite{provost1992quadratic}.} 
Denote by $\x$ a random vector with mean $\mathbf{\mu}$ and covariance matrix $\bm{\Xi}$. The quadratic form in random variable $\x$ associated with a symmetric matrix $\A$ is defined as $Q(\x) = \x^T\A\x$.
\end{definition}

Then, $\L$ introduced by  the MVG mechanism  is a quadratic form in multivariate Gaussian random variable attributed to $\mathcal{N}(\mathbf{0},\bm{\Xi})$ with $\A$ being the identity matrix.  In the following lemma, we recall the $t$-th moment of quadratic form in Gaussian random variable (which will also be used in our future analyses on the kurtosis and skewness of $\L$ in Section \ref{utility}).

\begin{lemma}\cite[p. 53]{provost1992quadratic} The $t$-th moment of a quadratic form in multivariate Gaussian random variable, i.e, $\x\sim\mathcal{N}(\mathbf{\mu},\bm{\Xi})$ with  $\bm{\Xi}$ being full-rank, is  given by 
\begin{small}
    \begin{equation*}
        \begin{aligned}
            \mathbb{E}[Q(\x)^t] = \Bigg\{&\sum_{t_1=0}^{t-1}{t-1\choose t_1}g^{(t-1-t_1)}\times \sum_{t_2=0}^{t_1-1}{t_1-1\choose t_2}g^{(t_1-1-t_2)}\\
        &\times\ldots\times\sum_{t_2=0}^{1}{1-1\choose t_2}g^{(0)}\Bigg\},\quad \text{where}
        \end{aligned}
    \end{equation*}
\end{small}\ignorespacesafterend
$g^{(k)} = 2^kk!(\Tr(\A\bm{\Xi})^{k+1}+(k+1)\mathbf{\mu}^T(\A\bm{\Xi})^k\A\mathbf{\mu})$,   $k\in[0,t-1]$.
\label{m}
\end{lemma}
As a result, the instantiation of Proposition  \ref{prop:curse} under the MVG mechanism  is 
\begin{small}
     \begin{equation}
\begin{aligned}
    \mathbb{E}_{MVG}[\L]  = \quad \mathbb{E}_{\mathclap{\substack{\\\\\ \mathbf{N}\sim \mathcal{N}_{M,N}(\mathbf{0},\S,\P)}}}[||\mathbf{N}||_F^2] = \quad \mathbb{E}_{\mathclap{\substack{\\\\\ \mathbf{n}\sim \mathcal{N}(\mathbf{0},\S\otimes\P)}}}[||\mathbf{n}||_2^2] 
    \stackrel{(a)} = \Tr[\bm{\Xi}]=  \Tr[\S\otimes \P] , 
\end{aligned}
\label{eq:mvg_df}
\end{equation}
\end{small}\ignorespacesafterend
where (a) is  obtained by applying Lemma \ref{m} with $t=1$, i.e., $\mathbb{E}[Q(\x)] = g^{(0)} = \Tr[\bm{\Xi}]$.  Furthermore, we can obtain  that
\begin{small}
     \begin{equation}
\begin{aligned}
    &\mathbb{E}_{\rm{MVG}}[\L]   \stackrel{(a)} =  \Tr[\S\otimes \P]  \stackrel{(b)}= \Tr[\S]\Tr[\P] \\ 
     \stackrel{(c)}= &|| \sigma(\S)  ||_1  ||\sigma(\P)||_1
\stackrel{(d)}\geq  || \sigma(\S)  ||_2 ||\sigma(\P)||_2\\
\stackrel{(e)}\geq &  \frac{MN}{||\sigma(\S^{-1})||_2||\sigma(\P^{-1})||_2}  
\stackrel{(f)} \geq   MN \frac{ 4\alpha^2  }{(-\beta+\sqrt{\beta^2+8\alpha\epsilon})^2}\\
\geq  &   \frac{MN 4\alpha^2  }{2\beta^2+8\alpha\epsilon - 2\beta\sqrt{\beta^2}} = MN\frac{\alpha}{2\epsilon} \\
  \stackrel{(g)}= & MN\frac{2H_r\gamma \Delta_2f+(H_r+H_{r,\frac{1}{2}})\gamma^2}{2\epsilon} 
 \stackrel{(h)}\geq   \frac{\frac{5}{4}H_r+\frac{1}{4}H_{r,\frac{1}{2}}}{2\epsilon} MN(\Delta_2f)^2,
\end{aligned}
\label{eq:mvg_order}
\end{equation}
\end{small}\ignorespacesafterend
where (a) follows   (\ref{eq:mvg_df}), (b) 
is due to the property of the Kronecker product \cite{schacke2004kronecker,roger1994topics}, (c) is because  $\S, \P\in\mathbb{PD}$ and they are both symmetric matrices, which means all their eigenvalues are positive, (d) is because $||\bm{y}||_1\geq||\bm{y}||_2$ for any vector $\bm{y}$, (e) is obtained by applying the harmonic mean-geometric mean inequality (see Appendix \ref{app:HM-GM}),   (f) is due to  the privacy guarantee of  the  MVG mechanism  in Theorem \ref{mvg_privacyguarantee}, and $\alpha$ and $\beta$ are   defined therein, (g) is obtained by substituting $\alpha$, and finally (h) is because $\gamma = \sup_{\bm{x}}||f(\bm{x})||$ (Theorem \ref{mvg_privacyguarantee}), and $\Delta_2f = \sup_{\bm{x}\sim\bm{x}'}||f(\bm{x})-f(\bm{x}')|| \leq \sup_{\bm{x}}||f(\bm{x})|| + \sup_{\bm{x}'}||f(\bm{x}')|| = 2 \gamma$, which suggests $\gamma\geq \frac{1}{2}\Delta_2f$. 

As a consequence, we have $\mathbb{E}_{MVG}[\L]    \geq C_M    (\Delta_2f)^2$, where $C_M  =  \frac{\frac{5}{4}H_r+\frac{1}{4}H_{r,\frac{1}{2}}}{2\epsilon} MN.$

In Theorem \ref{thm:summary-all-gaussian},   
we summarize the expected accuracy loss
of the classic Gaussian mechanism and its variants.

\begin{theorem}\label{thm:summary-all-gaussian} The expected accuracy loss of the classic Gaussian, of the analytic Gaussian, and of the MVG mechanisms are as follows:

\noindent\resizebox{1\linewidth}{!}{$\mathbb{E}_{classic}[\L]  = \Tr[\sigma^2\mathbf{I}_{M\times M}]\geq C_C  (\Delta_2f)^2, \quad C_C =  \frac{2\ln\left(\frac{1.25}{\delta}\right)}{\epsilon^2} M,$}

\noindent\resizebox{1\linewidth}{!}{$\mathbb{E}_{analytic}[\L] = \Tr[\sigma^2_A\mathbf{I}_{M\times M}] \geq C_A   (\Delta_2f)^2,\quad C_A = \frac{\left(\Phi^{-1}(\delta)\right)^2+ \epsilon}{ \epsilon^2} M,$}

\noindent\resizebox{0.98\linewidth}{!}{$\mathbb{E}_{MVG}[\L]   = \Tr[\S\otimes \P]  \geq C_M    (\Delta_2f)^2, \quad C_M =   \frac{\frac{5}{4}H_r+\frac{1}{4}H_{r,\frac{1}{2}}}{2\epsilon}  MN.$}

\end{theorem}

Theorem \ref{thm:summary-all-gaussian} rigorously shows that all existing suffer from the identified curse of full-rank covariance matrices. Another stealthy perspective to understand the identified curse of full-rank covariance matrices is by observing that the PDFs of     (\ref{eq:matrix_gaussian_distribution}) and 
the multivariate Gaussian distribution all involve determinants of covariance matrices in the denominators (which should be nonzero). Then, theses covariance matrices should always be full-rank. As a consequence, the curse cannot be removed unless we use a new form of PDF to generate the perturbation noise.

%% file: sections/hidden_evidence.tex
\section{An Ignored Clue from Dwork et al. \cite{dwork2006our,dwork2014algorithmic,dwork2006calibrating}}\label{sec:evidence}

In this section, we recall Dwork and Roth's proof of Gaussian mechanism achieving $(\epsilon,\delta)$-DP \cite{dwork2006our,dwork2014algorithmic,dwork2006calibrating}  when  the query function returns  a $M$-dimension vector.    At the end of their proof, we identify an ignored clue which 
corroborates  that    Gaussian noise with rank-1 covariance matrix is sufficient to achieve $(\epsilon,\delta)$-DP and also motivates our work.

Dwork and Roth essentially investigate the upper bound of the \textbf{privacy loss random variable} (PLRV) 
 on a pair of neighboring database $\bm{x}$ and $\bm{x}'$ defined as follows.

\begin{definition}\label{def_plrm}\cite[p. 6]{dwork2016concentrated}
    Consider running a randomized mechanism $\mathcal{M}$ on a pair of neighboring dataset $\bm{x}$ and $\bm{x}'$. For an outcome $\mathbf{s}$, PLRV on $\mathbf{s}$ is defined as the log-ratio of the probability density when $\mathcal{M}$ is running on each dataset, i.e.,  $\mathrm{PLRV}^{(\mathbf{s})}_{(\mathcal{M}(\bm{x})||\mathcal{M}(\bm{x}'))} = \ln\left(\frac{\Pr[\mathcal{M}(\bm{x}) = \mathbf{s}]}{\Pr[\mathcal{M}(\bm{x}') = \mathbf{s}]}\right)$.
\end{definition}

The following derivations are restatements of content from  \cite[p. 261-265]{dwork2014algorithmic}. 
$f(\cdot)$ is a query function, i.e., $f:\bm{x}\in\mathbb{N}^{|\mathcal{X}|}\rightarrow \mathbb{R}^M$. We are interested in multivariate Gaussian noise that can obscure the difference $\bm{v} \triangleq f(\bm{x})-f(\bm{x}')$. To achieve $(\epsilon,\delta)$-DP, it requires that the $\mathrm{PLRV}$ associated with the classic Gaussian mechanism (denoted as $\mathcal{G}$), i.e., 
\begin{small}
\begin{equation*}\label{eq:plrv}
\mathrm{PLRV}^{(\mathbf{s})}_{(\mathcal{G}(\bm{x})||\mathcal{G}(\bm{x}'))} =   \left|\frac{1}{2\sigma^2}\left(||\bm{n}||^2-||\bm{n}+\bm{v}||^2\right)\right|,\ \   (\text{\cite[p. 265]{dwork2014algorithmic}})
\end{equation*}
\end{small}\ignorespacesafterend
 is upper bounded by $\epsilon$ with all but $\delta$ probability. In $\mathrm{PLRV}^{(\mathbf{s})}_{(\mathcal{G}(\bm{x})||\mathcal{G}(\bm{x}'))}$, $\bm{n}$ is chosen from $\mathcal{N}(0,\bm{\Sigma})$, where  $\bm{\Sigma}$ is a diagonal matrix with entries $\sigma^2$. 

Then, Dwork and Roth     bound $\mathrm{PLRV}^{(\mathbf{s})}_{(\mathcal{G}(\bm{x})||\mathcal{G}(\bm{x}'))}$ by observing that the multivariate Gaussian distribution is   a     spherically symmetric distribution \cite{ali1980characterization}. Thus, when representing the noise  $\mathbf{n}$ using any fixed orthonormal basis $\bm{b}_1, \bm{b}_2, \cdots, \bm{b}_m$, i.e, $\mathbf{n} = \sum_{i=1}^M \lambda_i \bm{b}_i$, the corresponding coefficients are also attributed to the same   Gaussian distribution, i.e., $\lambda_i\sim\mathcal{N}(0,\sigma^2),i\in[1,M]$. Furthermore, without loss of generality, Dwork and Roth assume the first component (base) $\bm{b}_1$ is parallel to $\bm{v}$ (the difference between $f(\bm{x})$ and $f(\bm{x}')$). Consequently, we have 
\begin{small}
\begin{equation}\label{eq:plrv2}
\begin{aligned}
&\mathrm{PLRV}^{(\mathbf{s})}_{(\mathcal{G}(\bm{x})||\mathcal{G}(\bm{x}'))} =  \left|\frac{1}{2\sigma^2}\left(\left|\left|\sum_{i=1}^M \lambda_i \bm{b}_i\right|\right|^2-\left|\left|\sum_{i=1}^M \lambda_i \bm{b}_i+\bm{v}\right|\right|^2\right)\right| \\& = \left|\frac{1}{2\sigma^2}\left(||\bm{v}||^2+2\lambda_1 ||\bm{v}||\right)\right| \leq \frac{1}{2\sigma^2} \Big((\Delta_2f)^2 + 2\lambda_1\Delta_2f\Big), 
\end{aligned}
\end{equation}
\end{small}\ignorespacesafterend
where the second equality holds because $\bm{b}_1$ and $\bm{v}$ are orthogonal to $\bm{b}_i, i\in[2,M]$, and the last inequality is due to the definition of $l_2$ sensitivity. 
Since now $\mathrm{PLRV}^{(\mathbf{s})}_{(\mathcal{G}(\bm{x})||\mathcal{G}(\bm{x}'))}$ in (\ref{eq:plrv2}) only involves a single Gaussian random variable, i.e., $\lambda_1$, Theorem \ref{Gaussian_privacyguarantee} can be proved by following the same procedure when the query function returns a scalar value (\cite[p. 262-264]{dwork2014algorithmic}).

\textbf{A hidden Clue comes up to the surface.}  From (\ref{eq:plrv2}), it is clear that $\mathrm{PLRV}^{(\mathbf{s})}_{(\mathcal{G}(\bm{x})||\mathcal{G}(\bm{x}'))}$ 
is only related to $\lambda_1$ and $\Delta_2f$. Since $\lambda_1$ (a Gaussian random variable) is the coefficient of the of orthonormal base $\bm{b}_1$, it clearly suggests that, after decomposing a multivariate Gaussian noise using a set of  orthonormal basis vectors, only the coefficient of one  vector   contributes to the value of $\mathrm{PLRV}^{(\mathbf{s})}_{(\mathcal{G}(\bm{x})||\mathcal{G}(\bm{x}'))}$. In other words, the coefficients of other basis vectors have no impact on the privacy loss  of the Gaussian mechanism. Consequently, the privacy guarantee achieved by $\mathbf{n} = \sum_{i=1}^M \lambda_i \bm{b}_i$ is identical to that achieved by $\lambda_1 \bm{b}_1$. This hidden clue   
motivates our idea of using multivariate Gaussian noise whose covariance matrix has rank-1  (instead of multivariate Gaussian noise with full-rank covariance matrix) to achieve $(\epsilon,\delta)$-DP.

%% file: sections/mian_tech.tex
\section{Lifting the Curse}\label{sec:redfrog}

In what follows, we  present the R1SMG mechanism,
which lifts the identified curse of full-rank covariance
matrices. In particular, we first intuitively explain the feasibility of the proposed R1SMG mechanism. Next, we introduce its noise generation process,  and    present  a sufficient condition for it to achieve  $(\epsilon,\delta)$-DP. After that, we analyze the  expected accuracy loss of the R1SMG mechanism, and discuss the privacy leakage of utilizing   vectors in the null space of the noise.

\subsection{The Intuition behind R1SMG}\label{sec:explain-property}

As introduced in Section \ref{sec:contribution}, the expected accuracy loss of R1SMG is  lower bounded by  $C_R(\Delta_2 f)^2$ where for any fixed feasible $\epsilon$, $C_R = \frac{2}{\epsilon \psi}$ has a decreasing trend as $M$ increases, and converges to $\frac{2}{\epsilon}$ as $M$ goes large (Theorem \ref{theorem:r1smg_asymptotic}).  Thus, we have the following property (see Property \ref{claim:magic_claim}) that is   missing in all existing DP mechanisms. Here, we provide an intuitive explanation of it.

\begin{property}
\label{claim:magic_claim}
 For any fixed feasible  $\epsilon$, $\delta$, and $\Delta_2f$, the magnitude of the noise required  to attain $(\epsilon,\delta)$-DP on $f(\bm{x})\in\R^M$   can have a non-increasing trend regarding $M$.  
\end{property}

Due to the widely accepted common practice of perturbing each component of $f(\bm{x})$ using i.i.d. Gaussian noise to achieve DP, it makes sense that   larger dimensional    $f(\bm{x})$ requires larger magnitude of noises. Thus, Property \ref{claim:magic_claim} is counterintuitive and 
seems to be   ``magic''. Yet, it can be intuitively explained as follows. Given a database $\bm{x}$ represented as histogram, we consider a  normalized  counting query function $f(\bm{x}) = \frac{1}{M}Q\bm{x} \in \R^M$, where $Q$ is a binary matrix.   A higher dimension of $f(\bm{x})$ means a larger number of queries applied to the same dataset $\bm{x}$. Subsequently, the chances that the query results (i.e., rows in $Q$) become dependent with each other increases as the query number increases. For instance, suppose that the $i$th and $j$th query, i.e., $f(\bm{x})_i$ and $f(\bm{x})_j$ are dependent, particularly, $f(\bm{x})_j$ can be fully determined by $f(\bm{x})_i$. Then,   the privacy leakage on $\bm{x}$ by observing $f(\bm{x})_i$,  $f(\bm{x})_j$, or $[f(\bm{x})_i\ f(\bm{x})_j]$ would be identical. Hence, intuitively, no more noise is needed to perturb  $[f(\bm{x})_i\ f(\bm{x})_j]$ than that for perturbing $f(\bm{x})_i$. It implies that as $M$ increases and the queried results become dependent, it is possible to achieve DP by perturbing $f(\bm{x})$ using noise of a non-increasing magnitude. Moreover, when $M$ becomes sufficiently large such that the query vectors (rows in $Q$) are linearly dependent, by observing $f(\bm{x})$ or the set of independent query vectors, the privacy leakage of $\bm{x}$ would be identical. Therefore, the same amount of noise  can be used to sanitize both (refer to \cite{sun2019relationship} for an analysis from the perspective of mutual information).

\subsection{R1SMG: Multivariate Gaussian Noise  with A Random Rank-1 Covariance Matrix}\label{sec:rsmgm}
In what follows, we first provide the statistical background on the singular multivariate Gaussian distribution with a given rank-$r$ covariance matrix, e.g., $\bm{\Pi}$ and $Rank(\bm{\Pi})=r$. 
\begin{definition}\textbf{Singular Multivariate Gaussian Distribution \cite{khatri1968some,cramer1999mathematical,rao1973linear,srivastava2009introduction}.} 
A  $M$-dimensional   random variable $\mathbf{x} = [x_1,x_2,\cdots,x_{M}]^T\in\R^M$ has a singular multivariate Gaussian distribution with mean $\mathbf{\mu}\in\mathbb{R}^{M}$ and a singular covariance matrix $\bm{\Pi}\in\mathbb{PSD}^{M\times M}$ with rank-$r$, i.e., $\mathbf{x}\sim\mathcal{N}(\mathbf{\mu},\bm{\Pi})$ and $Rank(\bm{\Pi})=r<M$, if its  PDF is 
\begin{equation*}
f_{\mathbf{x};Rank(\bm{\Pi})=r}=\frac{(2\pi)^{-r/2}    \exp\big(-\frac{1}{2}  (\mathbf{x}-\mathbf{\mu})^T \bm{\Pi}^{\dagger} (\mathbf{x}-\mathbf{\mu})\big)       }{\sqrt{\sigma_1(\bm{\Pi}) \cdots\sigma_r(\bm{\Pi})}},
\label{singular_pdf}
\end{equation*}
where $\sigma_i(\bm{\Pi})$ is the $i$-th nonzero eigenvalue of $\bm{\Pi}$, and $\bm{\Pi}^{\dagger}$ is the Moore-Penrose generalized inverse of $\bm{\Pi}$.  Particularly, we have  $\bm{\Pi}^{\dagger} = \mathbf{V}_r\bm{\Lambda}_r^{-1}\mathbf{V}_r^T$, where $\bm{\Lambda}_r = \mathrm{diag}(\sigma_1(\bm{\Pi}), \ldots,\sigma_r(\bm{\Pi}))\in\R^{r\times r}$, and $\mathbf{V}_r\in\R^{M\times r}$ is the matrix of eigenvectors corresponding to the $r$ nonzero  eigenvalues.
\label{def:singular_pdf}
\end{definition}

Samples of  $f_{\mathbf{x};Rank(\bm{\Pi})=r}$ (e.g., $\mathbf{n}$) can be generated by applying the disintegration theorem \cite{fremlin2000measure}. In particular, if $\mathbf{\mu}=\mathbf{0}$, by defining a linear mapping 
\begin{equation}
\mathbf{n} = \mathbf{V}_r\bm{\Lambda}_r^{1/2}\mathbf{z},\quad \mathrm{and}\quad \mathbf{z}\sim\mathcal{N}(\mathbf{0},\mathbf{I}_{r\times r}),
\label{mapping_singular}
\end{equation}
where $\bm{\Lambda}_r^{1/2} = \text{diag}(\sqrt{\sigma_1(\bm{\Pi})},\sqrt{\sigma_2(\bm{\Pi})},\ldots,\sqrt{\sigma_r(\bm{\Pi})})\in\R^{r\times r}$ and $\mathbf{V}_r\in\R^{M\times r}$ is the matrix of eigenvectors corresponding to the $r$ nonzero  eigenvalues of $\bm{\Pi}$, 
then $\mathbf{n}$ is attributed to a singular multivariate Gaussian distribution with a rank-$r$ covariance matrix.  
This can easily be verified by checking   the covariance matrix of $\mathbf{n}$; 
\begin{small}
 \begin{equation}\label{eq:emp-covariance}
 \begin{aligned}
     Cov(\mathbf{n}) = & Cov(\mathbf{V}_r\bm{\Lambda}_r^{1/2}\mathbf{z})
 = \mathbf{V}_r\bm{\Lambda}_r^{1/2} Cov(\mathbf{z})(\mathbf{V}_r\bm{\Lambda}_r^{1/2})^T\\
\stackrel{(*)}= & \mathbf{V}_r\bm{\Lambda}_r^{1/2} \mathbf{I}_{r\times r}(\mathbf{V}_r\bm{\Lambda}_r^{1/2})^T = \mathbf{V}_r \bm{\Lambda}_r \mathbf{V}_r^T  = \bm{\Pi}
 \end{aligned},
 \end{equation}
 \end{small}\ignorespacesafterend
where $(*)$ is due to  $\mathbf{z}\sim\mathcal{N}(\mathbf{0},\mathbf{I}_{r\times r})$. Then, clearly, we have $Cov(\mathbf{n})=r$.

When $r=1$, the singular covariance matrix $\bm{\Pi}$ only has 1 nonzero  eigenvalue, and we denote it as $\sigma_{*}$. Thus, (\ref{mapping_singular}) becomes
\begin{equation}
\mathbf{n} = \mathbf{v}\sqrt{\sigma_*}z,\quad \mathrm{and}\quad z\sim\mathcal{N}(0,1),\quad \mathbf{v}^T\mathbf{v} = 1.
\label{mapping_singular_rank1}
\end{equation}

\textbf{Singular multivariate Gaussian noise  with a random rank-1 covariance matrix.} It is noteworthy that in (\ref{mapping_singular_rank1}), $\mathbf{v}$ needs to be generated randomly to thwart  attacks which takes advantage  of  vectors in the null space of $\mathbf{v}$ (or $\mathbf{n}$). In Section \ref{sec:null_space_attack}, we will show that by sampling $\mathbf{v}$ uniformly at random, such attacks will succeed with \textbf{zero probability}. 

Since $\mathbf{v}$ is an   orthonormal vector of dimension $M\times 1$, one common approach to generate $\mathbf{v}$ randomly is by uniformly sampling from a specific  Stiefel manifold, i.e., $\mathbb{V}_{1,M} = \{\mathbf{v}\in\R^{M\times 1}:\mathbf{v}^T\mathbf{v} = 1\}$, which represents  the unite sphere $\mathbb{S}^{M-1}$  embedded in $\R^{M}$. In Statistics literature, the PDF of the \emph{uniform distribution on the Stiefel manifold $\mathbb{V}_{1,M}$}  is given by 
\begin{equation}\label{eq:stiefel_uniform}
    f(\mathbf{v}) = \frac{1}{\mathrm{Vol}(\mathbb{V}_{1,M})}, \quad \forall \mathbf{v} \in \mathbb{V}_{1,M},
\end{equation}
(see  \cite[p. 280]{gupta2018matrix},    \cite[p. 17, equation 8.2.2]{camano2006statistics}, as well   \cite[p. 30]{chikuse2012statistics}  which  gives the characteristic  function of  (\ref{eq:stiefel_uniform})).  In particular, the constant  $\mathrm{Vol}(\mathbb{V}_{1,M}) = \frac{2 \pi^{M /2}}{\Gamma(M/2)}$ (see  \cite[p. 19 and 26]{gupta2018matrix}) is the total surface area or volume  of $\mathbb{V}_{1,M}$ and $\Gamma(\cdot)$ is the ordinary gamma function    \cite{gupta2018matrix}\footnote{{\color{black}This  result can be easily obtained by setting $p=1$ in Theorem 1.4.9 on page 25 of  \cite{gupta2018matrix}. In particular,  $\mathrm{Vol}(\mathbb{V}_{1,M})$ plays the same role on $\mathbb{V}_{1,M}$ as the Lebesgue measure plays in   Euclidean space \cite{camano2006statistics}. For example, when $M=3$, $\mathrm{Vol}(\mathbb{V}_{1,3}) = 4\pi$ is the surface area of an unit sphere in 3D space, and when $M=2$, $\mathrm{Vol}(\mathbb{V}_{1,2}) = 2\pi$ is the circumference  of an unit 
circle  in 2D space.}}.

We use $\mathbf{v}\sim\mathcal{U}(\mathbb{V}_{1,M})$ to represent that $\mathbf{v}$ is a random variable uniformly sampled from $\mathbb{V}_{1,M}$. As a consequence, the linear mapping defined in (\ref{mapping_singular_rank1}) becomes
\begin{equation}
\begin{aligned}
\mathbf{n} = \mathbf{v}\sqrt{\sigma_*}z,\quad 
\text{where}\quad  z\sim\mathcal{N}(0,1),\quad \mathbf{v}\sim\mathcal{U}(\mathbb{V}_{1,M}).
\end{aligned}
\label{mapping2}
\end{equation}

Now, we introduce the    R$1$SMG  mechanism and provide a sufficient condition for it to achieve differential privacy.
\begin{definition}\textbf{The Rank-1 Singular Multivariate Gaussian (R1SMG) Mechanism.} For an arbitrary $M$-dimensional query function,  $f(\bm{x})\in\R^{M}$,   the R1SMG mechanism is defined as 
    \begin{equation*}
    \mathcal{M}_{R1SMG}\big(f(\bm{x})\big)  = f(\bm{x})+ \mathbf{n},
\end{equation*}
where   $\mathbf{n}$ (generated via (\ref{mapping2})) is the noise attributed to a singular multivariate Gaussian distribution with a random rank-$1$ covariance matrix.
\label{def:rsmgm}
\end{definition}
If the query result is a matrix or tensor, we can first generate the noise as a vector  and then resize it into the desired format.  Please refer to the case studies    in Section \ref{sec:exp} for details.

We introduce an important lemma below that will be used for proving Theorem \ref{thm_multivariate}.

\begin{lemma}\textbf{Distribution  of Angle between Random Points on Unit Sphere (\cite[p. 1845 and 1860]{cai2013distributions}).} Let $h$ and $g$ be two randomly selected  points on   unit sphere $\mathbb{S}^{P-1}$ (embedded in $\R^P$), where $P>2$. Let $\mathbf{h}$ (resp. $\mathbf{g}$) be the vector connecting the center of the sphere and $h$ (resp. $g$). $\theta$ denotes the angle ($\theta\in[0, \pi]$) between $\mathbf{h}$ and $\mathbf{g}$. Then, we have 
\begin{equation}
    \Pr\left[\left|\theta-\frac{\pi}{2}\right|\geq \theta_0\right] \leq \sqrt{\pi}\frac{\Gamma(\frac{P}{2})}{\Gamma(\frac{P-1}{2})} cos(\theta_0)^{P-2},
\end{equation}
where $\theta_0$ is a given radian and $0\leq \theta_0\leq\frac{\pi}{2}$. 
\label{lemma:angle}
\end{lemma}

\begin{theorem}\label{thm_multivariate}
The R1SMG mechanism achieves $(\epsilon,\delta)$-DP when $M>2$,  if $\sigma_{*} \geq  \frac{2(\Delta_2f)^2}{\epsilon \psi}$ 
where   $\psi = \Big(\frac{\delta\Gamma(\frac{M-1}{2})}{\sqrt{\pi}\Gamma(\frac{M}{2})}\Big)^{\frac{2}{M-2}}$, 
and   $\Gamma(\cdot)$ is the Gamma  function.
\end{theorem}
\begin{proof} Inspired by Dwork and Roth's work~\cite{dwork2014algorithmic} (cf. Section \ref{sec:evidence}), we  investigate the PLRV associated with the   R1SMG mechanism. 
Assume that $\mathbf{n}$ and $\mathbf{n}'$ are the random noise used to perturb $f(\bm{x})$ and $f(\bm{x}')$, respectively.  Then, the R1SMG mechanism achieves $(\epsilon,\delta)$-DP if 
$\mathrm{PLRV}^{(\mathbf{s})}_{(R1SMG(\bm{x})||R1SMG(\bm{x}'))}=\ln\left(\frac{\Pr[f(\bm{x})+\mathbf{n} = \mathbf{s}\in\mathcal{S}]}{\Pr[f(\bm{x}')+\mathbf{n}'= \mathbf{s}\in\mathcal{S}]}\right)\leq \epsilon$  with all but $\delta$ probability (called the failing probability)~\cite{dwork2014algorithmic}, where $\mathcal{S}$ denotes all possible outcome of R1SMG. 

According to (\ref{mapping2}), without loss of generality, let 
$\mathbf{n} = \sqrt{\sigma_*} z_1\mathbf{h} $ and $\mathbf{n}' =  \sqrt{\sigma_*} z_2 \mathbf{g}$,
where   $z_1, z_2\sim\mathcal{N}(0,1)$ and $\mathbf{h},\mathbf{g}\sim\mathcal{U}(\mathbb{V}_{1,M})$. Let $\theta\in[0,\pi]$ be the angle between $\mathbf{h}$ and $\mathbf{g}$. Then, we can establish the failing probability (which should be at most $\delta$) by first constructing an event which is a subspace of the   universe $\Omega$, i.e., 
\begin{equation*}\label{eq:the_set}
\begin{aligned}
&\resizebox{1\linewidth}{!}{$\mathcal{A} = \left\{z_1, z_2, \mathbf{h}, \mathbf{g}: z_1, z_2\sim \mathcal{N}(0,1),  \mathbf{h}, \mathbf{g}\sim \mathcal{U}(\mathbb{V}_{1,M}),  \left|\theta-\frac{\pi}{2}\right| < \theta_0\right\}$}\\
&\subset  \Omega = \{z_1, z_2, \mathbf{h}, \mathbf{g}: z_1, z_2\sim \mathcal{N}(0,1),  \mathbf{h}, \mathbf{g}\sim \mathcal{U}(\mathbb{V}_{1,M})\}
\end{aligned},
\end{equation*}
where $\theta_0 = arccos\big(\left({\delta\Gamma(\frac{M-1}{2})}\big/{\sqrt{\pi}\Gamma(\frac{M}{2})}\right)^{1/(M-2)}\big)$, and $arccos(\cdot)$ is   the inverse of the cosine function. 
We denote the complementary   of $\mathcal{A}$ as $\mathcal{A}^c = \Omega\setminus\mathcal{A}$, i.e., 
\begin{equation*}\label{eq:the_set}
    \resizebox{1\linewidth}{!}{$\mathcal{A}^c = \left\{z_1, z_2, \mathbf{h}, \mathbf{g}: z_1, z_2\sim \mathcal{N}(0,1),  \mathbf{h}, \mathbf{g}\sim \mathcal{U}(\mathbb{V}_{1,M}),  \left|\theta-\frac{\pi}{2}\right| \geq  \theta_0\right\}.$}
\end{equation*}
Since $z_1$ and $z_2$ are independent of $\mathbf{h}$ and $\mathbf{g}$, by applying   Lemma \ref{lemma:angle} and setting $P=M$ and $\sqrt{\pi}\frac{\Gamma(\frac{P}{2})}{\Gamma(\frac{P-1}{2})} cos(\theta_0)^{P-2}=\delta$, we get  $\Pr[\mathcal{A}^c] \leq  \delta$.

Since $\mathbf{n}$ (resp. $\mathbf{n}'$) is singular multivariate Gaussian with zero mean by design (see (\ref{mapping2})), we have that  $f(\bm{x})+\mathbf{n} = \mathbf{s}\in \mathcal{S}$  (resp. $f(\bm{x}')+\mathbf{n}' = \mathbf{s}\in \mathcal{S}$) is attributed to singular multivariate  Gaussian distribution, with mean $f(\bm{x})$ (resp. $f(\bm{x}')$) and covariance matrix $\mathbf{h} \sigma_* \mathbf{h}^T$ (resp. $\mathbf{g} \sigma_*  \mathbf{g}^T$) (one can verify this by setting $r=1$ in (\ref{eq:emp-covariance})). Substituting the mean and generalized inverse covariance matrix into the PDF provided in Definition \ref{def:singular_pdf} (with $r=1$),  we have a counterpart of (\ref{eq:plrv2}) as
\begin{equation}
\begin{aligned}\label{eq:needs-to-show2}
&\mathrm{PLRV}^{(\mathbf{s})}_{(R1SMG(\bm{x})||R1SMG(\bm{x}'))} = \ln\left(\frac{\Pr[f(\bm{x})+\mathbf{n} = \mathbf{s}\in\mathcal{S}]}{\Pr[f(\bm{x}')+\mathbf{n}'= \mathbf{s}\in\mathcal{S}]}\right)\\
=& \resizebox{0.9\linewidth}{!}{$\ln\left(\frac{{(2\pi)^{-\frac{1}{2}}    \exp\left(-\frac{1}{2}  \big(\mathbf{s}-f(\bm{x})\big)^T \big(\mathbf{h} \sigma_*^{-1} \mathbf{h}^T\big) \big(\mathbf{s}-f(\bm{x})\big)\right)       }\Big/{\sqrt{\sigma_*}}}{{(2\pi)^{-\frac{1}{2}}    \exp\left(-\frac{1}{2}  \big(\mathbf{s}-f(\bm{x}')\big)^T \big(\mathbf{g} \sigma_*^{-1} \mathbf{g}^T\big) \big(\mathbf{s}-f(\bm{x}')\big)\right)       }\big/{\sqrt{\sigma_*}}}\right)$}\\
= &  \ln\left( \frac{\exp\left (-\frac{1}{2\sigma_*}    \Big(  \mathbf{h}^T(\mathbf{s}-f(\bm{x}))   \Big)^2    \right )}{\exp\left (-\frac{1}{2\sigma_*}    \Big( \mathbf{g}^T(\mathbf{s}-f(\bm{x}'))   \Big)^2\right )} \right) \\
=&\frac{1}{2\sigma_*}  \Big(    \Big(  \underbrace{ \mathbf{g}^T (\mathbf{s}-f(\bm{x}')) }_{\rho_1\in \R}  \Big)^2  - \Big( \underbrace{ \mathbf{h}^T (\mathbf{s}-f(\bm{x})) }_{\rho_2\in \R}\Big)^2  \Big) \\
\leq & \frac{1}{2\sigma_*}  \Big(|\rho_1|+|\rho_2|\Big)^2,  \qquad \qquad\forall   (z_1,z_2,\mathbf{h},\mathbf{g})\in   \mathcal{A}.
\end{aligned}
\end{equation}

In the following, we derive a tight upper bound on   $|\rho_1|+|\rho_2|$. 
In particular, with $||\mathbf{g}||_2 = ||\mathbf{h}||_2=1$, we    notice that
\begin{gather*}
\begin{aligned}
 |\rho_1|+|\rho_2| 
 & \leq ||\mathbf{g}||_2||\mathbf{s}-f(\bm{x}')||_2 + ||\mathbf{h}||_2||\mathbf{s}-f(\bm{x})||_2\\
 &  =  ||\mathbf{s}-f(\bm{x}')||_2 +  ||\mathbf{s}-f(\bm{x})||_2.
\end{aligned}
\end{gather*}
Thus, we aim to   bound   $||\mathbf{s}-f(\bm{x}')||_2 +  ||\mathbf{s}-f(\bm{x})||_2$.

First, we observe that the angle between $(\mathbf{s}-f(\bm{x}))$ and  $(\mathbf{s}-f(\bm{x}'))$ is      identical to the angle between $\mathbf{h}$ and $\mathbf{g}$, because    $(\mathbf{s}-f(\bm{x}))=\mathbf{n}$ and $(\mathbf{s}-f(\bm{x}'))=\mathbf{n}'$, and $\mathbf{n}$ and $\mathbf{n}'$ are obtained by   scaling $\mathbf{h}$ and $\mathbf{g}$. 
We use $\theta$ to represent the angle between $\mathbf{h}$ and $\mathbf{g}$ (i.e.,  identically between $(\mathbf{s}-f(\bm{x}))$ and $(\mathbf{s}-f(\bm{x}'))$). Thus, from a geometric perspective, $||\mathbf{s}-f(\bm{x}')||_2 +  ||\mathbf{s}-f(\bm{x})||_2$ is the sum of the length of two edges of a triangle in $\R^{M}$, and the angle between these two edges is $\theta$. Moreover, the length of the third edge (opposite to this specific angle) is $|| (\mathbf{s}-f(\bm{x}')) - (\mathbf{s}-f(\bm{x}))||_2 
  =||f(\bm{x})- f(\bm{x}') ||_2$.
For better understanding, we visualize the above description   in Figure \ref{fig:geometric}(a), where the red point (resp. blue point) indicates $\mathbf{n}$ (resp. $\mathbf{n}'$) in $\R^{M}$, and the red dashed line (resp. blue dashed line) is the randomly sampled $\mathbf{h}$ (resp. $\mathbf{g}$) (note that we show $\mathbf{h}$ and $\mathbf{g}$ as bidirectional, because $\mathbf{n}$ and $\mathbf{n}'$ can point to the opposite direction of $\mathbf{h}$ and $\mathbf{g}$, respectively).

\begin{figure}[htb]
  \begin{center}
     \includegraphics[width=1\columnwidth]{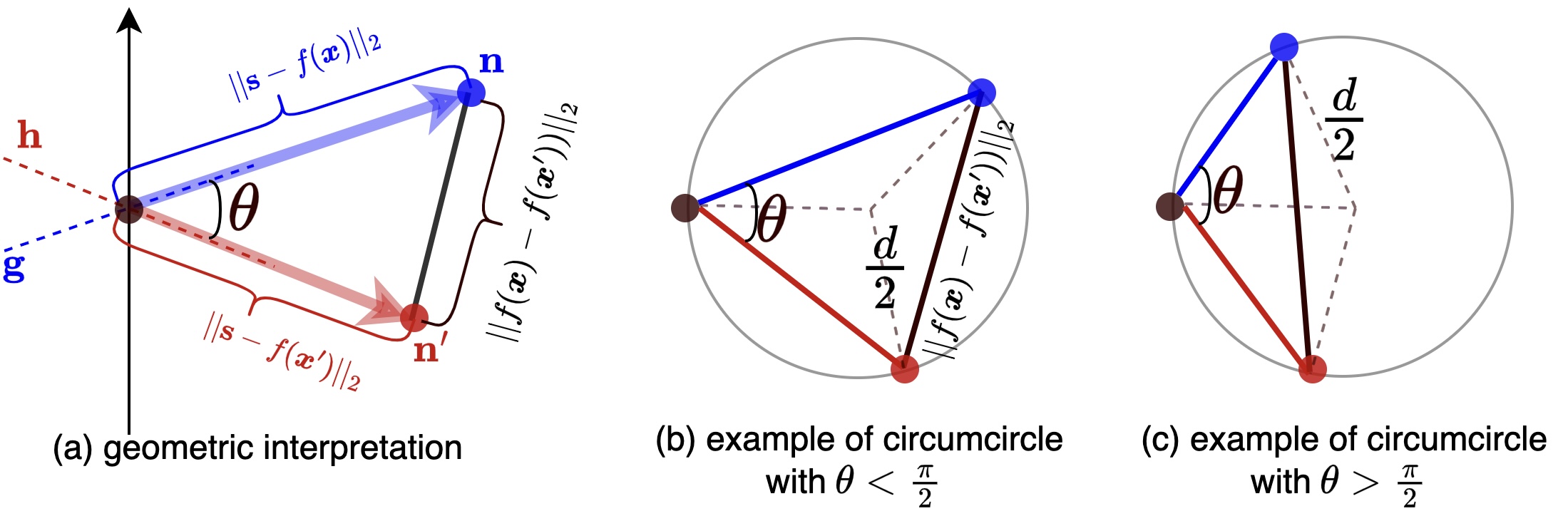}\\
       \end{center}
  \caption{(a) Geometric interpretation of $|\rho_1|+|\rho_2|$.   Circumcircles of the described triangle with (b) $\theta<\frac{\pi}{2}$ and (c) $\theta>\frac{\pi}{2}$.}
  \label{fig:geometric}
\end{figure}

It is well-known that the length of each edge of a triangle is upper bounded by  the diameter  of the circumscribed circle of the triangle. We show this fact in Figure \ref{fig:geometric}(b) and Figure (c), where $\frac{d}{2}$ represents the  radius of the circumcircle.  According to the law of sine, we have $d =   ||f(\bm{x})-f(\bm{x}')  ||_2\big/{sin(\theta)}$. Thus, 
\begin{equation}\label{eq:plrv_r1smg}
\resizebox{0.87\linewidth}{!}{$||\mathbf{s}-f(\bm{x}')||_2 +  ||\mathbf{s}-f(\bm{x})||_2 
\leq   2d = \frac{2||f(\bm{x})-f(\bm{x}')||_2}{sin(\theta)}\leq \frac{2\Delta_2 f}{sin(\theta)},$}
\end{equation}
where the last inequality is because the query sensitivity on neighboring datasets is $\Delta_2f = \sup_{\bm{x}\sim\bm{x}'}||f(\bm{x})-f(\bm{x}')||_2$. 

In (\ref{eq:needs-to-show2}), we have   $(z_1,z_2,\mathbf{h},\mathbf{g})\in    \mathcal{A}$, i.e.,  $\theta\in(\frac{\pi}{2}-\theta_0,\frac{\pi}{2}+\theta_0)$, $\theta_0\in[0,\frac{\pi}{2}]$. Thus, due to the symmetry of the sine function around $\frac{\pi}{2}$, we can obtain  $\frac{2\Delta_2f}{sin(\theta)}\leq \frac{2\Delta_2f}{sin(\frac{\pi}{2}-\theta_0)}=\frac{2\Delta_2f}{cos(\theta_0)}$. As a result, we get  $|\rho_1| +|\rho_2| \leq$ $||\mathbf{s}-f(\bm{x}')||_2 + ||\mathbf{s}-f(\bm{x})||_2 <  \frac{2\Delta_2f}{cos(\theta_0)} = 2\Delta_2f\Big/\Big(\frac{\delta\Gamma(\frac{M-1}{2})}{\sqrt{\pi}\Gamma(\frac{M}{2})}\Big)^{\frac{1}{M-2}}$.

Finally, from (\ref{eq:needs-to-show2}), we can obtain the following sufficient condition for the   R$1$SMG mechanism to achieve $(\epsilon,\delta)$-DP: 
\begin{equation*}
    \begin{aligned}
& \mathrm{PLRV}^{(\mathbf{s})}_{(R1SMG(\bm{x})||R1SMG(\bm{x}'))}\\
\leq & \frac{1}{2\sigma_*} ( |\rho_1|+|\rho_2| )^2  \leq \frac{1}{2\sigma_*}  \Big(2\Delta_2f\Big/\Big(\frac{\delta\Gamma(\frac{M-1}{2})}{\sqrt{\pi}\Gamma(\frac{M}{2})}\Big)^{\frac{1}{M-2}}\Big)^2\leq \epsilon, 
    \end{aligned}
\end{equation*}
which leads to $\sigma_*\geq \frac{2(\Delta_2f)^2}{\epsilon\psi}$. Moreover, due to the constraint of $P > 2$ in Lemma \ref{lemma:angle}, we require $M>2$ for the R1SMG mechanism. Thus, Theorem \ref{thm_multivariate} follows.
\end{proof}

\begin{remark}\label{epsilon-remark} 
In the classic Gaussian mechanism, the privacy budget is upper bounded by 1 (Theorem \ref{Gaussian_privacyguarantee}). Similarly, there is also an upper bound on $\epsilon$  for the   R1SMG mechanism. This can be   explained geometrically using Figure~\ref{fig:small_budget}. In particular, Figure~\ref{fig:small_budget}(a) shows the case where $\epsilon<1$ in the classic Gaussian mechanism. Clearly as long as $\bm{n}$ and $\bm{n}'$ can obscure the difference between $f(\bm{x})$ and $f(\bm{x})'$, i.e., $f(\bm{x})+\bm{n} = f(\bm{x}')+\bm{n}'= \mathbf{s}$, the noise components along the direction of $\bm{v} = f(\bm{x})-f(\bm{x}')$ are also sufficient to obscure the difference between $f(\bm{x})$ and $f(\bm{x})'$, i.e., $\lambda_1\bm{b}_1 - (\bm{b}_1^T\bm{n}')\bm{b}_1 = \bm{v}$. Whereas, if $\epsilon$ exceeds the upper bound, the magnitude of $\bm{n}$ and $\bm{n}'$ might be  too small to obscure $\bm{v}$, since $||\bm{n}||$ and $||\bm{n}'||$ are proportional to $\frac{1}{\epsilon}$. In other words,  when $\epsilon$ is beyond the upper bound, $\bm{n}$, $\bm{n}'$, and $\bm{v}$ cannot form a triangle, and hence we have $\lambda_1\bm{b}_1 - (\bm{b}_1^T\bm{n}')\bm{b}_1 \neq \bm{v}$ as shown in Figure~\ref{fig:small_budget}(b). This is also true for the R1SMG mechanism, i.e., when $\epsilon$ is too large, $\bm{n} = \sqrt{\sigma_*}z_1\mathbf{h}$, $\bm{n}'=\sqrt{\sigma_*}z_2\mathbf{g}$, and $\bm{v}$ cannot form a triangle as shown in Figure~\ref{fig:small_budget}(c). As a result, in order to form a triangle as shown in Figure~\ref{fig:small_budget}(d), an upper bounded $\epsilon$ is required for the R1SMG mechanism. In this work, we   let the privacy budget of the R1SMG mechanism be upper bounded by    $\frac{1}{M}$ of that of  the classic Gaussian mechanism, i.e., $\frac{1}{M}$. This is because  $\bm{v}$ is obscured using noise with $M$ degrees of freedom by the classic Gaussian mechanism, whereas it is  obscured  using noise with 1 degree of freedom by the R1SMG mechanism due to the rank-1 constraint. We will provide a tight upper bound on $\epsilon$ for the R1SMG mechanism in a separate study.
\begin{figure}[htb]
  \begin{center}
     \includegraphics[width=1\columnwidth]{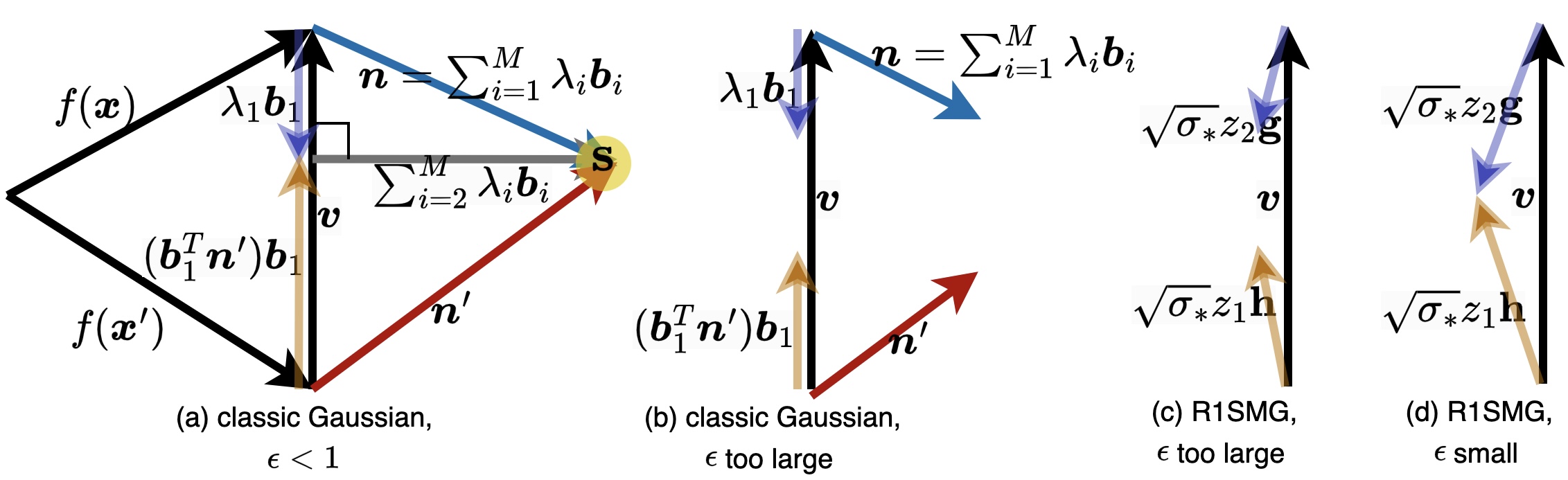}\\
       \end{center}
         \vspace{-5mm}
  \caption{Geometric interpretations on  the  constraints on $\epsilon$ in the classic Gaussian and R1SMG mechanism.}
  \label{fig:small_budget}
\end{figure}
\end{remark}

\textbf{Implementation of R1SMG.} According to the following theorem, we can obtain the desired random variable distributed on the Stiefel manifold by transforming samples drawn i.i.d. from standard Gaussian distribution.
\begin{theorem} \cite[p. 29, Theorem 2.2.1]{chikuse2012statistics}\label{thm:sample_V} Let $\mathbf{x}\in\R^{M\times r}$, whose elements are i.i.d. Gaussian random variables from  $\mathcal{N}(0,1)$. Then, $\mathbf{v} = \mathbf{x}(\mathbf{x}^T\mathbf{x})^{-1/2}$ is uniformly distributed on $\mathbb{V}_{r,M}$.
\end{theorem}
Thus, we can apply Theorem \ref{thm:sample_V} with $r=1$ to draw the desired samples  from $\mathbb{V}_{1,M}$. Then, applying (\ref{mapping2}), we can generate   an instance of   random noise for the R1SMG mechanism.

\subsection{Expected Accuracy Loss Analysis}\label{sec:discussion_r1smg}

In this section, we investigate the expected accuracy loss of the   R1SMG mechanism, i.e., $\mathbb{E}_{R1SMG}[\L]$. 
The   results are summarized in Theorem \ref{theorem:r1smg_asymptotic}.
\begin{theorem}\label{theorem:r1smg_asymptotic} For any fixed feasible $\epsilon>0, 0<\delta<1$, given a dataset $\bm{x}$ and a query result $f(\bm{x})\in\R^M$. We have 
$\mathbb{E}_{R1SMG}[\L]=||R1SMG(f(\bm{x}))  - f(\bm{x})||_2^2 = \Tr[\bm{\Pi}] = \sigma_* \geq C_R(\Delta_2f)^2$, where   $\sigma_*$ is the only nonzero  eigenvalue of $\bm{\Pi}$, $C_R = \frac{2}{\epsilon \psi}$, and $\psi = \Big(\frac{\delta\Gamma(\frac{M-1}{2})}{\sqrt{\pi}\Gamma(\frac{M}{2})}\Big)^{\frac{2}{M-2}}$.   $C_R$ has a decreasing trend as $M$ increases. When $M$ goes large, $C_R$ converges to $\frac{2}{\epsilon}$ and $\mathbb{E}_{R1SMG}[\L]$ can be as low as $\frac{2(\Delta_2 f)^2}{\epsilon}$.
\end{theorem}
\begin{proof} 
According to the noise generation process in R1SMG, i.e., (\ref{mapping2}), we have
\begin{small}
\begin{equation}
\begin{aligned}
      & \mathbb{E}_{R1SMG}[\L] =  \mathbb{E}[||\mathbf{n}||_2^2] \\
       =&\mathbb{E}\Big[ (\mathbf{v}   \sigma_*^{1/2}z)^T\mathbf{v}   \sigma_*^{1/2}z \Big] = \sigma_*\mathbb{E}[z^2] \stackrel{(a)}= \sigma_*\stackrel{(b)}\geq  2(\Delta_2f)^2/(\epsilon \psi), 
\end{aligned}
    \label{eq:freedom_r}
\end{equation}
\end{small}\ignorespacesafterend
where   (a) holds because $z^2\sim\chi^2(1)$,  
and (b) is due to  Theorem \ref{thm_multivariate}. From (\ref{eq:freedom_r}), we also have $\mathbb{E}_{R1SMG}[\L] = \sigma_* = \Tr[\bm{\Pi}]$.

To show the decreasing trend of $\mathbb{E}_{R1SMG}[\L]$ without dealing with the cumbersome notation of $\psi$, we use some intermediate results obtained in the proof of Theorem \ref{thm_multivariate}. To be more specific, we have
 $\psi = sin^2(\theta)$, where $\theta\in(\frac{\pi}{2}-\theta_0,\frac{\pi}{2}+\theta_0)$ is the angle between two instances of $M$-dimensional unit vectors, and $\theta_0 = arccos\left(\left({\delta\Gamma(\frac{M-1}{2})}\big/{\sqrt{\pi}\Gamma(\frac{M}{2})}\right)^{1/(M-2)}\right)$.  
Cai et al. \cite{cai2013distributions} has proved that when $M>2$, $\theta$ concentrates around $\frac{\pi}{2}$, and the concentration becomes stronger as the dimension $M$ grows. In particular, $\theta$ converges to $\frac{\pi}{2}$ at the rate of $\sqrt{M}$ when $M$ approaches infinity  \cite[p. 1840]{cai2013distributions}. Thus, $\frac{1}{sin^2(\theta)}$ has a decreasing trend and   converges to 1 due to the symmetry of sine function on $[0,\pi]$. Subsequently, we can have that $C_R = \frac{2}{\epsilon \psi}$ has a decreasing trend as well. In addition, when $M$ goes large, we get that $C_R$ converges to $\frac{2}{\epsilon}$ and hence $\mathbb{E}_{R1SMG}[\L]$ can be as low as $\frac{2(\Delta_2f)^2}{\epsilon}$\footnote{Note that the asymptotic property of $C_R$ is studied when $\epsilon$ is always feasible. Suppose that we are interested in the asymptotic property of $C_R$ when $M_1< M < M_2$. Then, as explained in Remark~\ref{epsilon-remark}, a feasible $\epsilon$ refers to $\epsilon < \frac{1}{M_2}$.}.
\end{proof}

\noindent\textbf{Curse Lifted.} (\ref{eq:freedom_r}) clearly shows that we have lifted the identified curse by considering noise attributed to a singular multivariate Gaussian distribution with rank-1 covariance matrix.  In particular,  the expected accuracy loss introduced by the R$1$SMG mechanism is $\Tr[\bm{\Pi}]$, which equal to the only nonzero eigenvalue $\sigma_*$ but  is \textbf{not} the summation of $M$ positive eigenvalues any more like in  the existing Gaussian mechanisms.

From Theorem \ref{thm:summary-all-gaussian} and Theorem \ref{theorem:r1smg_asymptotic}, we   arrive at   Corollary \ref{corollary:psi}.

\begin{corollary}\label{corollary:psi} 
The R1SMG mechanism leads to expected accuracy loss  on a lower order of magnitude by at least $M$ or $MN$ compared with the classic Gaussian, analytic Gaussian, and MVG mechanisms. In particular, we have $C_R= \Theta(\frac{\epsilon C_C}{M}) = \Theta(\frac{C_A}{M}) = \Theta(\frac{C_M}{MN})$. 
\end{corollary}
\begin{proof}
According to   Theorem \ref{thm_multivariate}, we have $\frac{1}{\psi}  =  \Big(\frac{\sqrt{\pi}\Gamma(\frac{M}{2})}{\delta\Gamma(\frac{M-1}{2})}\Big)^{\frac{2}{M-2}}  =   \exp\Big(  \frac{2}{M-2} \ln\Big(\frac{\sqrt{\pi}\Gamma(\frac{M}{2})}{\delta\Gamma(\frac{M-1}{2})}\Big)\Big).$ 
By connecting  $\frac{\Gamma(\frac{M}{2})}{\Gamma(\frac{M-1}{2})}$ to the normalization constant of the   Beta distribution $Beta(\alpha,\beta)$ with $\alpha = \frac{M-1}{2}$ and $\beta = \frac{1}{2}$, it gives $\Theta\Big(\frac{\Gamma(\frac{M}{2})}{\Gamma(\frac{M-1}{2})}\Big) = \Theta\left(\frac{\sqrt{M-2}}{\sqrt{2}}\right)$ \cite{cai2013distributions}. As a result, we can obtain $C_R = \Theta\left(\frac{2}{\epsilon}\exp\left(\frac{2}{M-2}\ln \frac{\sqrt{M-2}}{\sqrt{2}}\right)\right)$. Compare $C_R$ with the results in  Theorem \ref{thm:summary-all-gaussian}, we get that  $C_R= \Theta(\frac{\epsilon C_C}{M}) = \Theta(\frac{C_A}{M}) = \Theta(\frac{C_M}{MN})$.
\end{proof}

Moreover, in Appendix \ref{sec:versus_M}, we plot the empirical $\mathbb{E}[\L]$ versus $M$ achieved by the R1SMG mechanism and by other mechanisms, respectively, and clearly show that $\mathbb{E}_{R1SMG}[\L]$ has a decreasing trend as $M$ increases and eventually converges, whereas all  the other mechanisms   result in expected accuracy loss increasing with $M$ (cf. Figure \ref{fig:scale_study}).

\subsection{Discussion on Privacy Leakage of Utilizing   Vector in the Null Space of $\mathbf{v}$}\label{sec:null_space_attack}

The R$1$SMG mechanism does not span the entire space of $\R^{M}$. Therefore, it may raise  concerns about  privacy leakage if one   takes advantage of  the       vector  in the null space of $\mathbf{v}$. In other words, let $\mathbf{s}$ be the output of R$1$SMG and $\mathbf{u}\in Null(\mathbf{v}) = \{\mathbf{u}|<\mathbf{u},\mathbf{v}>=0\}$. Since  $\mathbf{n}\in Span(\mathbf{v})$, one   can have  $<\mathbf{u},\mathbf{s}> = <\mathbf{u},f(\bm{x})+\mathbf{n}> = <\mathbf{u},f(\bm{x})>$, 
which implies potential privacy leakage of     queries on $f(\bm{x})$.

However, we   highlight that $\mathbf{v}$ is randomly sampled from $\mathcal{U}(\mathbb{V}_{1,M})$, which is an intermediate noise to produce the final noise of the R$1$SMG mechanism and will not be made public. 
Note that this is not  against the core idea of DP, i.e., ``privacy without obscurity", because the entire noise generation process, i.e, (\ref{mapping2}), including  the choice of  distribution parameters are transparent.  Thus, the probability that a constructed $\mathbf{u}$ that is in the null space of a randomly sampled $\mathbf{v}$ is zero, i.e., 
$\Pr[\mathbf{u}\in Null(\mathbf{v})] = \frac{\mu(Null(\mathbf{v}))}{\mu(\mathbb{R}^{M})} = 0$, where $\mu(\cdot)$ is the Lebesgue measure of a measurable set.

For better explanation, , we give   toy examples in $\R^2$ and $\R^3$ (when $M=2$ and $M=3$). Note that in practice R1SMG requires $M>2$, here we use $\R^2$ just to visualize the idea (the probability that
a constructed $\mathbf{u}$ that is in the null space of a randomly sampled $\mathbf{v}$ is zero) which is   independent of the dimension. In $\R^2$, $\mathbb{V}_{1,2}$ is the unit circle shown in Figure \ref{fig:manifold_toy} (left). Suppose that $\mathbf{v}_1 = [-1/\sqrt{5},2/\sqrt{5}]^T$ is the randomly sampled variable in $\mathbb{V}_{1,2}$. Then $Null(\mathbf{v}_1)$ is   $\mathcal{I}$ (the dashed line orthogonal to $\mathbf{v}_1$). Clearly, the probability of sampling a point from $\R^2$ that resides on a specific line is 0. Likewise, in $\R^3$, $\mathbb{V}_{1,3}$ is the unit sphere shown in Figure \ref{fig:manifold_toy} (right). Suppose that $\mathbf{v}_2 = [1/\sqrt{3},-1/\sqrt{3},1/\sqrt{3}]^T$ is the randomly sampled variable in $\mathbb{V}_{1,3}$. Then $Null(\mathbf{v}_2)$ is    $\mathcal{F}$ (the   plane orthogonal to $\mathbf{v}_2$). The probability of sampling a point from $\R^3$ that resides on a specific plane is also 0. Although, it is publicly known that $\mathbf{v}$ is   sampled from $\mathcal{U}(\mathbb{V}_{1,M})$, the Lebesgue measure of $Null(\mathbf{v})$ (given $\mathbf{v}\subset\mathbb{V}_{1,M}$) is still 0. For example, the probability of sampling the blue points in Figure \ref{fig:manifold_toy} (left) is 0, and the probability of sampling points residing on $\mathcal{C}\subset\mathcal{F}$ (the circle orthogonal to $\mathbf{v}_2$) in Figure \ref{fig:manifold_toy} (right) is also 0. 
Hence, by generating $\mathbf{v}$ randomly, the   R$1$SMG mechanism will cause privacy leakage of using the   vectors in the null space of $\mathbf{v}$  to have   probability zero. 
\begin{figure}[htb]
  \begin{center}
     \includegraphics[width=  1\columnwidth]{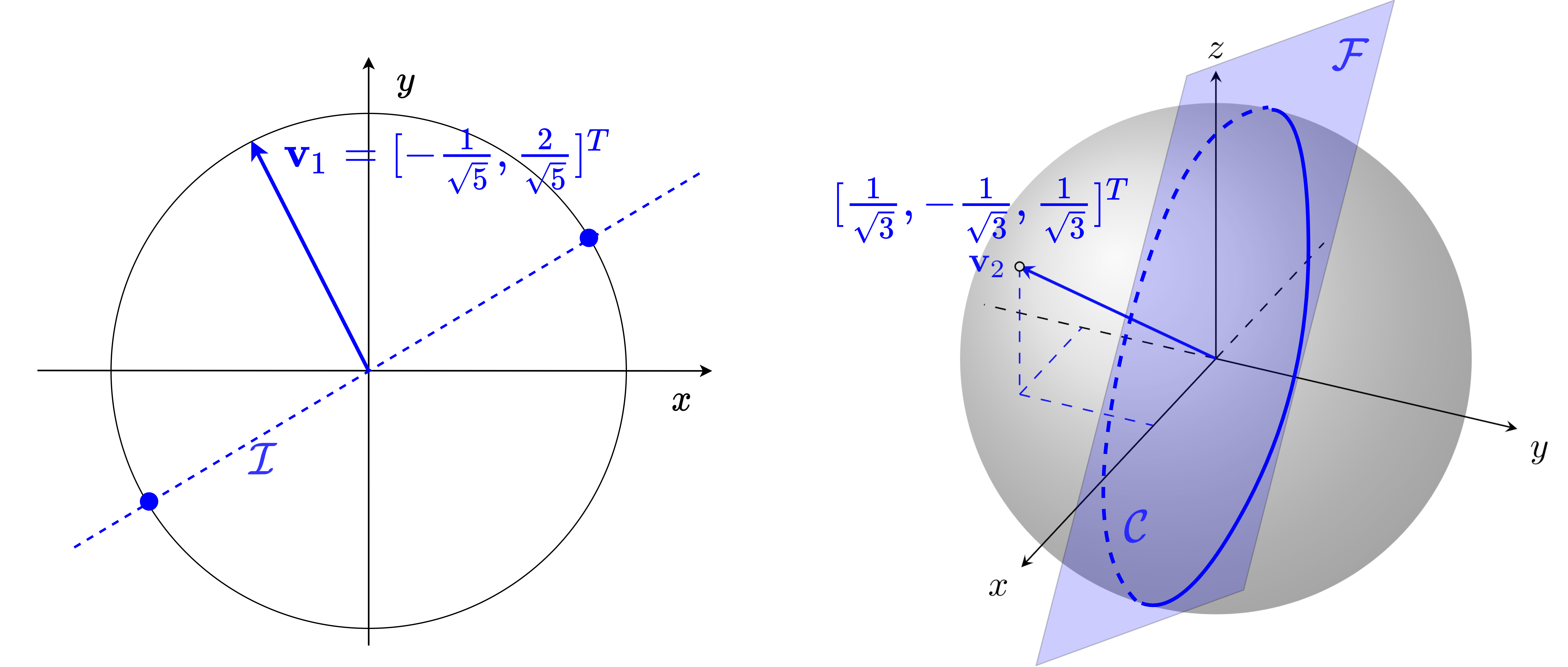}
       \end{center}
\caption{\label{fig:manifold_toy}Visualization of the null space of the noise generated by the R$1$SMG mechanism. Left: $\mathbb{V}_{1,2}$. Right: $\mathbb{V}_{1,3}$.}
\end{figure}

%% file: sections/utility_compare.tex
\section{Accuracy Stability}\label{utility}

Now, we evaluate the  accuracy stability achieved by various output perturbation DP mechanisms by studying the kurtosis and skewness of the distribution of $\L$ (the non-deterministic accuracy loss defined in Definition \ref{def:utility_cost}). In particular,

$\bullet$ Kurtosis, a descriptor of ``tail extremity'' of a probability distribution, is defined as the ratio between the $4$th moment and the square of the $2$nd moment of a random variable, i.e., $\frac{\mathbb{E}[\L^4]}{(\mathbb{E}[\L^2])^2}$. A larger kurtosis means that outliers or extreme large values are less  likely to be generated in a given probability distribution \cite{westfall2014kurtosis}.

$\bullet$ Skewness, a descriptor of the ``bulk'' of a probability distribution, is defined as the ratio between the $3$rd moment and the square root of the cube of the $2$nd moment of a random variable, i.e., $\frac{\mathbb{E}[\L^3]}{(\mathbb{E}[\L^2])^{3/2}}$. A larger skewness means that the bulk of the samples is at the left  region of the PDF and the right tail is longer.

 As a result, in order  to have high accuracy  stability on the perturbed query result, $\L$ with both \textbf{larger kurtosis and skewness is preferred}. We summarize the theoretical  results for various mechanisms in Theorem \ref{kk} and Theorem \ref{ss}.

\begin{theorem} The kurtosis of the distribution of $\L$ in the R1SMG mechanism is $\frac{35}{3}$ which is larger than that of the classic Gaussian mechanism, of the analytic Gaussian mechanism, and of   the MVG mechanism, i.e., the PDF of $\L$  in the   R1SMG mechanism is more leptokurtic than that in the   classic   Gaussian, in the analytic Gaussian, and in the MVG mechanism.
\label{kk}
\end{theorem}

\begin{proof} First, for the proposed R1SMG, we have 
\begin{gather*}
    \begin{aligned}
   & \resizebox{1\linewidth}{!}{${\rm{Kurt}_{R1SMG}}(\L)=\frac{\mathbb{E}[\L^4]}{(\mathbb{E}[\L^2])^2}\stackrel{(a)}=\frac{\mathbb{E}\left[\left((\mathbf{v}\sqrt{\sigma_*}z)^T\mathbf{v}\sqrt{\sigma_*}z\right)^4\right]}{\left(\mathbb{E}\left[\left((\mathbf{v}\sqrt{\sigma_*}z)^T\mathbf{v}\sqrt{\sigma_*}z\right)^2\right]\right)^2}    =\frac{\mathbb{E}\left[(z\sigma_*z)^4\right]}{\left(\mathbb{E}\left[(z\sigma_*z)^2\right]\right)^2} $}   \\
& \resizebox{0.8\linewidth}{!}{$\stackrel{(b)}=  \frac{48\sigma_*^4+32\sigma_*^3\sigma_*+12(\sigma_*^2)^2+12\sigma_*^2(\sigma_*)^2+(\sigma_*)^4}{\big( 2\sigma_*^2+(\sigma_*)^2  \big)^2} = \frac{35}{3}$},
    \end{aligned}
\end{gather*}
where $(a)$ is due to the noise generation process of R1SMG in (\ref{mapping2}) and $(b)$ is obtained by applying Lemma \ref{m} with $\A=1$ and $\XX=\sigma_*$ (i.e., the case of univariant Gaussian).

Then, we show that both classic Gaussian mechanism and the MVG mechanism will result in $\L$ with kurtosis less than $\frac{35}{3}$. In particular, the classic Gaussian mechanism (denoted as $\mathcal{G}$) adds i.i.d. noise to each entry of $f(\bm{x})$, $\mathbf{n}_{i}\sim\mathcal{N}(0,\sigma^2), i\in[1,M]$. Thus, we have 
$\mathbb{E}_{\mathcal{G}}[\L^t] 
 = \mathbb{E}_{z_{i}\sim\mathcal{N}(0,1)} \left[\left(\sum_{i} \sigma^2 z_{i}^2\right)^t \right]  =    \sigma^{2t} \mathbb{E}_{U\sim\chi^2(M)}   [U^t ]    =       \sigma^{2t} 2^t \frac{\Gamma(t+\frac{M}{2})}{\Gamma(\frac{M}{2})}$, 
where the last equality follows from the  $t$th moments of Chi-squared random variable. Set  $t$ to be 4 and 2, one can verify that   ${\rm{Kurt}_{\mathcal{G}}}(\mathcal{L})= 
\frac{\mathbb{E}_{\mathcal{G}}[\L^4] }{(\mathbb{E}_{\mathcal{G}}[\L^2])^2 } =
{      \sigma^8 2^4\frac{\Gamma(4+\frac{M}{2})}{\Gamma(\frac{M}{2})}      }   \Big /{    \Big(    \sigma^4 2^2\frac{\Gamma(2+\frac{M}{2})}{\Gamma(\frac{M}{2})}     \Big)^2      } <\frac{35}{3}, \forall M>1.$

The same analysis holds for the analytic Gaussian mechanism, as it adds i.i.d. Gaussian noise with variance $\sigma_A^2$.

MVG   introduces the matrix-valued noise attributed to a matrix-valued Gaussian distribution, i.e., $\mathcal{N}(\mathbf{0},\S,\P)$, and the vectorization of the noise matrix is also attributed to a multivariate Gaussian distribution, i.e., $\mathcal{N}(\mathbf{0},\S\otimes\P)$ \cite{gupta2018matrix}. 
Denote $\XX=\S\otimes\P$, and let $k>1$ be an   integer, we first have
\begin{small}
\begin{gather*}
\begin{aligned}
&(\Tr[\XX])^k  =(\Tr[\S\otimes\P])^k =
\left[\left(\sum_{i=1}^{{\rm{rank}}(\S)}\sigma_i(\S)\right)  \left(\sum_{j=1}^{{\rm{rank}}(\P)}\sigma_j(\P)\right)\right]^k\\
>& \sum_{i=1,j=1}^{{\rm{rank}}(\S),{\rm{rank}}(\P)}\big(\sigma_i(\S)\sigma_j(\P)\big)^k= \Tr\left[(\S\otimes\P)^k\right] = \Tr\left[\XX^k\right] .
\end{aligned}
\end{gather*}
\end{small}\ignorespacesafterend
Similarly,  for any positive integer $k_1>k_2\geq 1$ and symmetric $\XX\in\mathbb{PD}$, we have
\begin{small}
    \begin{equation*}
        \begin{aligned}
&\Tr[\XX^{k_1}](\Tr[\XX])^{k_2}>  \Tr[\XX^{k_1}]\Tr[\XX^{k_2}]  \\=&\left(\sum_{i=1}^{{\rm{rank}}(\XX)}(\sigma_i(\XX))^{k_1}\right)\left(\sum_{i=1}^{{\rm{rank}}(\XX)}(\sigma_i(\XX))^{k_2}\right)    >  \sum_{i=1}^{{\rm{rank}}(\XX)}    (\sigma_i(\XX))^{k_1+k_2}\\ = &\Tr[\XX^{k_1+k_2}].
        \end{aligned}
    \end{equation*}
\end{small}\ignorespacesafterend
Then, according to Lemma \ref{m}, we have
\begin{small}
\begin{gather*}
    \begin{aligned}
&{\rm{Kurt}_{MVG}}(\L) =   \frac{\mathbb{E}[\L^4]}{(\mathbb{E}[\L^2])^2}=\frac{\mathbb{E}[(\mathbf{n}^T\mathbf{n})^4]}{(\mathbb{E}[(\mathbf{n}^T\mathbf{n})^2])^2} = \frac{\mathbb{E}[Q(\mathbf{n})^4]}{(\mathbb{E}[Q(\mathbf{n})^2])^2}\\
= & \resizebox{0.95\linewidth}{!}{$\frac{48\Tr[\XX^4]+32\Tr[\XX^3]\Tr[\XX]+12(\Tr[\XX^2])^2+12\Tr[\XX^2](\Tr[\XX])^2+(\Tr[\XX])^4}{\big( 2\Tr[\XX^2]+(\Tr[\XX])^2  \big)^2}$}.
    \end{aligned}
\end{gather*}
\end{small}\ignorespacesafterend
Next, we can obtain 

\noindent\resizebox{1\linewidth}{!}{$48\Tr[\XX^4] < \frac{104}{3} (\Tr[\XX^2])^2+\frac{8}{3}\Tr[\XX^2](\Tr[\XX])^2+\frac{32}{3}(\Tr[\XX])^4$}
and $32\Tr[\XX^3]\Tr[\XX]<32\Tr[\XX^2](\Tr[\XX])^2$. Thus, we get
\noindent\resizebox{1\linewidth}{!}{${\rm{Kurt}_{MVG}}(\L)< \frac{\frac{35}{3}\left(4\Tr[\XX^2]+4\Tr[\XX^2](\Tr[\XX])^2+(\Tr[\XX])^4\right)}{ 4\Tr[\XX^2]+4\Tr[\XX^2](\Tr[\XX])^2+(\Tr[\XX])^4   } = \frac{35}{3}$},
which concludes the proof.
\end{proof}

\begin{theorem} 
The skewness of the distribution of $\L$ in the  R1SMG mechanism is $\frac{5\sqrt{3}}{3}$ which is larger than that of the classic  Gaussian mechanism, of  the analytic Gaussian mechanism,  and of the MVG mechanism, i.e., the PDF of $\L$   in the R1SMG mechanism is more right-skewed than that in  the classic  Gaussian, in the analytic Gaussian,   and in the MVG mechanism.
\label{ss}
\end{theorem}

\begin{proof} The proof follows the same procedure as the proof of Theorem \ref{kk}, thus we only show the key steps here.

For the proposed R1SMG, we have ${\rm{skew}_{\mathrm{R1SMG}}}(\L)=\frac{\mathbb{E}[\L^3]}{(\mathbb{E}[\L^2])^{3/2}}=\frac{\mathbb{E}\left[(z\sigma_*z)^3\right]}{\left(\mathbb{E}\left[(z\sigma_*z)^2\right]\right)^{3/2}} =  \frac{   \sigma_*^3 2^3\frac{\Gamma(3+\frac{1}{2})}{\Gamma(\frac{1}{2})}         }{    \left(    \sigma_*^2 2^2\frac{\Gamma(2+\frac{1}{2})}{\Gamma(\frac{1}{2})}      \right)^{3/2}  }    =\frac{5\sqrt{3}}{3}.$

For the classic Gaussian mechanism, one can verify that \\\resizebox{1\linewidth}{!}{${\rm{skew}_{\mathcal{G}}}(\mathcal{L})= \frac{      \sigma_G^6 2^3\frac{\Gamma(3+\frac{M}{2})}{\Gamma(\frac{M}{2})}      }    {    \left(    \sigma_G^4 2^2\frac{\Gamma(2+\frac{M}{2})}{\Gamma(\frac{M}{2})}     \right)^{3/2}      } = \frac{      \frac{\Gamma(3+\frac{M}{2})}{\Gamma(\frac{M}{2})}      }    {    \left(    \frac{\Gamma(2+\frac{M}{2})}{\Gamma(\frac{M}{2})}     \right)^{3/2}      } <\frac{5\sqrt{3}}{3}, \forall M>1$.}

The same analysis holds for the analytic Gaussian mechanism, as it adds i.i.d. Gaussian noise with variance $\sigma_A^2$.

For the MVG mechanism, it is easy to check that\\  ${\rm{skew}_{MVG}}(\L)=\frac{\mathbb{E}[\L^3]}{(\mathbb{E}[\L^2])^{3/2}}   =\frac{\mathbb{E}[(\mathbf{n}^T\mathbf{n})^3]}{(\mathbb{E}[(\mathbf{n}^T\mathbf{n})^2])^{3/2}} = \frac{\big(8\Tr[\XX^3]+6\Tr[\XX^2]\Tr[\XX]+ (\Tr[\XX])^3 \big)}{\big( 2\Tr[\XX^2]+(\Tr[\XX])^2  \big)^{3/2}}$,  $\XX  =\S\otimes\P$. Besides, we have   shown that for any positive integer $k$, $k_1$, and $k_2$, we have  
$(\Tr[\XX])^k  >  \Tr[\XX^k]$ and $\Tr[\XX^{k_1}](\Tr[\XX])^{k_2}>\Tr[\XX^{k_1+k_2}]$.
Thus, one can   check that $({\rm{skew}_{MVG}}(\L))^2 =\frac{\big(8\Tr[\XX^3]+6\Tr[\XX^2]\Tr[\XX]+ (\Tr[\XX])^3 \big)^2}{\big( 2\Tr[\XX^2]+(\Tr[\XX])^2  \big)^{3}}<\frac{\frac{25}{3}\big( 2\Tr[\XX^2]+(\Tr[\XX])^2  \big)^{3}}{\big( 2\Tr[\XX^2]+(\Tr[\XX])^2  \big)^{3}}$, i.e., ${\rm{skew}_{MVG}}(\L)<\frac{5\sqrt{3}}{3}$.
\end{proof}

 \begin{remark}
We can consider the univariate Gaussian with unit covariance, i.e., $\mathcal{N}(\mu,1)$,  as a reference distribution, whose kurtosis value and skewness value are $3$ and $0$, respectively. Then, it suggests that the distribution of $\L$ in the R1SMG mechanism are much more leptokurtic and right-skewed than $\mathcal{N}(\mu,1)$. Moreover, since the noise  used in the classic Gaussian, analytic Gaussian, and MVG mechanisms are characterized by full-rank covariance matrices, their kurtosis values and skewness values asymptotically converge to $3+12/H$ and $\sqrt{8/H}$, respectively \cite{johnson1995continuous} ($H$ is the degree of freedom of the obtained $\L$, and $H=M$ for $f(\bm{x})\in\R^M$). It means that the distributions of $\L$ in the Gaussian mechanism and MVG  are similar to $\mathcal{N}(\mu,1)$ when $M$ is large. 
\label{all_remark}
\end{remark}

In Section \ref{sec:nyc}, by using  2D count query as an example, we will   empirically show that the distribution of $\L$ obtained in the R1SMG mechanism is more   leptokurtic and right-skewed than those of the
classic Gaussian, the analytic Gaussian, and the MVG mechanisms (cf. Figure \ref{fig:utility_cost_L}). In other words, the  R1SMG mechanism can provide differentially private   query results with the highest accuracy stability. 

Based on the above theoretical analysis, we can have the following corollary. 
\begin{corollary}\label{corollary:kkss} 
The R1SMG  mechanism outperforms  the classic Gaussian, the analytic Gaussian,  and the MVG mechanism, because it   boosts the utility of the query results by achieving lower expected accuracy loss and higher accuracy stability. 
\end{corollary}

%% file: sections/experiments.tex
\section{Experiments}
\label{sec:exp}

In this section, we conduct three case studies to validate the utility boosting achieved by the   R1SMG mechanism, i.e.,  2D count  query, principal component analysis (PCA), and deep learning, all in a differentially private manner. The query results of these case studies are either matrices or tensors, so the R1SMG, classic Gaussian, and analytic Gaussian mechanisms will   first generate noise in vector form, and then reshape the noise into the    forms required by different studies.

\subsection{Case Study   \rom{1}:  Uber Pickup Count Query}
\label{sec:nyc}

In this case study, we utilize the New York City (NYC) Uber pickups dataset \cite{nyc_pickups}, on which we are  interested  in releasing the counts of Uber pickups in different areas  of NYC from ``4/1/2014 00:11:00'' to ``4/3/2014 23:57:00'' in a differentially private manner\footnote{We   choose 3 days instead of more days for a better visualization result.}. The size of the count query $f(\bm{x})$ is determined by the partition areas (defined later) of NYC. Location and trajectory privacy breaches have been reported and investigated in  several research works \cite{liu2016dependence,xiao2015protecting,jiang2013publishing,ou2018releasing}. Thus, it is important to guarantee that the existence or absence of any pickup record is private when sharing the count query.

In particular, we partition the map of NYC   into small areas by evenly dividing  the latitude and longitude into 89 disjoint intervals\footnote{Note that one can also choose arbitrary  partition size in practice \cite{hay2016principled}. Here we just set it to  89 as an   example.}. As a result, the considered  query is  $f(\bm{x})\in\R^{89\times 89}$, whose entries record  the number of Uber pickups in different small areas. The squared sensitivity of this query is 2, because the absence or presence of a specific pickup can change at most 2 entries of $f(\bm{x})$ by 1.  We also consider that the data consumer (i.e., the query result recipient) has the prior knowledge of the valid pickup areas in the resulted $89\times 89$ small areas. Thus, he can perform post-processing on the noisy count query to eliminate the values   in   invalid areas (e.g., it is impossible to have Uber pickups over the Hudson River). We visualize the non-private Uber pickup count query    in Figure \ref{fig:nyc_pickups}(a). We observe   high volumes of pickups around Soho, Fifth avenue, and LaGuardia Airport.

\begin{figure}[htb]
  \begin{center}
   \begin{tabular}{cc}
     \includegraphics[width= .45\columnwidth]{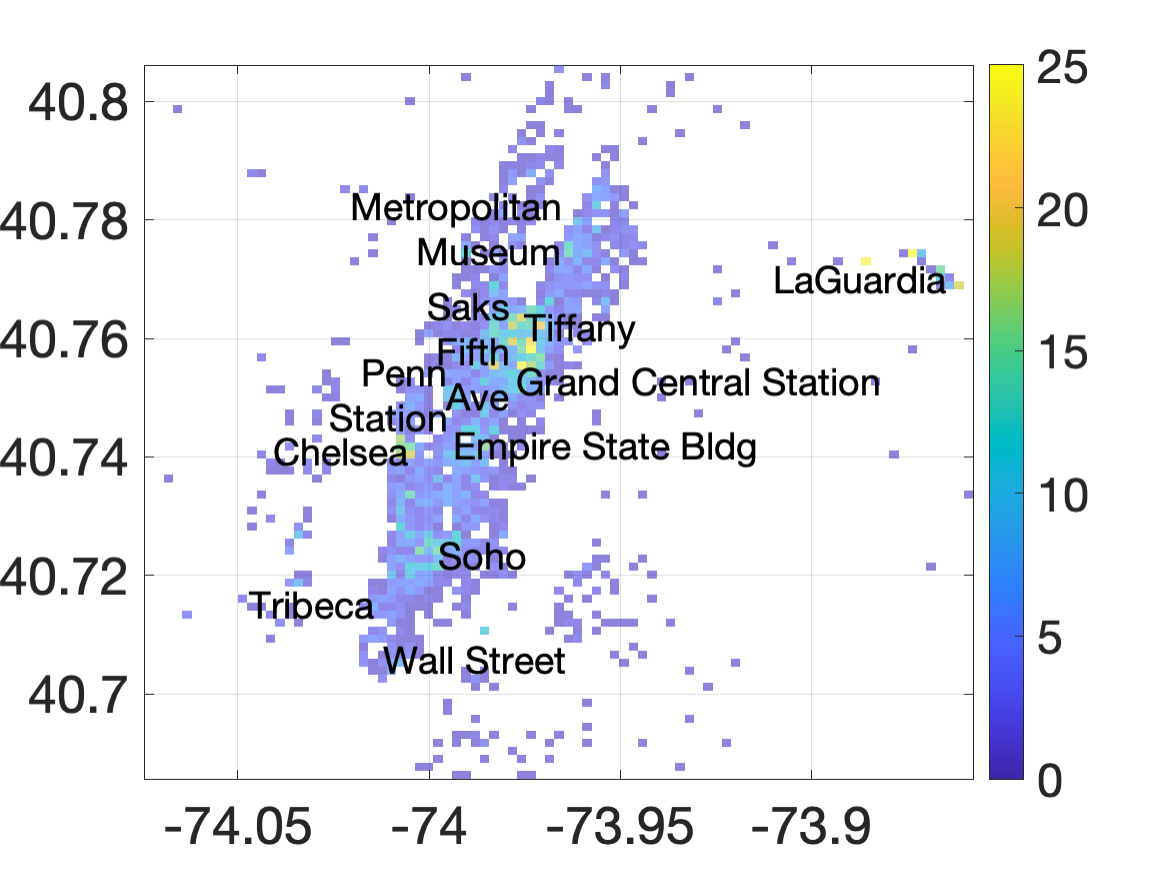}
&   \includegraphics[width= .45\columnwidth]{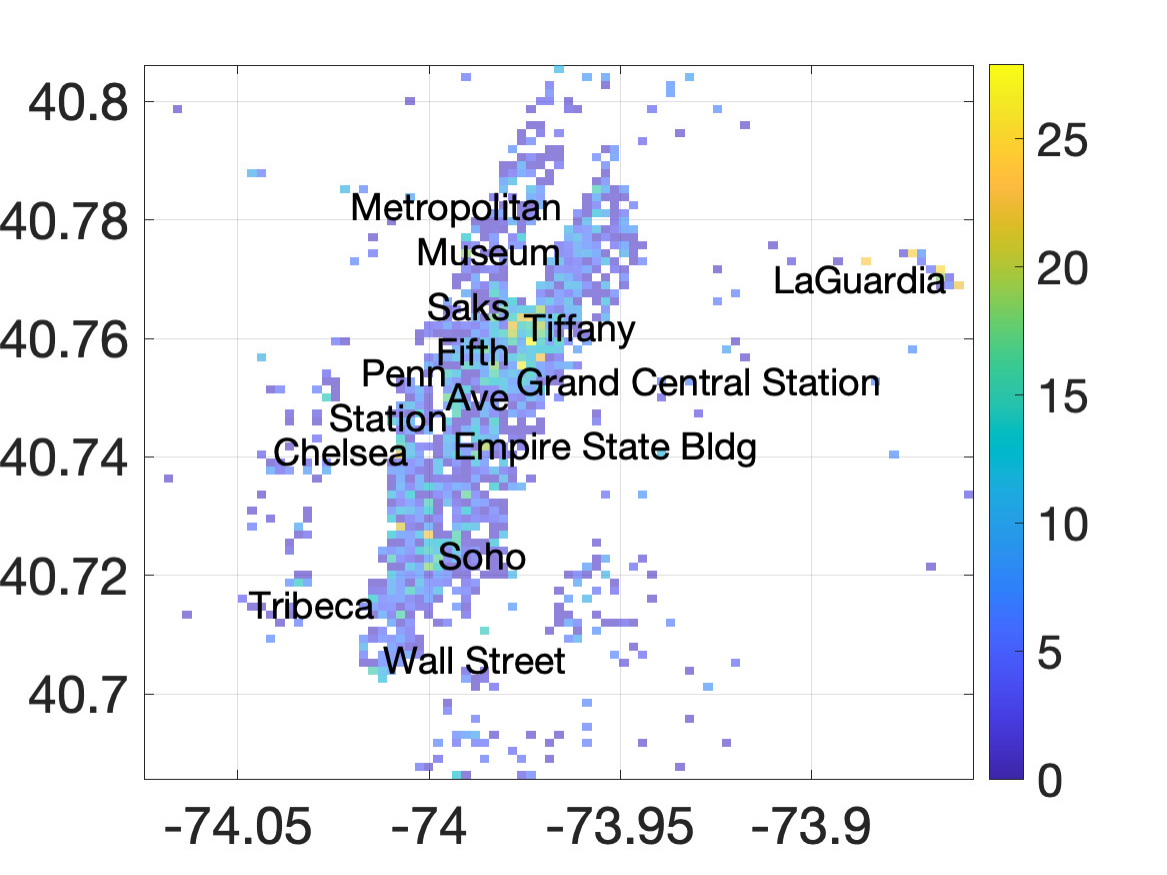}\\
	{\small  (a)   Non-private}&
    {\small  (b)   R1SMG, $\epsilon = 10^{-5}$}\\
 \includegraphics[width= .45\columnwidth]{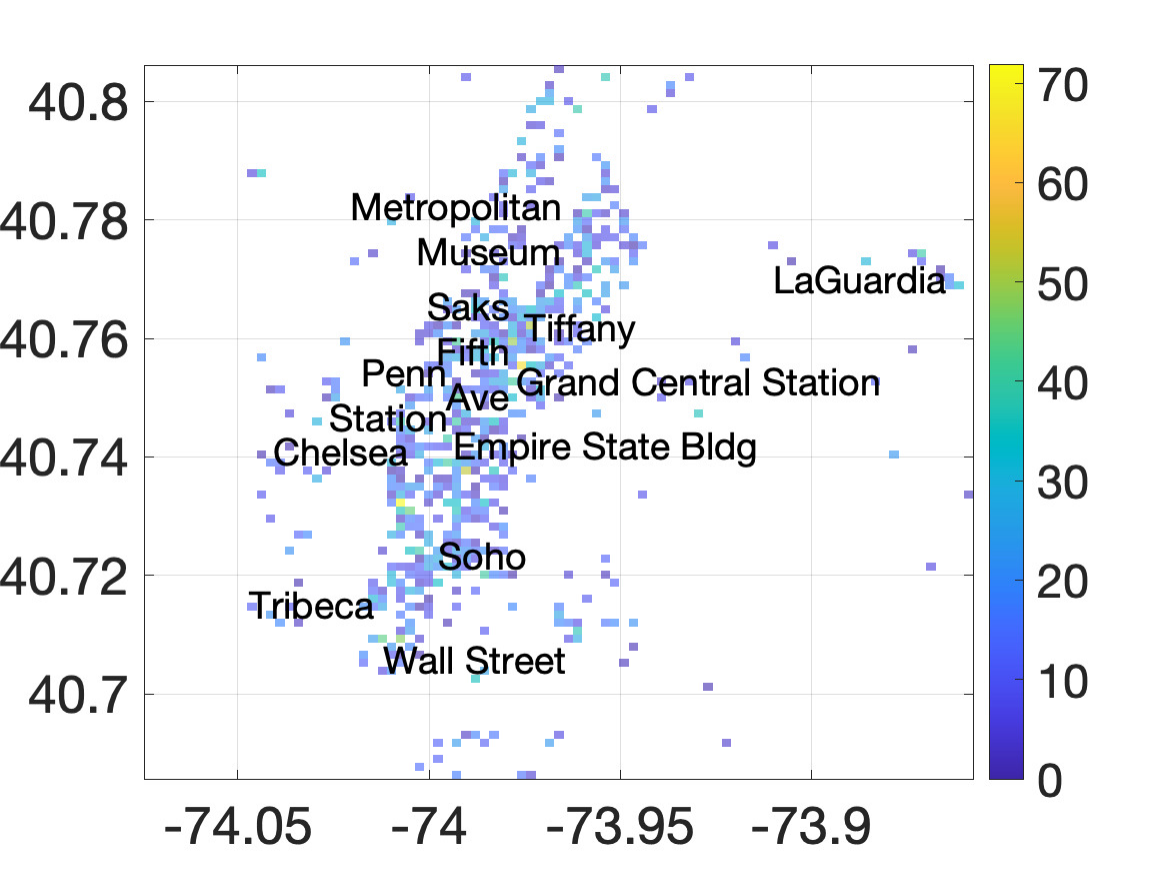}
&\includegraphics[width= .45\columnwidth]{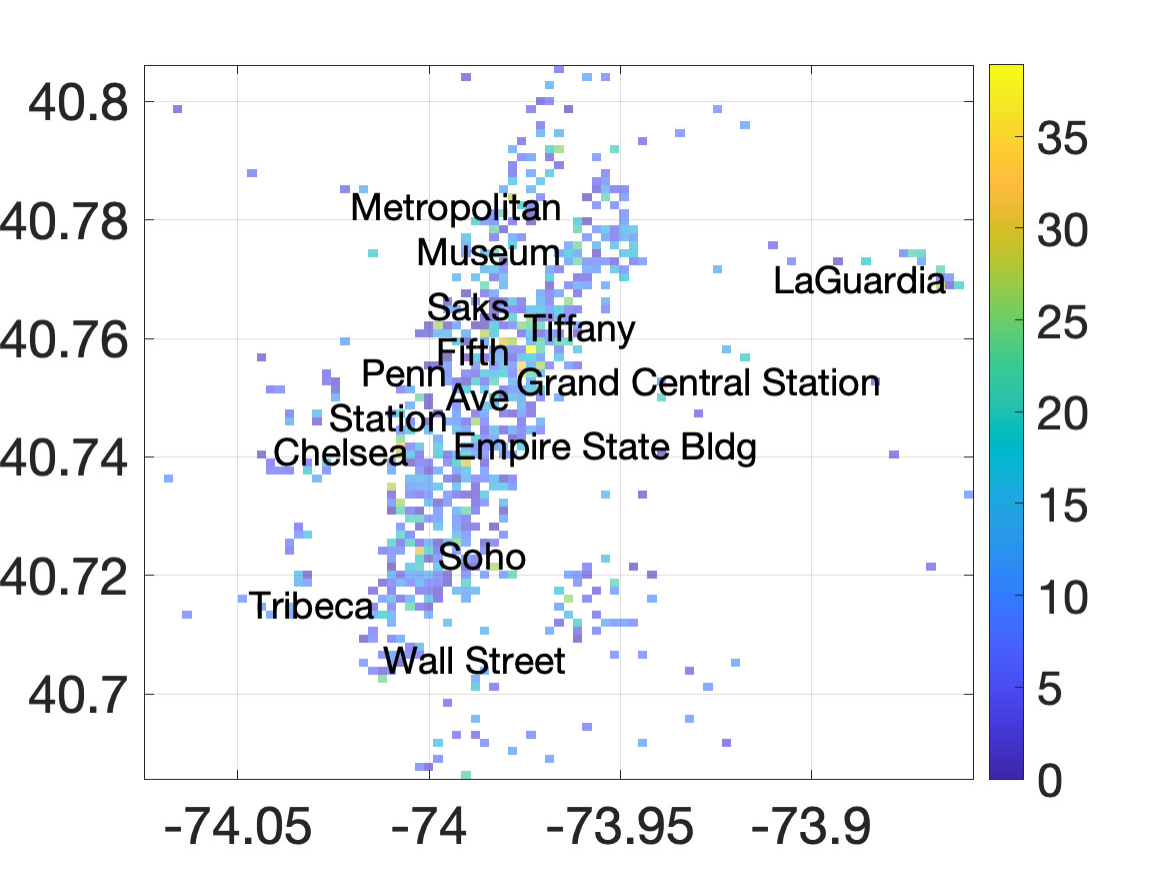}\\
     {\small  (c)    classic Gauss.}&
      {\small  (d)     Analytic   Gauss.}\\
\includegraphics[width= .45\columnwidth]{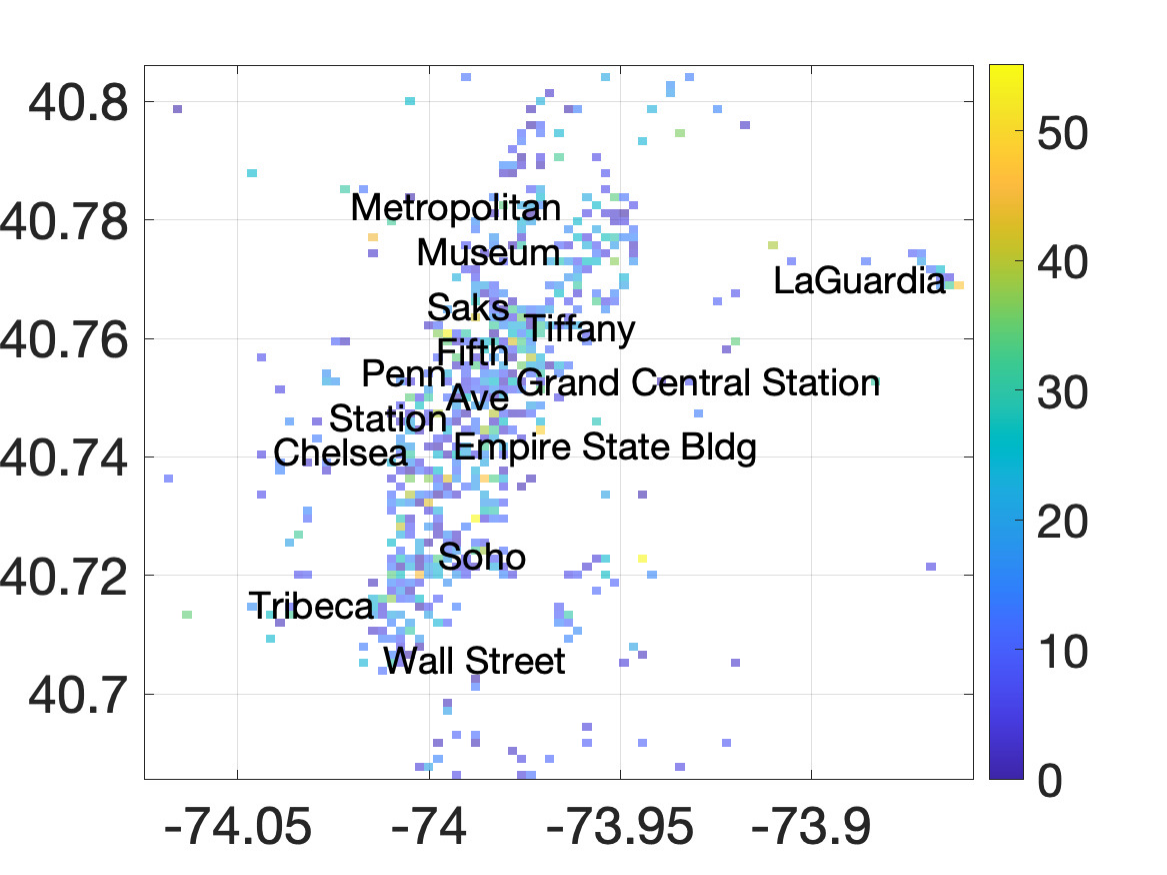}
&   \includegraphics[width=.45\columnwidth]{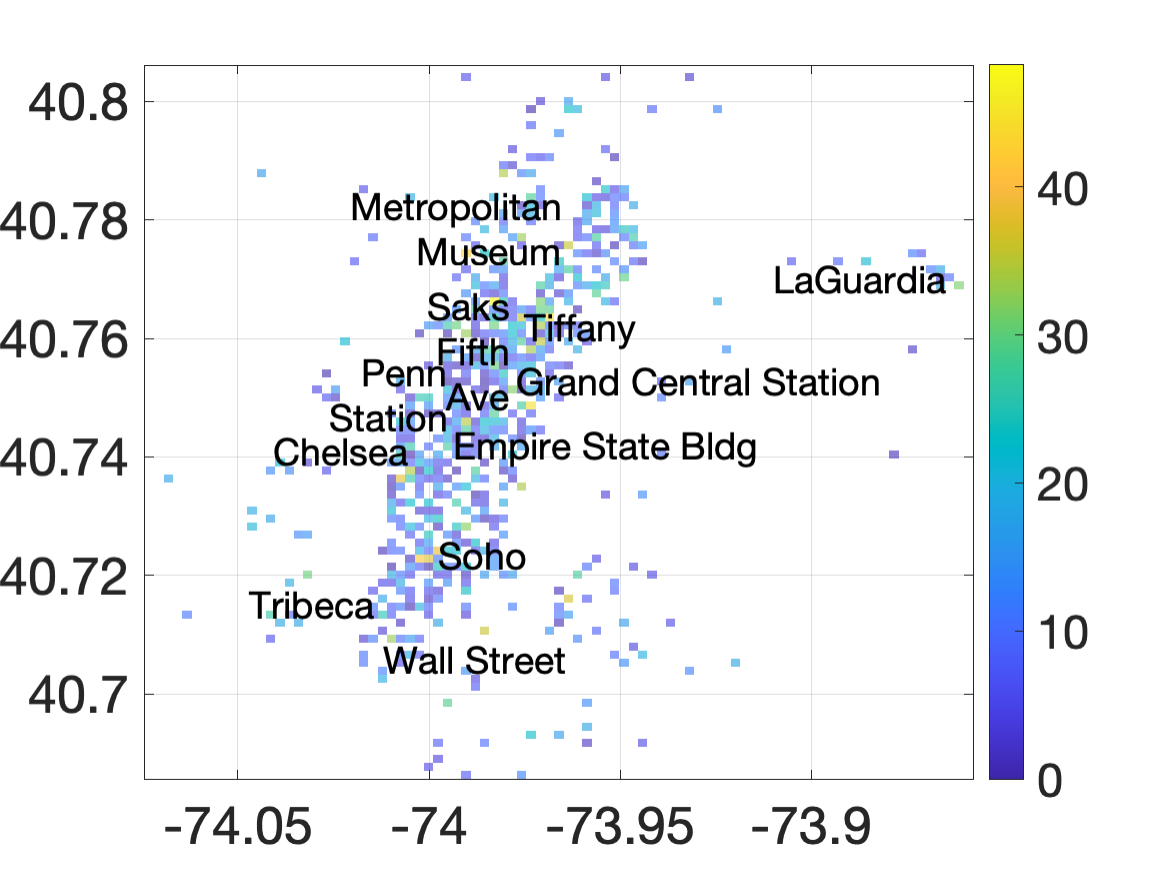}\\
          	{\small  (e)    MVG}&
    {\small  (f)   MGM}\\
\includegraphics[width= .45\columnwidth]{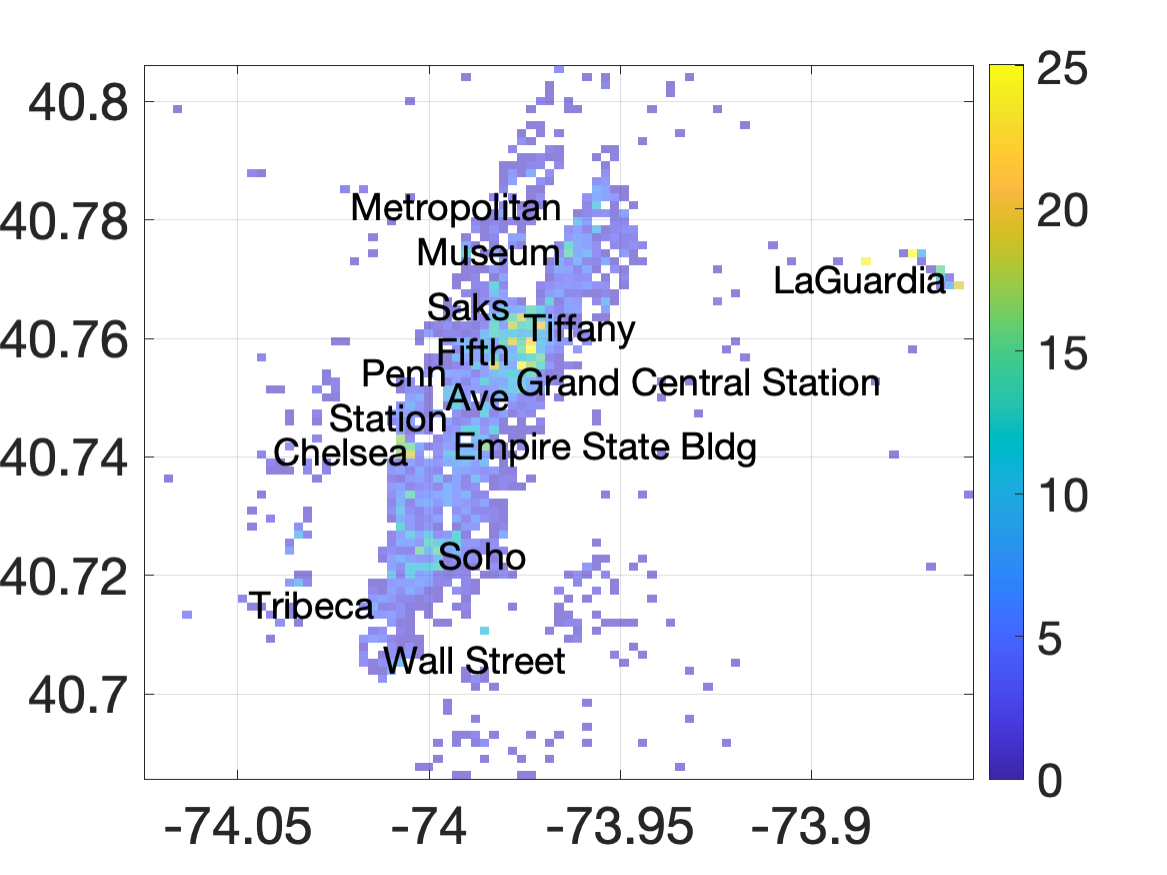}
&\includegraphics[width= .45\columnwidth]{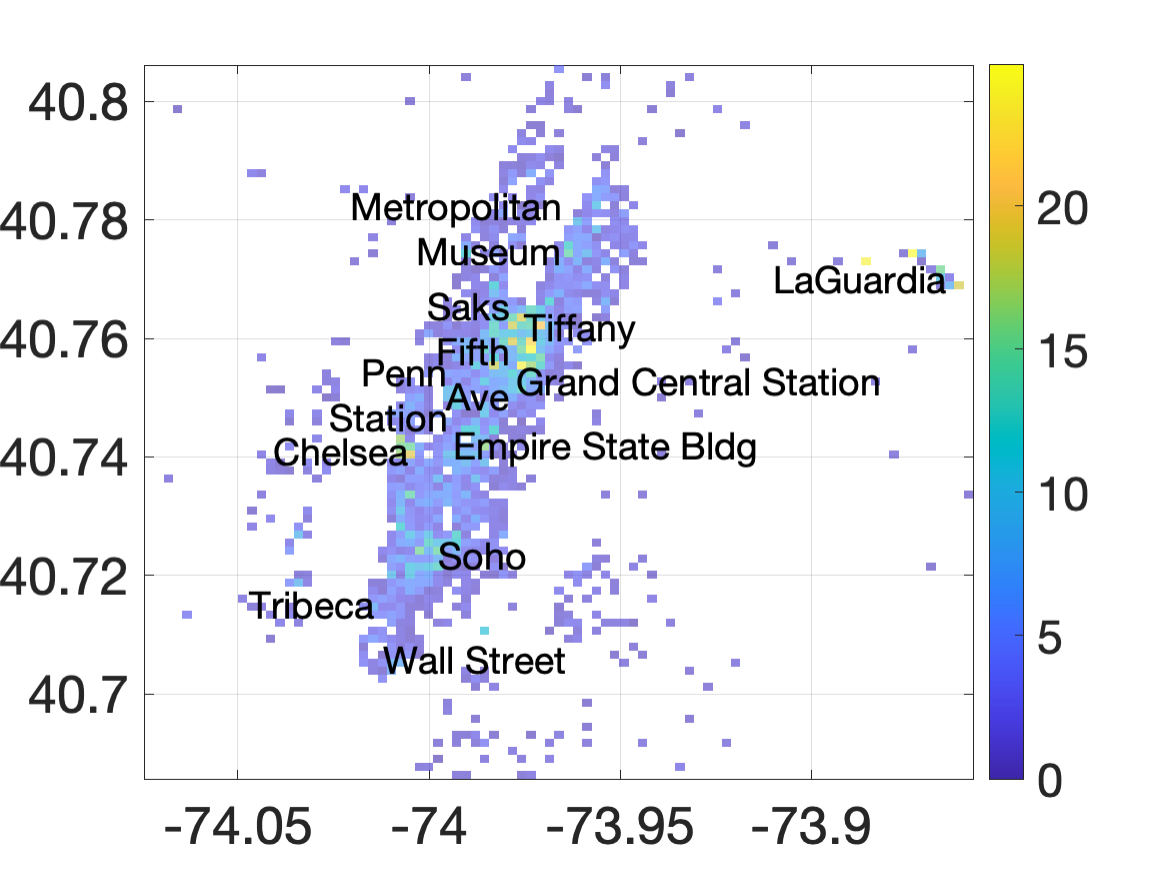}\\
       {\small  (g)     DAWA}&{\small  (h)     $H_b$}
    \end{tabular}
       \end{center}
       \vspace{-4mm}
  \caption{\label{fig:nyc_pickups}Visualization of 
  (a) non-private counts, (b)-(h) are differentially private 2D counts obtained by  the R1SMG,   classic Gaussian,   analytic Gaussian,     MVG,   MGM,    DAWA, and $H_b$ mechanisms, respectively. $\epsilon$ is $10^{-5}$ for the R1SMG mechanism and is 0.5 for the other mechanisms.}
\end{figure}

\textbf{Comparisons with  Mechanisms.} In additional to the $(\epsilon,\delta)$-DP output perturbation mechanisms, i.e., the   classic Gaussian, analytic Gaussian\footnote{\url{https://github.com/BorjaBalle/analytic-gaussian-mechanism}.}, MVG\footnote{\url{https://github.com/inspire-group/MVG-Mechansim}.}, and MGM (a variant of MVG~\cite{yang2021matrix}) mechanisms, we also compare the R1SMG mechanism with the  mechanisms that are specially developed for differentially private 2D count  queries. In particular, they are (i) DAWA \cite{li2014data}, which obtains differentially private count queries    by adding noise that is adapted to both the input data and   the given query set, and (ii) $H_b$ \cite{qardaji2013understanding}, which answers range queries using noisy hierarchies of equi-width histograms\footnote{Source code of both DAWA and $H_b$ are available at \url{https://github.com/dpcomp-org/dpcomp\_core}.}.

\begin{figure*}[htb]
  \begin{center}
   \begin{tabular}{cc}
     \includegraphics[width= .5\textwidth]{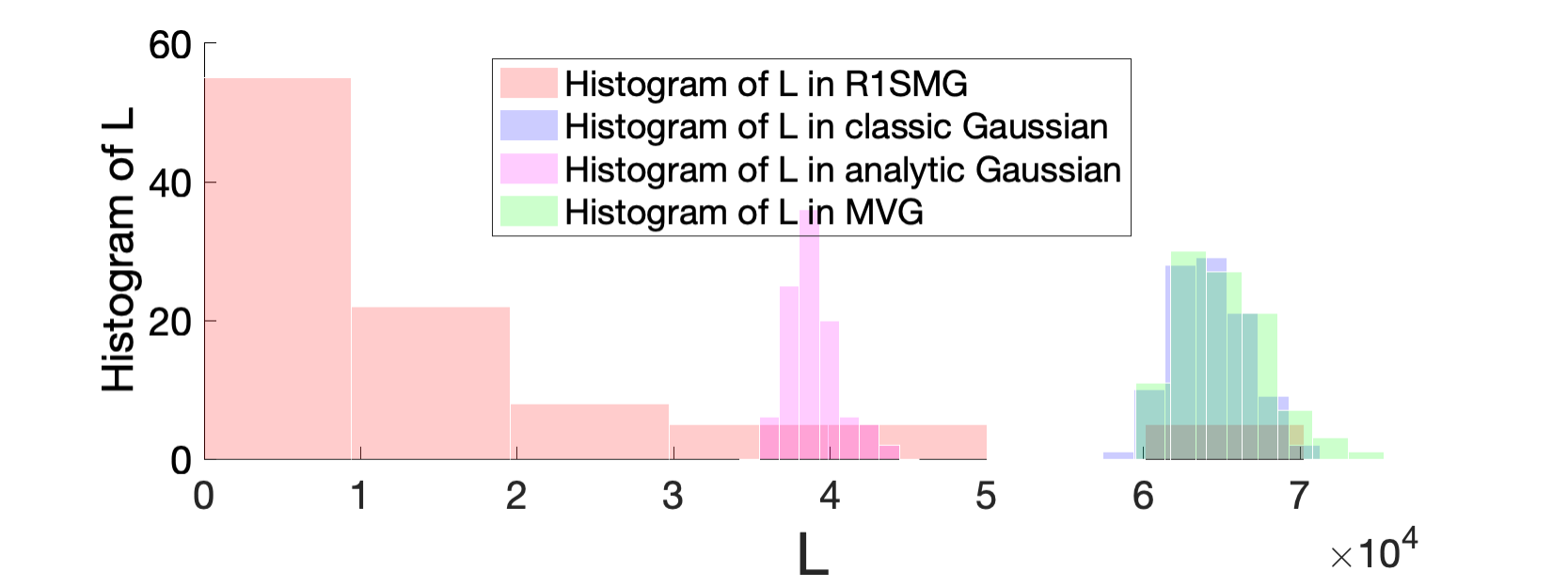}
     &
     \includegraphics[width= .5\textwidth]{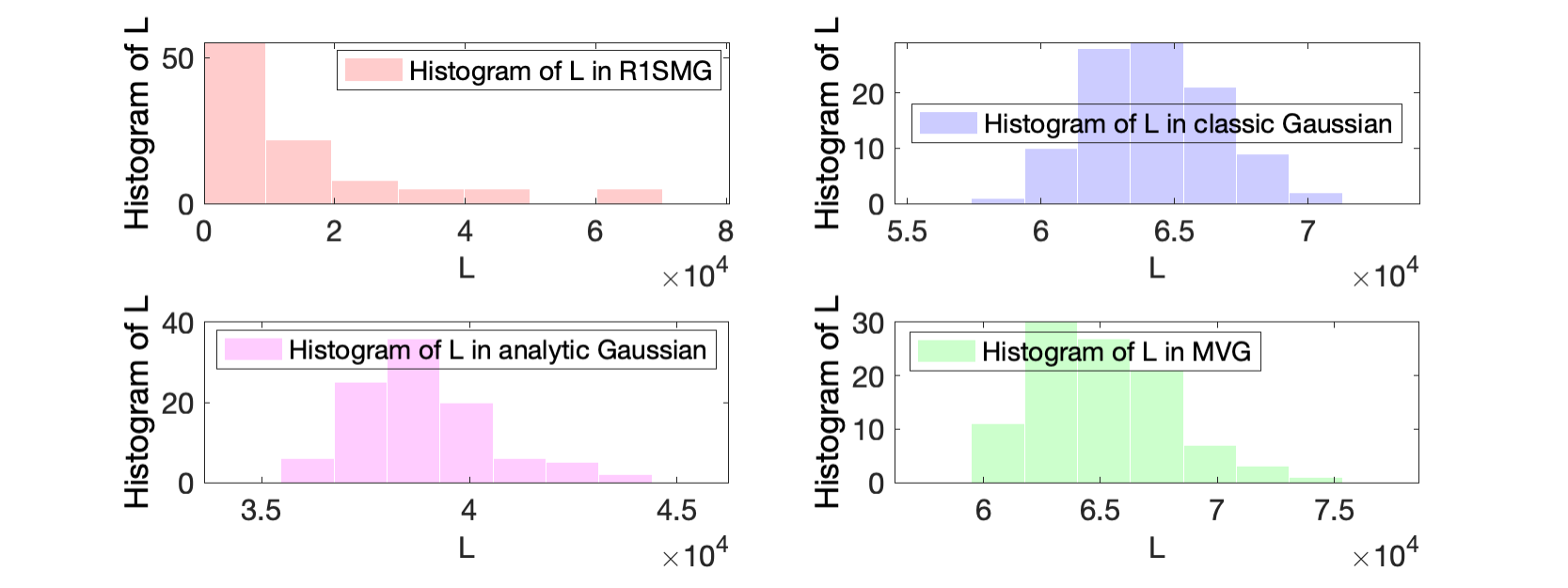} \\
     {(a)  Histogram of  $\mathcal{L}$ for various mechanisms}&
    {(b)  Histogram of  $\mathcal{L}$ (zoomed in)  }\\
 \includegraphics[width= .5\textwidth]
   {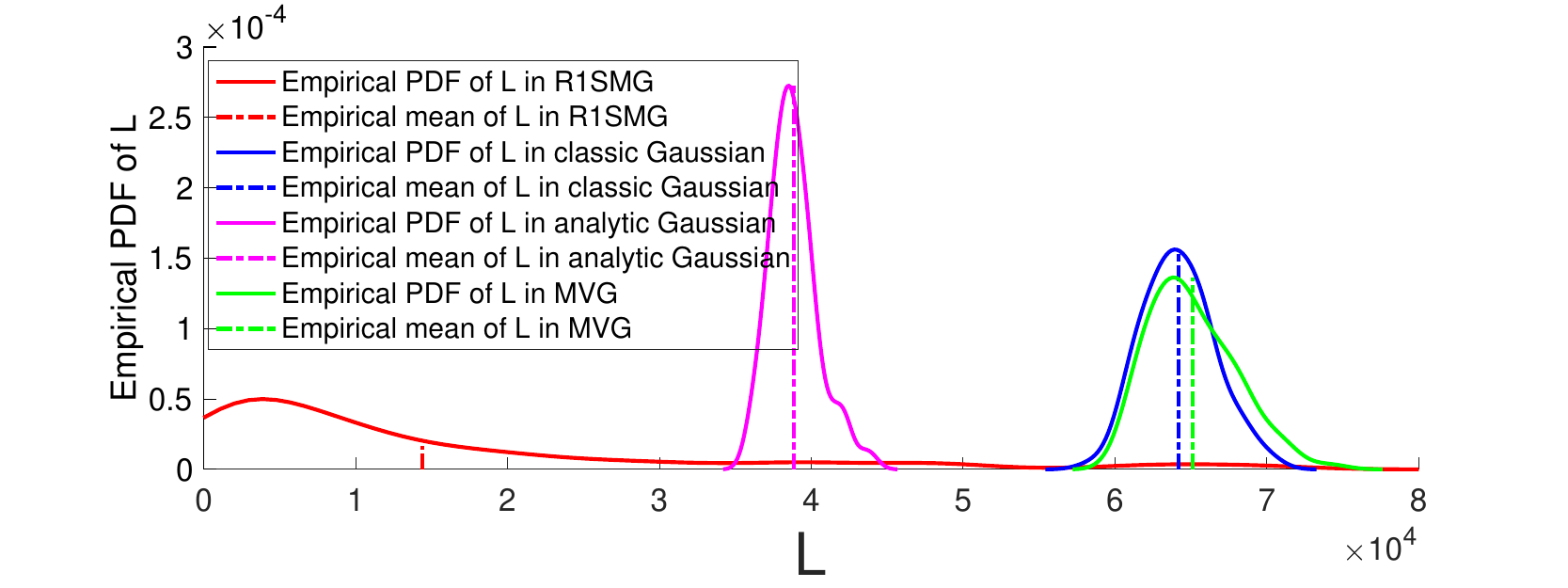}   &
\includegraphics[width= .5\textwidth]{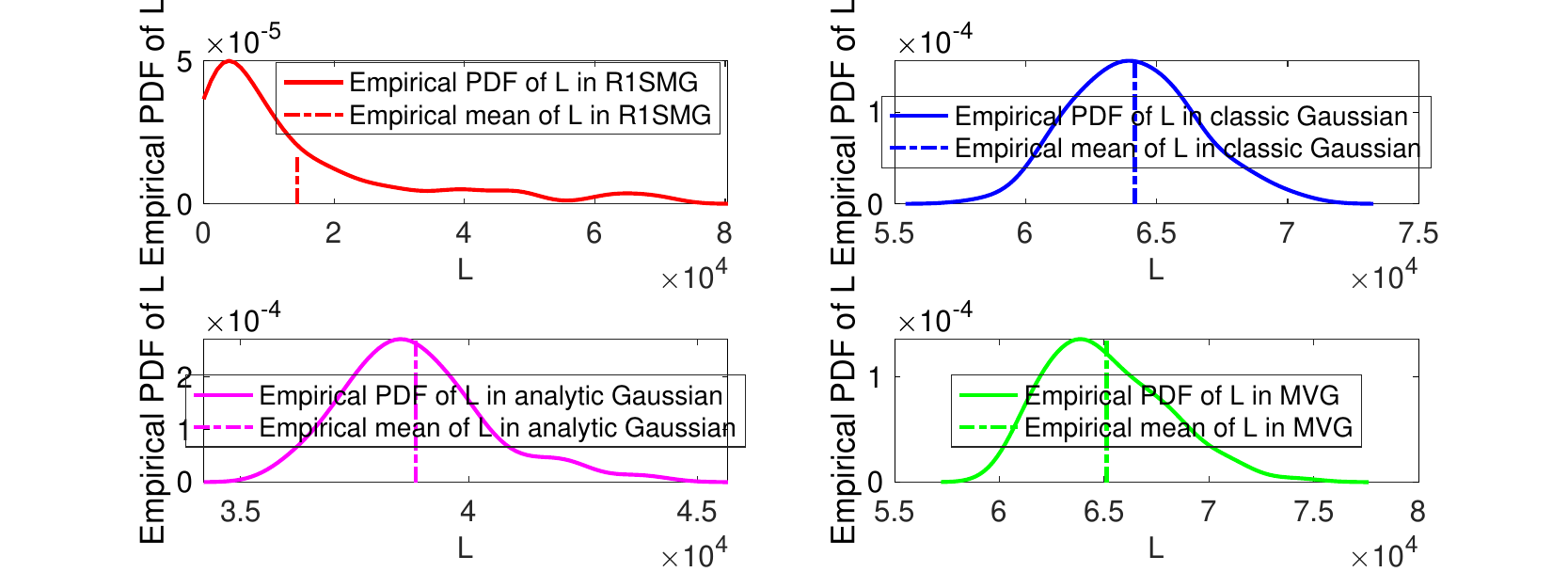}\\
{(c) Empirical PDF and mean of  $\mathcal{L}$ for various mechanisms }
    &  {(d)  Empirical PDF and mean of  $\mathcal{L}$   (zoomed in)}\\
              \end{tabular}
       \end{center} 
       \vspace{-3mm}
  \caption{\label{fig:utility_cost_L}Accuracy loss    introduced by different output perturbation mechanisms when  $\delta = 10^{-7}$, $\epsilon = 10^{-5}$ for the R1SMG mechanism and $\epsilon = 0.5$ for the other mechanisms.}
\end{figure*}

\textbf{Results.} We visualize the differential-privately released count queries obtained by the   R1SMG and the other comparing mechanisms in Figures \ref{fig:nyc_pickups}(b)-(h). The privacy budgets for all the comparing  mechanisms are 
$\epsilon = 0.5$ and  $\delta = 10^{-7}$. For the R1SMG mechanism, we set $\epsilon = 10^{-5}<\frac{0.5}{89^2}$ according to Remark~\ref{epsilon-remark} (more fair comparison results obtained by letting $\epsilon = 10^{-5}$ for the other mechanisms are shown in Figure~\ref{fig:nyc-small-eps} in Appendix~\ref{app:small-eps}, where even more serious performance degradation of the previous mechanisms can be observed).

Clearly, the R1SMG mechanism outperforms all the comparing   output perturbation mechanisms (i.e., Figure  \ref{fig:nyc_pickups}(c)-(f))  even under a very restricted privacy budget, as it preserves   the global and local patterns of the non-private count query in Figure  \ref{fig:nyc_pickups}(a) in the best way. In contrast, the classic Gaussian,   MVG, and MGM  mechanisms greatly compromise the utility of the count query results although their privacy budget is $10^4$ times as large as that of our R1SMG mechanism. For example, in Figures  \ref{fig:nyc_pickups}(c), (e), and (f)  many areas on the Manhattan island have no pickup counts, and   the maximum count is also increased significantly (i.e., higher than 60, 50, and 40, respectively,  compared to the non-private maximum count which is only around 27).  
Although the analytic Gaussian mechanism outperforms the classic Gaussian and MVG mechanisms by calibrating the noise variance exactly, it still introduces noise with higher magnitude than the R1SMG mechanism.  Finally, we observe that the R1SMG mechanism has comparable performance with DAWA and $H_b$ that are developed specially for differentially private range queries.    Note that DAWA has slightly better results  because it is data-dependent (the noise is customized to the input data \cite{li2014data}), and that  $H_b$  adopts ``constrained inference” (an accuracy boosting scheme) to control the mean square error of counting queries at various granularities.   Moreover, not only does the R1SMG mechanism introduces comparable errors with those by   the two task-specific mechanisms under  a much stricter privacy guarantee, but also it is very promising in other real-world tasks because it is not task-dependent.

\textbf{Utility boosting and stability.} 
Next, we   empirically   corroborate that the   R1SMG mechanism boosts the utility of the NYC Uber pickup count query by achieving lower accuracy loss and higher accuracy stability.   Specifically, by setting $\epsilon=10^{-5}$ for the R1SMG mechanism and $\epsilon=0.5$ for the other mechanisms and $\delta = 10^{-7}$,  we repeat the noise generation process  100 times  for the R1SMG, classic Gaussian, analytic Gaussian, and MVG  mechanisms, respectively, and  calculate the accuracy loss (i.e., $\mathcal{L}$, in Definition \ref{l}). In Figure  \ref{fig:utility_cost_L}, we plot the corresponding histograms,   empirical average values of $\mathcal{L}$,  and the empirical PDFs, obtained by different    mechanisms. In particular, Figure  \ref{fig:utility_cost_L}(a) plots the histograms of $\mathcal{L}$ for various mechanisms on the same $x$-axis, and Figure  \ref{fig:utility_cost_L}(b) is the zoomed in histograms for each mechanism. Figure  \ref{fig:utility_cost_L}(c) demonstrates the empirical PDF and  mean of $\mathcal{L}$ for each  mechanism  on the same $x$-axis, and Figure  \ref{fig:utility_cost_L}(d) shows the corresponding zoomed in plots.

Figure \ref{fig:utility_cost_L}   shows that the empirical mean accuracy loss obtained by using the R1SMG mechanism is   less than those obtained by  the other  mechanisms, in spite of the fact that the privacy budget is smaller than $1/10^4$ of the others.  This means that under  very restricted  privacy guarantee,  the R1SMG mechanism  is able to   introduce  additive noise with small magnitude (squared Frobenius norm), which lead to lower empirical accuracy loss.  In particular, $\L$ obtained by the R1SMG mechanism is around $1.5\times 10^4$  when $\epsilon = 10^{-5}$, whereas, $\L$ obtained by the other mechanisms are around  $5\times 10^4$  even with $\epsilon=0.5$.  We can also find that the histogram and the empirical PDF of $\mathcal{L}$ obtained by the R1SMG mechanism are much more leptokurtic and right-skewed than those achieved by  the others. Particularly, in Figure~\ref{fig:utility_cost_L}(d), compared with those obtained by the other mechanisms that have bell shapes, the empirical PDF of $\L$ achieved by the R1SMG mechanism has a longer and fatter tail, a higher and sharper central peak, and the density is more densely distributed on the left side.  The results corroborate the analysis in Remark \ref{all_remark}. Thus,  under very strict  $(\epsilon,\delta)$-DP requirement, it is still  less likely for the R1SMG mechanism to generate  additive   noise with large magnitude.

\subsection{Case Study   \rom{2}: PCA}

In this case study, we explore differentially private principal component analysis (PCA) using the Swarm Behaviour dataset \cite{swarm}. It   contains 24016 records of   flocking behaviour (i.e., the way that groups of birds, insects, fish or other animals, move close to each other). Each record has 2400 features characterizing   a flocking observation, e.g., radius of separation of groups of insects, moving direction, moving velocity, etc. Each data record is   assigned a binary label, and ``1" represents ``flocking behavior" and ``0" represents ``not flocking". The values of all features are normalized between $-1$ and $1$. Similar to~\cite{Chanyaswad2018}, we   consider the query $f(\bm{x})$ as the covariance matrix of this dataset, i.e., $f(\bm{x}) = {D^TD}/{24016}\in\R^{2400\times 2400}$, where  $D\in\R^{24016\times 2400}$ denotes the considered Swarm Behaviour data.\footnote{Since covariance matrices are symmetric, the data consumer will perform post processing $\big(\mathcal{M}(f(\bm{x}))+\mathcal{M}(f(\bm{x}))^T\big)/2$ to obtain symmetric  result. The privacy guarantee  naturally holds due to the post-processing immunity  \cite{dwork2014algorithmic}.} Then, the $l_2$ sensitivity of the query function is   $s_2(f)  = \sup_{D_1, D_2}\frac{||{D_1}(i)^T{D_1}(i)-{D_2}(i)^T{D_2}(i)||_F}{24016} = \frac{\sqrt{2\times 2400^2}}{24016}=\frac{2400\sqrt{2}}{24016}$ (cf. Appendix B of~\cite{Chanyaswad2018}).

We conduct singular value decomposition on the queried result (i.e., the noisy empirical covariance matrix) to extract the  components of the dataset.  The considered evaluation metric is the total deviation which is defined as $\Delta = \textstyle\sum_{i=1}^{2041}\Delta_i = \textstyle\sum_{i=1}^{2041} |\lambda_i-\widetilde{\mathbf{v}_i}^T  \mathbf{C}  \widetilde{\mathbf{v}_i}|$, 
where $\mathbf{C}$ is the   ground-truth (non-private) empirical  covariance matrix, $\lambda_i$ is the $i$th eigenvalue of $\mathbf{C}$, and $\widetilde{\mathbf{v}_i}$ is the $i$th   differentially private  component (eigenvector). Specifically, $\Delta_i$ quantifies the deviation of the variance from $\lambda_i$ in the direction of the $i$th component. The smaller   $\Delta$ is, the higher   utility the obtained components have,  and  we have $\Delta = 0$   for the non-private baseline. 

Since   PCA   usually serves as a precursor to classification task, we also    evaluate the utility of the queried data by  projecting the dataset onto the subspace spanned by the first 20 PCs obtained from various noisy covariance matrices, and then  measure  the classification
accuracy using the projected data.  The higher the classification  accuracy, the higher the utility of the perturbed  covariance matrix. In the experiments, we use $60\%$ and $40\%$ of the dataset for training and testing, respectively, and the classification algorithm is the linear SVM. Note that the accuracy of the non-private baseline (i.e., classification using the original dataset project onto the first 20 PCs of the original covariance matrix) is $96.89\%$.

\textbf{Comparisons with other Mechanisms.} 
In addition to the    output perturbation mechanisms, we also consider the principal components obtained from 3 task-specific differentially private algorithms   designed specially  for PCA, i.e., (i) PPCA \cite{chaudhuri2012near}, which generates privacy-preserving components by sampling orthonormal matrices from  the matrix Bingham distribution parametermized by  the  original non-private  empirical covariance matrix scaled by the privacy budget,  (ii) MOD-SULQ \cite{chaudhuri2012near}, which directly perturbs the covariance matrix using   calibrated Gaussian noise, and (iii) the Wishart mechanism \cite{jiang2016wishart}, which perturbs the empirical covariance matrices using noise attributed to the Wishart distribution.

\begin{figure}[htb]
  \begin{center}
   \begin{tabular}{cc}
     \includegraphics[width= .47\columnwidth]{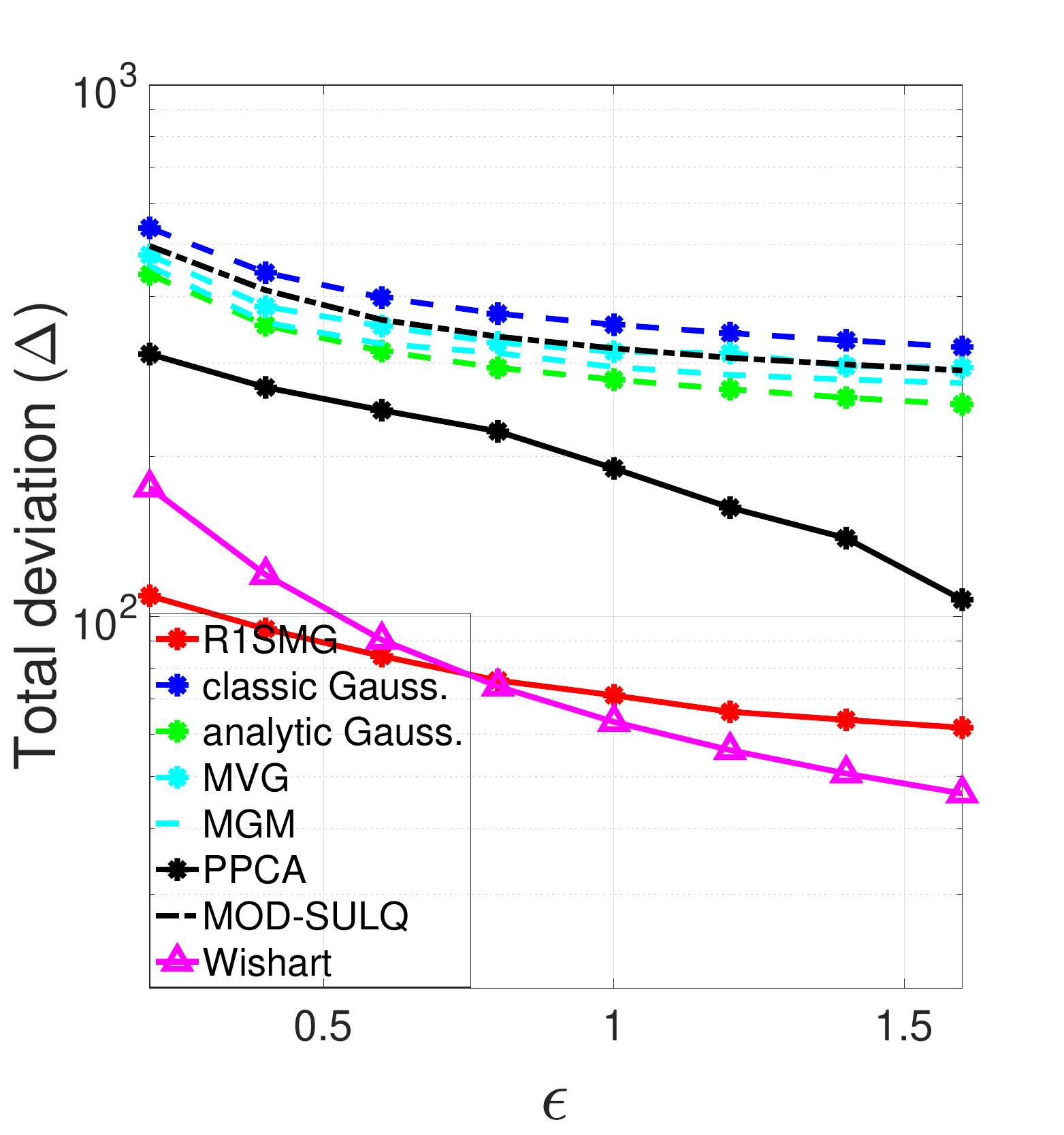}
&   \includegraphics[width= .47\columnwidth]{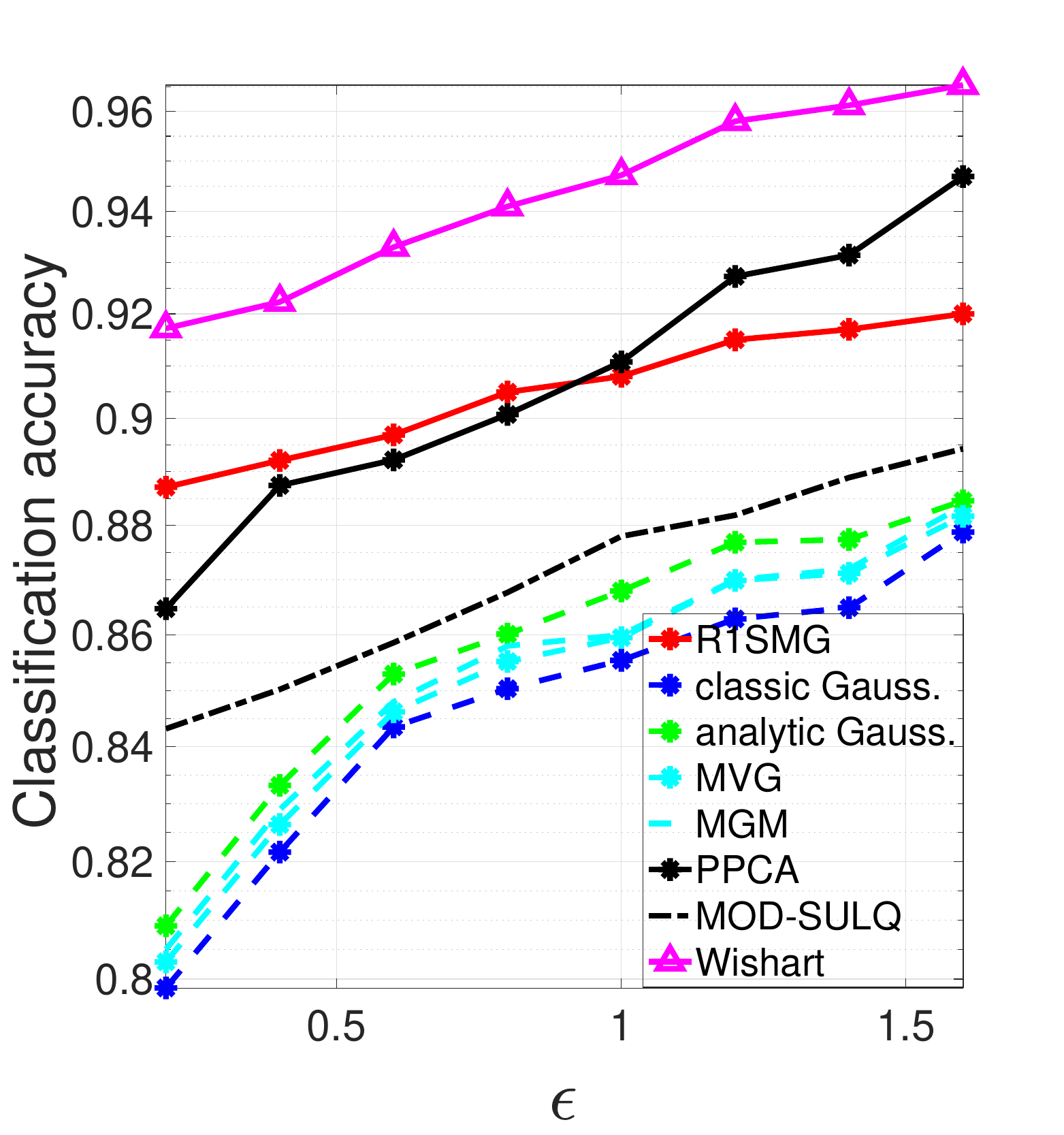}\\
	{\small  (a)   Deviation versus $\epsilon$.}&
    {\small  (b)   Accuracy versus $\epsilon$.} 
    \end{tabular}
       \end{center}
       \vspace{-4mm}
  \caption{\label{fig:pca_exp}Utility   of the noisy covariance matrices obtained by various mechanisms.    (a) The total deviation.    (b) The classification accuracy measured on the projected testing dataset.  $\epsilon$ used in the R1SMG mechanism is only $1/10^7$ of the others.
  }
\end{figure}

\textbf{Results.} In this experiment, we vary  $\epsilon$  from $0.2$ to $1.6$ and set  $\delta = 10^{-7}$ for all the comparing mechanisms and let the privacy budget of the R1SMG be only $1/10^7$ of the others according to the discussion in Remark~\ref{epsilon-remark} (more fair comparison results obtained by letting all  the other mechanisms use the same $\epsilon$ as R1SMG are discussed in  Appendix~\ref{app:small-eps}).

Figure \ref{fig:pca_exp}(a) and (b) plot the total derivation ($\Delta$) and classification accuracy on the projected testing data  obtained from all mechanisms. Obviously, the  R1SMG mechanism has high utility although its the privacy budget is extremely limited, e.g., $\epsilon<10^{-7}$, which  is also close to the non-private baselines (i.e.,   $\Delta=0$ and $96.89\%$ accuracy). 
This is because the R1SMG mechanism  introduces perturbation noise of much smaller magnitudes   in spite of very strict privacy guarantee, compared with the other mechanisms. Taking the MVG mechanism for example,   according to Theorem \ref{mvg_privacyguarantee}, even when $\epsilon = 1.6$ and $\delta=10^{-7}$,    MVG    requires $||\sigma(\S^{-1})||_2||\sigma(\P^{-1})||_2\leq  0.011$ to achieve the desired privacy guarantee. This  means that the covariance matrices will have very large singular values, which implies large expected accuracy loss (as discussed in Section \ref{sec:curse}) and causes the original data be overwhelmed by the additive   noise.  In fact, according to  \cite{Chanyaswad2018}, for  the MVG  mechanism to have a decent performance on differentially private PCA, it requires $\epsilon\geq 5$.  Moreover, in this case study, the  R1SMG mechanism  even  has similar utility performance of that of the task-specific mechanisms including  PPCA, MOD-SULQ, and Wishart. In other words, under very limited $\epsilon$, the R1SMG mechanism not only achieves very small $\Delta$, but also gives high classification accuracy. This  suggests that R1SMG is very promising in boosting the accuracy of high-dimensional  differentially private  query even if $\epsilon\ll 1$. Note that in this   study,  the R1SMG mechanism also achieves the most stable accuracy, i.e., the accuracy loss introduced by it has leptokurtic and right-skewed  histogram and   empirical PDF, but the others still have bell-shaped empirical PDFs (similar to Figure \ref{fig:utility_cost_L}). We omit the plots due to space limit.

\subsection{Case Study   \rom{3}: Deep Learning with DP}

In this case study, we  conduct preliminary validation to assess the  feasibility of replacing the classic Gaussian mechanism with the   R1SMG   in the DPSGD (differentially private stochastic gradient descent) optimization framework \cite{abadi2016deep}, and demonstrate that     R1SMG   can boost the utility (testing accuracy) of a differentially-private deep learning model and improve its  training efficiency by reducing the number of epochs.  Please note that this case study is not intended to be thorough or comprehensive, as differentially-private deep learning constitutes a distinct research direction in its own right.

As a proof of concept, we consider  the MNIST \cite{lecun1998gradient} and CIFAR-10 \cite{krizhevsky2014cifar} dataset, and compare our R1SMG-based DPSGD with the seminal work of DPSGD in deep learning (Abadi et al. CCS'16~\cite{abadi2016deep}).  For a given deep neural network, Abadi et al.  introduce calibrated i.i.d. Gaussian noise  ($\mathcal{N}(\mathbf{0},\sigma^2C^2\mathbf{I})$) to the  gradients computed on grouped batches of training data, where $C$ is a norm clipping parameter to control the sensitivity (Algorithm 1 in \cite{abadi2016deep}).  They also developed the privacy accountant scheme to tightly bound the cumulative privacy loss for the entire training process. In particular, the total privacy budget $(\epsilon,\delta)$ is determined by noise level $\sigma$, sampling ratio $q$, and the number of epochs $E$ (so number of   gradient descent steps is $T = E/q$).

In the experiment, we give our R1SMG-based DPSGD the same total privacy budget $(\epsilon,\delta)$, $C$, $q$, and $E$ as those of DPSGA in \cite{abadi2016deep}.  Then, we amortize $(\epsilon,\delta)$ to all   gradient descent steps in our R1SMG-based DPSGD, i.e., each step is assigned with $(\epsilon_0,\delta_0)$.  Although $\epsilon_0$ might go beyond the upper bound discussed in Remark~\ref{epsilon-remark}, it makes the comparison with DPSGD clear and straightforward. To be more specific, $(\epsilon_0,\delta_0)$  is obtained by solving the following equations formulated using  the strong composition theorem \cite{dwork2010boosting} considering $\frac{E}{q}$ repetitive executions of gradient descent.
\begin{equation}\label{eq:strong-composition}
\begin{cases}
&\epsilon = \sqrt{2\frac{E}{q}\ln\left(\frac{1}{\delta'}\right)} (q\epsilon_0) + \frac{E}{q} (q\epsilon_0) \left(e^{(q\epsilon_0)}-1\right)\\
&\delta = \frac{E}{q}(q\delta_0)+\delta'
\end{cases},
\end{equation}
where $q\epsilon_0$ and $q\delta_0$ are due to  privacy amplification via sampling \cite{beimel2014bounds,kasiviswanathan2011can} (see discussion on page 3, right column of \cite{abadi2016deep}). For example, based on the TensorFlow tutorial for DPSGD on MNIST \cite{tensorflowprivacy}, when $C=1$, $\sigma=1.1$, $q = \frac{256}{60000}$, and $E = 30$, it leads to $\epsilon = 1.795$. Then, setting $\delta = 10^{-5}$, $\delta' = 10^{-10}$ and solving  (\ref{eq:strong-composition}), we have $(\epsilon_0,\delta_0) = (0.71,3.33\times10^{-7})$.  When $C = 1$, $\sigma = 0.5$, $q = \frac{256}{60000}$, and $E = 30$, it leads  to $(16.983, 10^{-5})$-DP for   DPSGD, which makes  
$(\epsilon_0,\delta_0) = (5.42,3.33\times10^{-7})$ for each gradient descent step in our R1SMG-based DPSGD.

Note that using strong composition theorem to assign amortized privacy budget to each step of our R1SMG-based DPSGD   \uline{gives more favor to the classic Gaussian mechanism based DPSGD in the comparison.} This is because strong composition is much looser than the privacy accountant scheme used in     DPSGD (empirically validated by  Figure 2 of \cite{abadi2016deep}). Thus, the amortized  
$(\epsilon_0,\delta_0)$ can be limited.  Developing a tight composition framework for the R1SMG mechanism in SGD-based learning will be a separate research topic.   

Even though the original DPSGD is given more favor in the comparison,  our R1SMG-based DPSGD still outperforms it by achieving higher training and testing accuracy. We first show the comparison results on MNIST using  the above parameters setups  in Figure \ref{fig:cnn-plot}. In particular, Figure \ref{fig:cnn-plot}(a) and (b) show the training and testing accuracy on MNIST for the first 30 epochs achieved by various approaches when the entire learning process is $(1.795,10^{-5})$-DP and $(16.983,10^{-5})$-DP (parameters discussed above), respectively. The blue lines and dashed lines in Figure \ref{fig:cnn-plot} are the result of non-private SGD (the one uses original gradients calculated from the data).

\begin{figure}[htb]
  \begin{center}
   \begin{tabular}{cc}
     \includegraphics[width= 0.47\columnwidth]{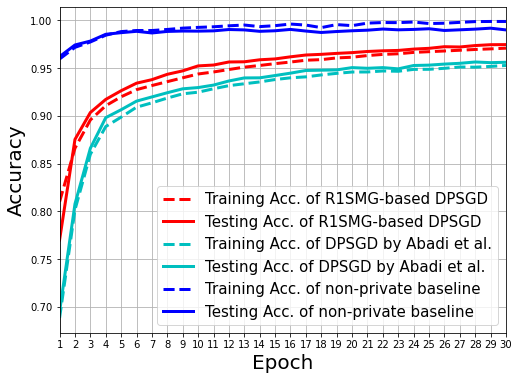} &
      \includegraphics[width= 0.47\columnwidth]{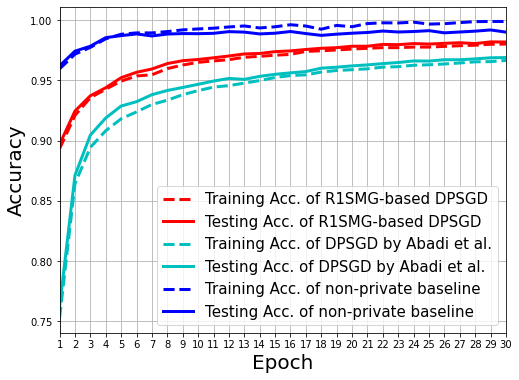} \\
     {(a)   The entire   process} 
&
     {(b)  The entire   process}\\
     {is $(1.795,10^{-5})$-DP.} & {is $(16.983,10^{-5})$-DP.}\\
              \end{tabular}
       \end{center} 
  \caption{\label{fig:cnn-plot}Training and testing acc. on MNIST.}
\end{figure}

Compared with   DPSGD, R1SMG-based DPSGD   achieves higher training and testing accuracy in just a few epochs for the MNIST dataset. For example, when the total privacy budget  is $(16.983,10^{-5})$, our method achieves 95\% accuracy on the testing dataset using 5 epochs, whereas   DPSGD requires 14 epochs. Our R1SMG-based DPSGD also leads to performance that is closer to the non-private baseline, because even though the amortized privacy budget for each R1SMG-based DPSGD step is limited, the magnitude of the perturbation noise introduced by the R1SMG mechanism is much less than that required by the classic Gaussian in DPSGD.

\begin{figure}[htb]
  \begin{center}
   \begin{tabular}{cc}
     \includegraphics[width= 0.47\columnwidth]{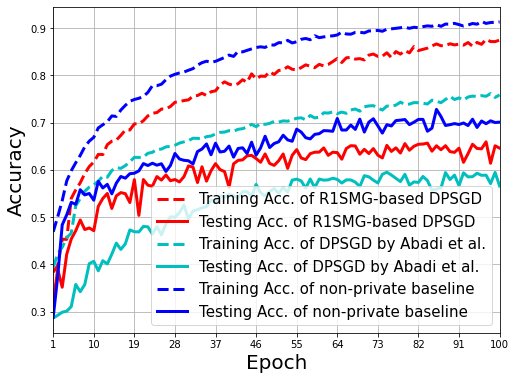} &
      \includegraphics[width= 0.47\columnwidth]{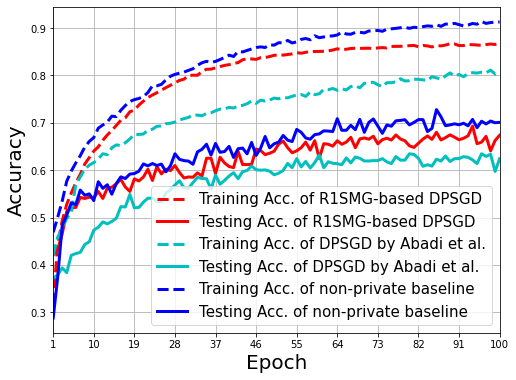} \\
     {(a)   The entire   process} 
&
     {(b)  The entire   process}\\
     {is $(2.673,10^{-5})$-DP.} & {is $(20.713,10^{-5})$-DP.}\\
              \end{tabular}
       \end{center} 
  \caption{\label{fig:cnn-plot-cifar}Training and testing acc. on CIFAR-10.}
\end{figure}

Then, we conduct the experiments on CIFAR-10 dataset with the same parameter setups while using 100 epochs, i.e., $E = 100$. We also adopt  the   network architecture from the TensorFlow CNN  tutorial \cite{tensorflowprivacy-cifar10}. The privacy budget for each step of R1SMG-based DPSGD, i.e., $(\epsilon_0,\delta_0)$ is also decided by solving (\ref{eq:strong-composition}) numerically. In particular, when $\sigma=1.1$ and $\sigma = 0.5$, the entire process of DPSGD is $(2.673,10^{-5})$-DP and $(20.713,10^{-5})$-DP, respectively, which make $(\epsilon_0,\delta_0)$ become  $(0.54,10^{-7})$ and  $(3.52,10^{-7})$, respectively, for R1SMG-based DPSGD. The comparison results are shown in Figure \ref{fig:cnn-plot-cifar}. 
Clearly, our R1SMG-based DPSGD still outperforms DPSGD, i.e., both   training and testing accuracy achieved by R1SMG-based DPSGD are closer to the non-private baselines compared with those achieved by DPSGD. In particular, R1SMG-based DPSGD improve the training efficiency by achieving a higher training accuracy (e.g., $>80\%$) in much fewer epochs compared with the original DPSGD.

%% file: sections/related_work.tex
\section{Related Work}
\label{rw}
Many works have attempted  to improve the classic Gaussian mechanism. 
Here, we review some representative ones. 

The analytic Gaussian mechanism proposed by Balle et al. \cite{balle2018improving} (see Theorem \ref{def_AG}) improves the classic Gaussian mechanism   by calibrating the variance of the Gaussian noise directly using the Gaussian cumulative
density function instead of the tail bound approximation, and develop an analytical solution to the variance given specific choices $\epsilon$ and $\delta$. Although this mechanism can reduce the magnitude of the noise when perturbing a scalar-valued quantity, it still suffers the curse of full-rank covariance matrices when applied to perturb  high-dimensional  query result (as  discussed in Section \ref{sec:curse}).

Zhao et al. \cite{zhao2019reviewing} investigate the classic Gaussian mechanism under high privacy budgets, i.e., $\epsilon>1$, and derive a closed-form upper bounds for the noise variance used in the analytic Gaussian mechanism. Their mechanism can achieve utility higher than the  classic Gaussian mechanisms, but still lower than the  analytic Gaussian mechanism. 

Chanyaswad et al. \cite{Chanyaswad2018} (see Definition \ref{def_MVG} and Theorem \ref{mvg_privacyguarantee}) propose the MVG mechanism and develop the technique of directional noise to restrict the impact of the  perturbation noise on the utility of a matrix-valued query function. However, to obtain  the directional noise, one either needs a domain expert to determine the principal components of the queried data or has to calculate them via principal component analysis which consumes some privacy budget.

Some other variants of the classic Gaussian mechanism includes the Generalized Gaussian mechanism \cite{liu2018generalized}, which investigates noise calibration under the   general $l_p$ sensitivity, and the discrete Gaussian mechanism \cite{NEURIPS2020_b53b3a3d}, which uses noise attributed to discrete Gaussian distribution to protect the privacy of query results in a discrete domain.

Our proposed  R1SMG mechanism differs with the above  mechanisms, as we explicitly require the covariance matrix of the additive noise to be rank-1. Thus, the R1SMG mechanism lifts the identified curse  and  boosts the utility of   query results leading to guaranteed DP at much lower cost of accuracy loss.

%% file: sections/conclusions.tex
\section{Conclusions}
\label{sec:con}

In this paper, we  developed a novel DP mechanism, i.e., R1SMG, to boost the utility and stability of differentially-private  query results.  In particular, we   first  identify the curse of full-rank covariance matrices in all existing Gaussian-based DP mechanisms. Then, we lift the curse by   developing   R1SMG which perturbs the query results using noise that follows a 
singular multivariate Gaussian distribution with a random rank-1 covariance matrix. We rigorously analyze the privacy guarantee of  the R1SMG mechanism, and theoretically demonstrate that it can achieve much lower accuracy loss, particularly, on a lower order of magnitude by at least $M$ or $MN$,    and much higher accuracy stability as well, than   the classic Gaussian, analytic Gaussian, and the MVG mechanisms.

\section*{Acknowledgement}
This work was supported in part by the National Science Foundation under Grants EEC-2133630 and CNS-2125460.

%% file: sections/appendix.tex
\appendix

\section{Proof of (e) in Equation (\ref{eq:mvg_order})}\label{app:HM-GM}

Define $\sigma(\S) \defeq   [a_1, a_2,\ldots, a_M]$ and $\sigma(\P) \defeq  [b_1, b_2,\ldots, b_N]$. Since $\S, \P\in\mathbb{PD}$ and they are both symmetric, we have $a_m>0,\forall m\in[1,M]$ and $b_n>0,\forall n\in[1,N]$. Then, to prove (e) in   (\ref{eq:mvg_order}), we essentially need to show that
\begin{equation}\label{eq:needprove}
\begin{aligned}
\sqrt{\sum_{m=1}^Ma_m^2}\times \sqrt{\sum_{n=1}^Nb_n^2}\geq \frac{M}{\sqrt{\sum_{m=1}^M\frac{1}{a_m^2}}} \times \frac{N}{\sqrt{\sum_{n=1}^M\frac{1}{b_n^2}}}.
\end{aligned}
\end{equation}

Apply   harmonic mean-geometric mean inequality to $a_1^2, a_2^2, \ldots, a_M^2$, we have 

\begin{equation}\label{eq:p1}
    \frac{a_1^2+ a_2^2+ \ldots+ a_M^2}{M} \geq \left(a_1^2\cdot a_2^2\cdot \ldots\cdot a_M^2   \right)^\frac{1}{M}.
\end{equation}
Similarly, for $\frac{1}{a_1^2}, \frac{1}{a_2^2}, \ldots, \frac{1}{a_M^2}$, we have
\begin{equation}\label{eq:p2}
    \frac{\frac{1}{a_1^2}+ \frac{1}{a_2^2}+ \ldots+ \frac{1}{a_M^2}}{M}\geq \left(\frac{1}{a_1^2}\cdot \frac{1}{a_2^2}\cdot \ldots\cdot \frac{1}{a_M^2}  \right)^\frac{1}{M}.
\end{equation}
Multiplying (\ref{eq:p1}) and (\ref{eq:p2}) gives 
\begin{equation*}
    \begin{aligned}
 &\frac{a_1^2+ a_2^2+ \ldots+ a_M^2}{M}     \times \frac{\frac{1}{a_1^2}+ \frac{1}{a_2^2}+ \ldots+ \frac{1}{a_M^2}}{M}\\
\geq  &\left(a_1^2\cdot a_2^2\cdot \ldots\cdot a_M^2   \right)^\frac{1}{M} \left(\frac{1}{a_1^2}\cdot \frac{1}{a_2^2}\cdot \ldots\cdot \frac{1}{a_M^2}  \right)^\frac{1}{M} = 1.
    \end{aligned}
\end{equation*}
Taking square root of the above, we have
    \begin{equation*}
  \frac{\sqrt{\sum_{m=1}^Ma_m^2}}{\sqrt{M}}\times \frac{\sqrt{\sum_{m=1}^M\frac{1}{a_m^2}}}{\sqrt{M}}\geq 1 
 \Leftrightarrow \sqrt{\sum_{m=1}^Ma_m^2} \geq  \frac{M}{\sqrt{\sum_{m=1}^M\frac{1}{a_m^2}}}.
\end{equation*}

By applying the same procedure on $ [b_1, b_2,\ldots, b_N]$, we have $\sqrt{\sum_{n=1}^Nb_n^2} \geq  \frac{N}{\sqrt{\sum_{n=1}^N\frac{1}{b_n^2}}}$. Thus, (\ref{eq:needprove}) is proved, and the step   (e) in   (\ref{eq:mvg_order}) follows.

\section{$\mathbb{E}[\L]$ versus $M$ for Various Mechanisms}\label{sec:versus_M}

Next, we  empirically  investigate the accuracy loss of the query results $f(\bm{x})\in\R^M$ when $M$   increases. In Figure \ref{fig:scale_study}, by assuming $\Delta_2f=1$, varying  $M$ from $10^2$ to $10^{6}$, and fixing $\delta = 10^{-10}$, we plot the expected accuracy loss introduced by the     R1SMG, classic Gaussian, and analytic Gaussian mechanisms when $\epsilon\in\{10^{-7},10^{-8},10^{-9}\}$. Note that we do not include MVG  mechanism  in the comparison, because (i) it  requires the prior knowledge of $\gamma = \sup_{\bm{x}}||f(\bm{x})||_F$, which depends on specific datasets and  applications (e.g., see Theorem \ref{mvg_privacyguarantee}), and (ii) it will introduce more noise than the analytic Gaussian mechanism. 
\begin{figure}[htb]
 \begin{center}
     \includegraphics[width=  0.9\columnwidth]
     {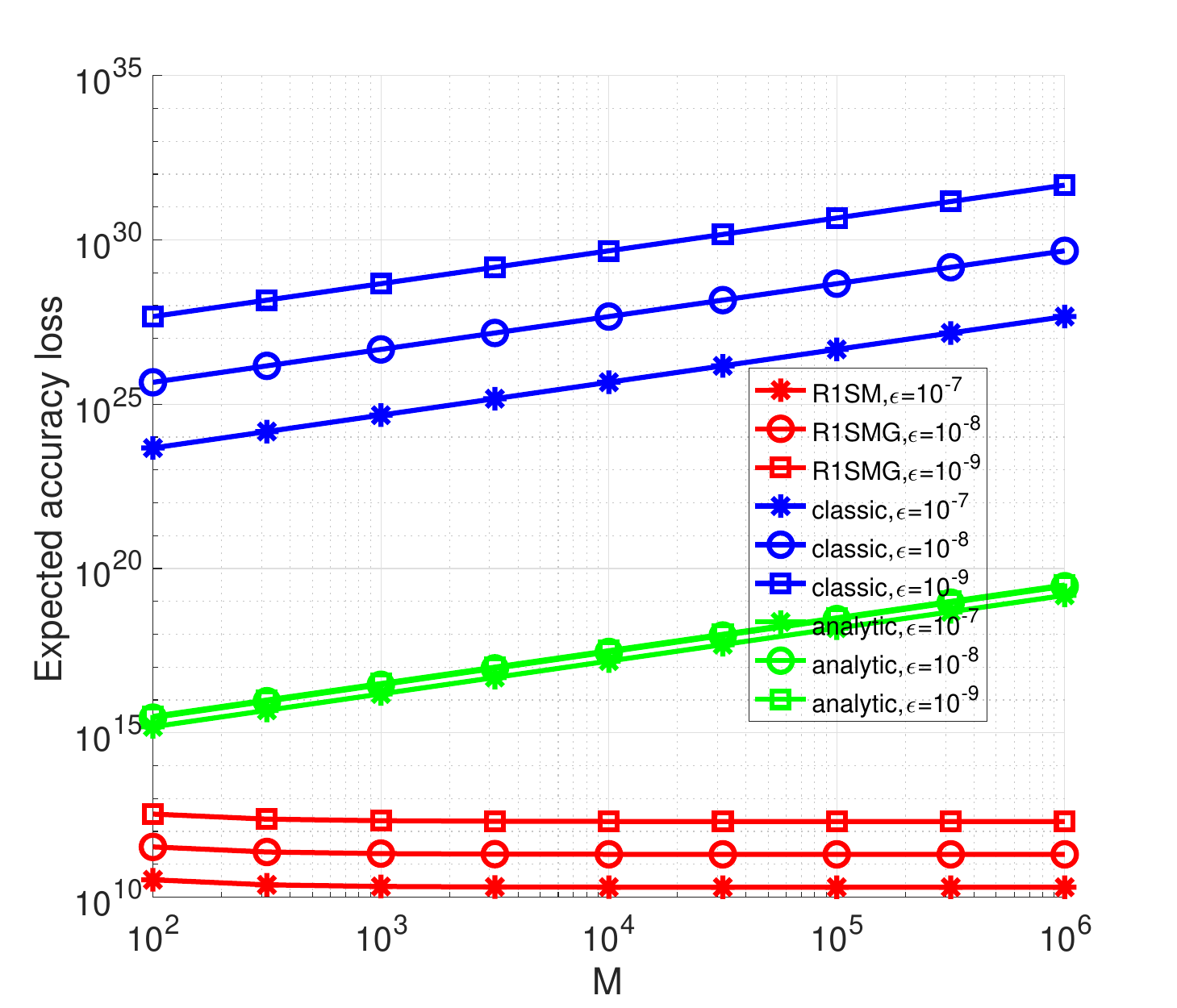}\\
       \end{center}
       \vspace{-3mm}
  \caption{\label{fig:scale_study}Expected accuracy loss of the R1SMG, classic Gaussian, and analytic Gaussian mechanisms under increasing query size  when $\epsilon\in\{10^{-7},10^{-8},10^{-9}\}$ and $\delta = 10^{-10}$.}
\end{figure}

From Figure \ref{fig:scale_study}, we observe that, for both classic and analytic Gaussian mechanisms, the expected accuracy loss (i.e., $\Tr[\sigma^2\mathbf{I}_{M\times M}]$ and $\Tr[\sigma_A^2\mathbf{I}_{M\times M}]$) scales linearly with $M$, which  validates our theoretical findings in Section \ref{sec:curse}. In contrast, as $M$ increases, the expected accuracy loss  incurred  by the R1SMG mechanism (i.e., $\Tr[\bm{\Pi}]=\sigma_*$) decreases and converges to $\frac{2(\Delta_2f)^2}{\epsilon}$. 
Thus, the   R1SMG mechanism not only breaks the curse of full-rank covariance matrices, but also significantly improves the accuracy (utility) of the perturbed $f(\bm{x})$, as it is able to achieve $(\epsilon,\delta)$-DP by using much weaker noise even when   the dimension of $f(\bm{x})$ increases.

\section{Potential Improvement of the MVG}\label{app:mvg-improve} 
When studying the MVG mechanism, we observe one potential solution to improve its privacy guarantee. 
In particular, one intermediate step in the proof of the MVG mechanism requires deriving an upper bound for 
\begin{small}
\begin{equation}\label{eq:mvg-tech}
\begin{aligned}
        \clubsuit = &\Tr\Big[\P^{-1}\mathbf{Y}^T\S^{-1}\Delta+\P^{-1}\Delta^T\S^{-1}\mathbf{Y}\\
        &+\P^{-1}f(D_2)^T\S^{-1}f(D_2)-\P^{-1}f(D_1)^T\S^{-1}f(D_1)\Big],
\end{aligned}
\end{equation}
\end{small}\ignorespacesafterend
where $\mathbf{Y}$ stands for the output of the MVG mechanism and $\Delta = f(D_1)-f(D_2)$ (see \cite{Chanyaswad2018} page 244, right column). Due to the negative sign in  the fourth term, the authors simply bound its absolute value, i.e., $\big|\Tr\big[\P^{-1}f(D_1)^T\S^{-1}f(D_1)\big]\big|$, which results in a very loose upper bound related to  the quadratic form in the   Frobenius norm of the matrix query result (i.e, $\gamma^2$ in Theorem \ref{mvg_privacyguarantee}), and further compromises the utility.

In fact, the original proof of the MVG mechanism can be improved if we apply the transformation 
\begin{small}
    \begin{equation*}
    \begin{aligned}
        &\Tr\big[\P^{-1}f(D_2)^T\S^{-1}f(D_2)-\P^{-1}f(D_1)^T\S^{-1}f(D_1)\big]=\\
        & \Tr\big[\P^{-1} \big(f(D_2)^T\S^{-1}\big(f(D_2)-f(D_1)\big) +\big(f(D_2)-f(D_1)\big)^T\S^{-1}f(D_1)\big)\big],
    \end{aligned}
\end{equation*}
\end{small}\ignorespacesafterend
and  rewrite $\clubsuit$ as
\begin{small}
\begin{equation*}
    \clubsuit =\Tr\Big[\P^{-1}\big(\mathbf{Y}-f(D_2)\big)^T\S^{-1}\Delta+\P^{-1}\Delta^T\S^{-1}\big(\mathbf{Y}-f(D_1)\big)\Big].
\end{equation*}
\end{small}\ignorespacesafterend
Then, we can bound    $\clubsuit$ using singular values, and hence the  quadratic form in the Frobenius norm of the   query result  can be completely removed. We do not follow the above   steps to further develop an improved version of the MVG mechanism, because the improved version is still     a ``victim'' of the identified  curse of full-rank covariance matrices.

\section{Performance of the Other DP Mechanisms under Extremely Strict Privacy Guarantee  }\label{app:small-eps}

\noindent\textbf{Case Study I.} In Figure~\ref{fig:nyc-small-eps}, we visualize the differentially private pickup counts achieved by the other mechanisms when they are under the same strict privacy budget as that of the R1SMG mechanism, i.e.,  $\epsilon = 10^{-5}$. Clearly, under extreme restricted privacy guarantee, these mechanisms achieve very poor utility as the number of pickups are on the order of $10^5$.

\noindent\textbf{Case Study II.} In Figure~\ref{fig:pca-small-eps-others}, we show the total deviation caused by the other mechanisms when $\epsilon$ varies from $0.2\times 10^{-7}$ to $1.6\times 10^{-7}$ (i.e., the same privacy budget as that of the R1SMG mechanism). Under   extremely strict privacy requirement, we can clearly observe that all the other DP mechanisms also achieve very poor utility, because $\Delta$ is on   the order of $10^3$ while $\Delta$ is around 50  for the R1SMG mechanism as shown in Figure~\ref{fig:pca_exp}(a). Besides, when $\epsilon$ varies from $0.2\times 10^{-7}$ to $1.6\times 10^{-7}$, the classification accuracy  achieved by all the other DP  mechanisms
are around 0.5 (close to random guessing). Thus, we omit the plot here.

\begin{figure}[htb]
  \begin{center}
   \begin{tabular}{cc}
 \includegraphics[width= .45\columnwidth]{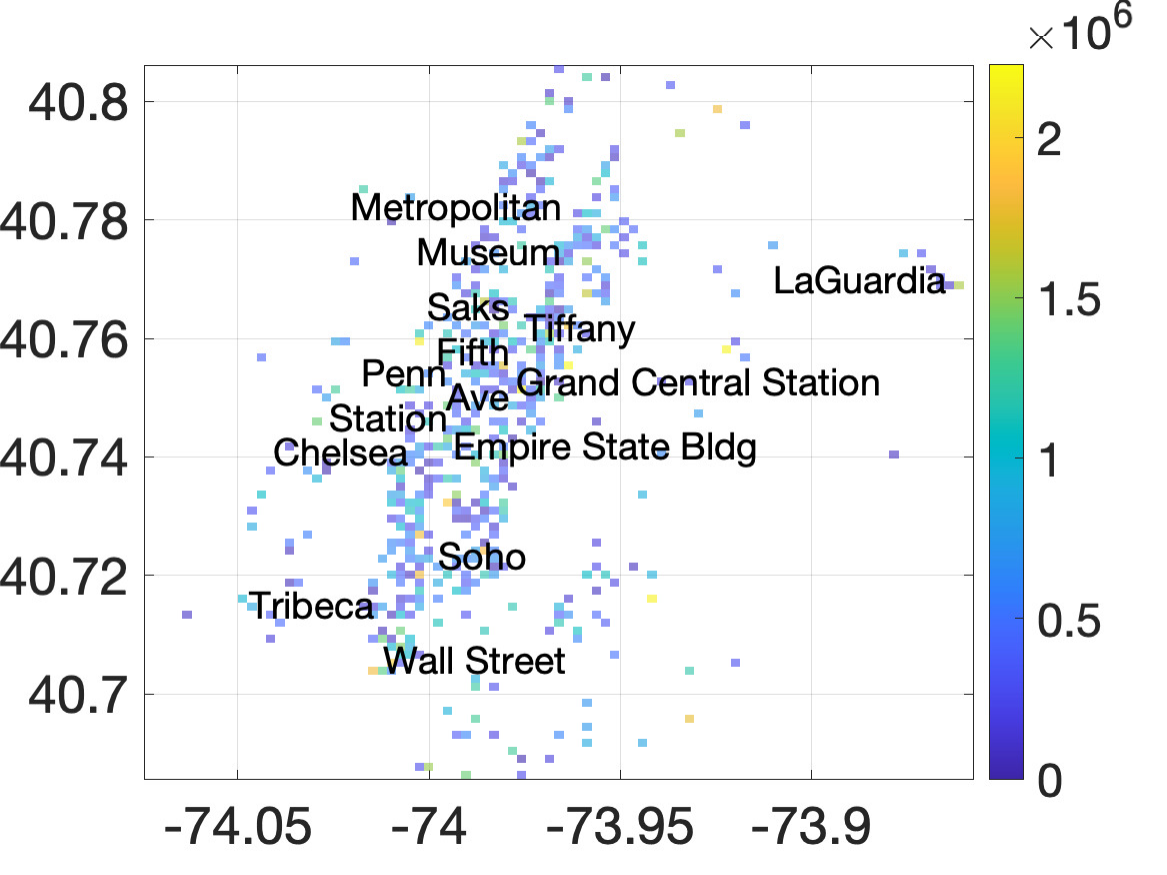}
&\includegraphics[width= .45\columnwidth]{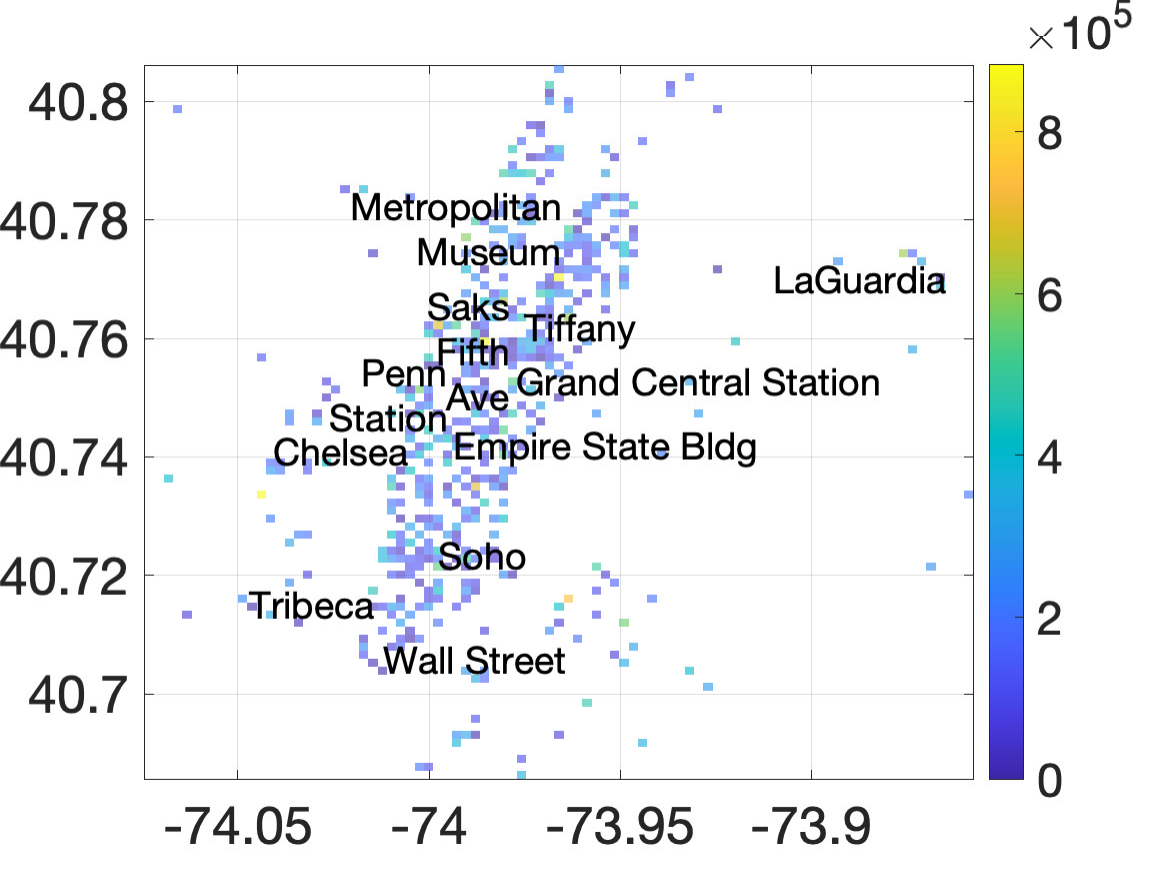}\\
     {\small  (a)    classic Gauss.}&
      {\small  b-eps-converted-to.pdfd)     Analytic   Gauss.}\\
\includegraphics[width= .45\columnwidth]{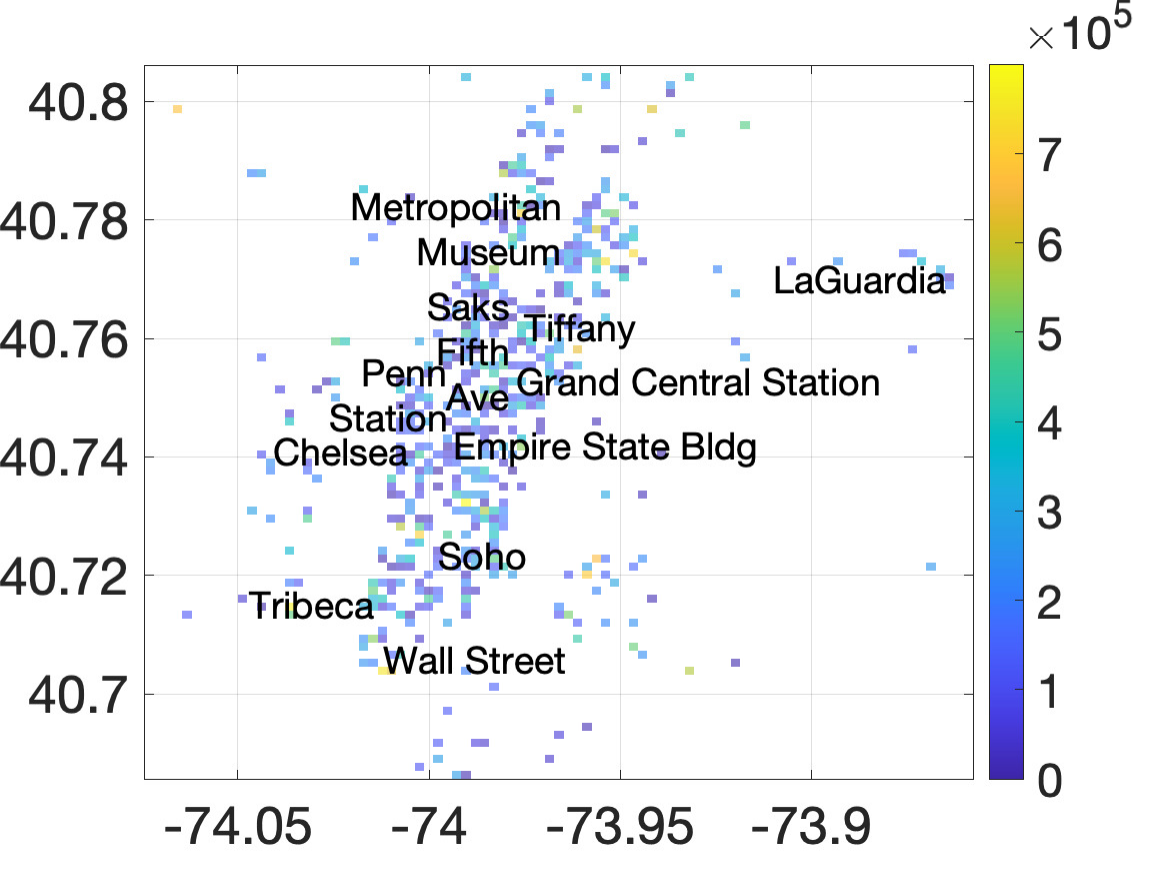}
&   \includegraphics[width=.45\columnwidth]{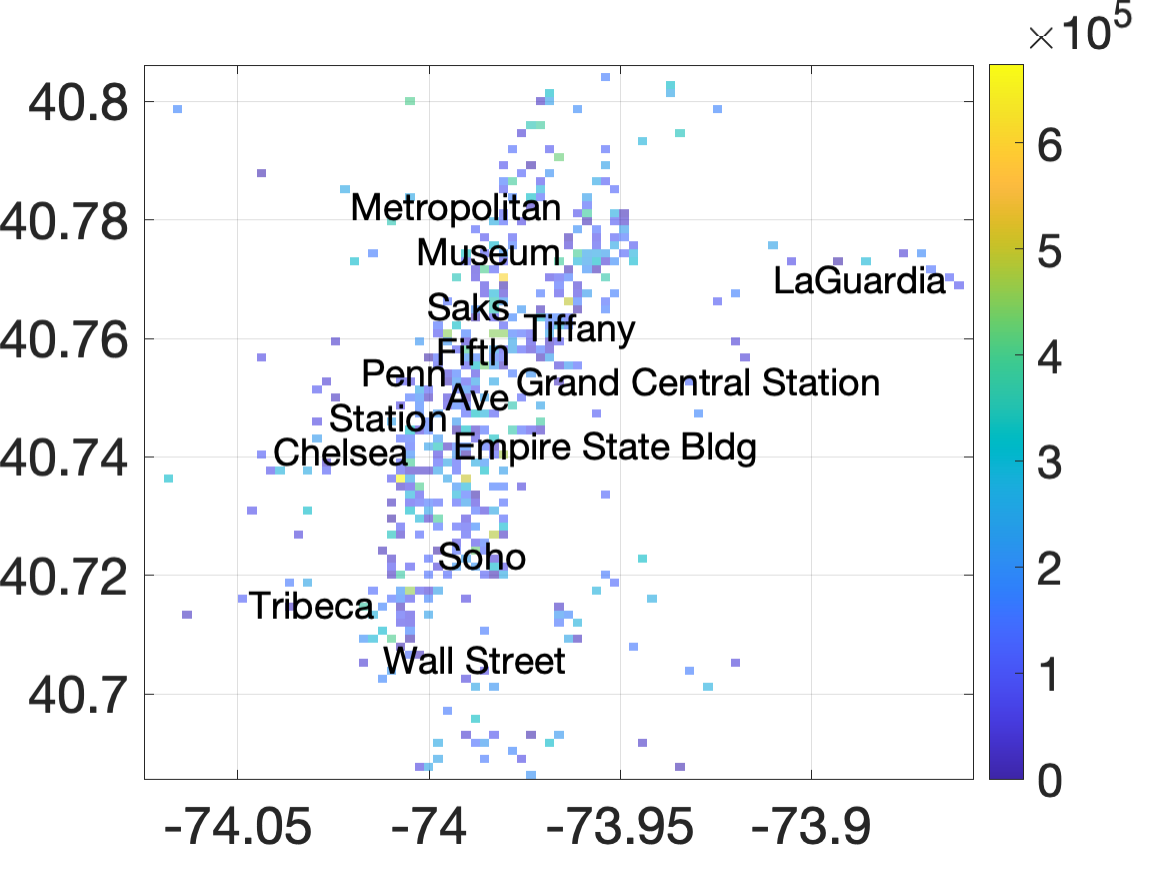}\\
          	{\small  (c)    MVG}&
    {\small  (d)   MGM}\\
\includegraphics[width= .45\columnwidth]{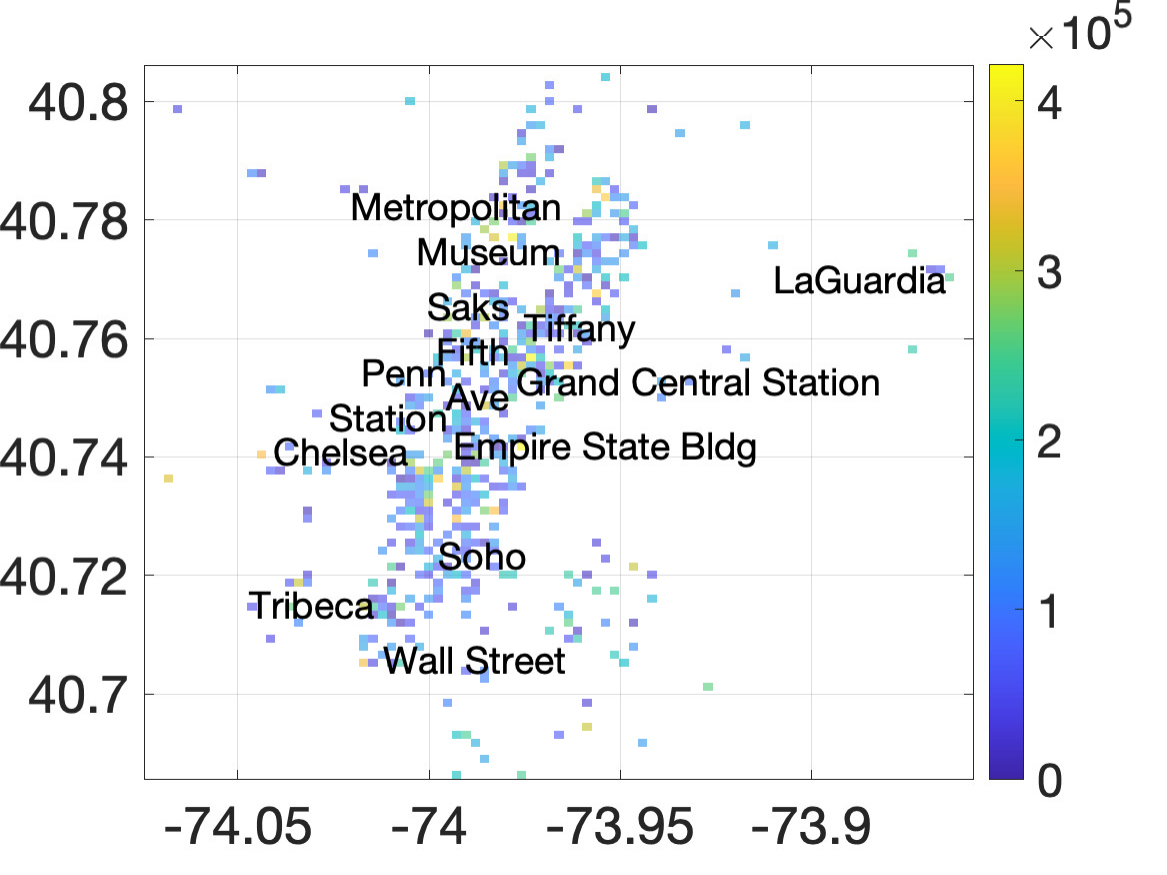}
&\includegraphics[width= .45\columnwidth]{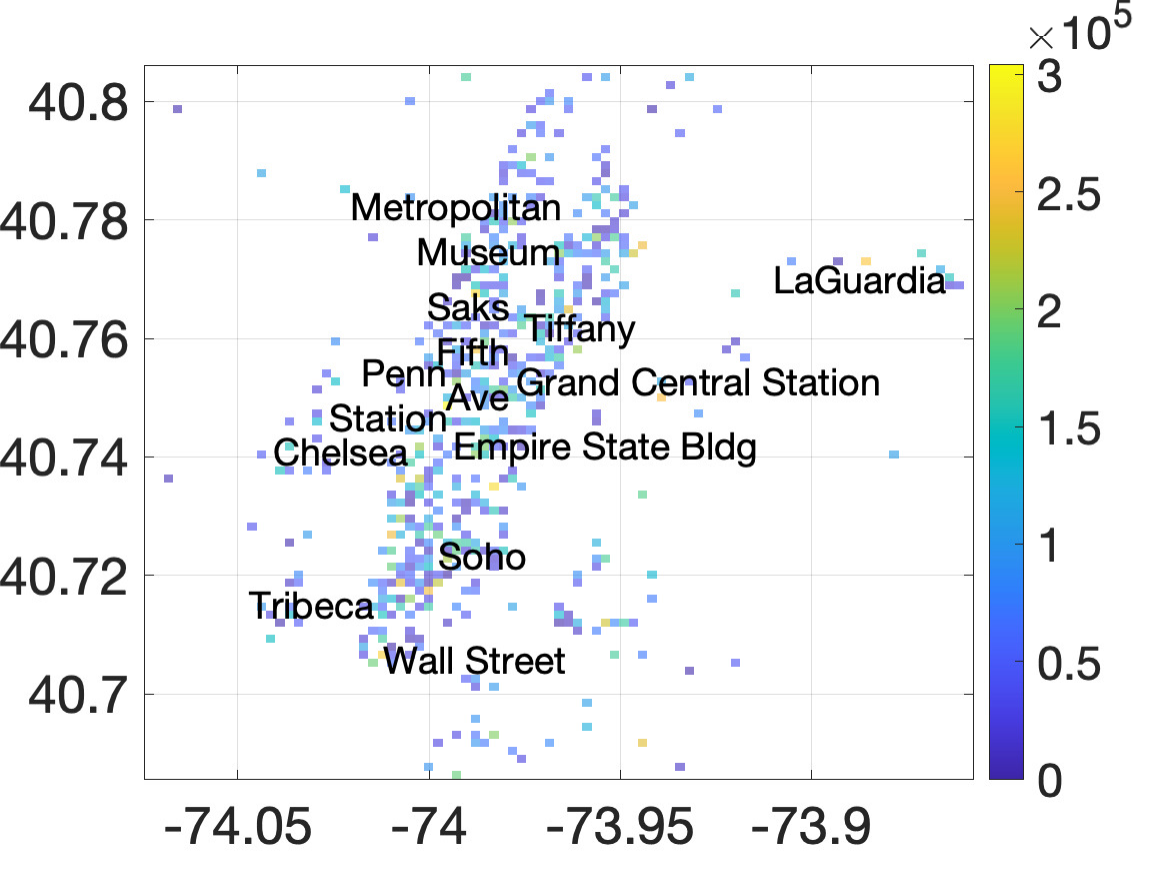}\\
       {\small  (e)     DAWA}&{\small  (f)     $H_b$}
    \end{tabular}
       \end{center}
       \vspace{-4mm}
  \caption{\label{fig:nyc-small-eps}Visualization of   query count of Uber pickups.    
  (a)-(f) are differentially private 2D counts obtained by  the    classic Gaussian,   analytic Gaussian,     MVG,   MGM,    DAWA, and $H_b$ mechanisms, respectively when $\epsilon=10^{-5}$ and  $\delta = 10^{-7}$.}
\end{figure}

\begin{figure}[htb]
 \begin{center}
     \includegraphics[width=  0.7 \columnwidth]
     {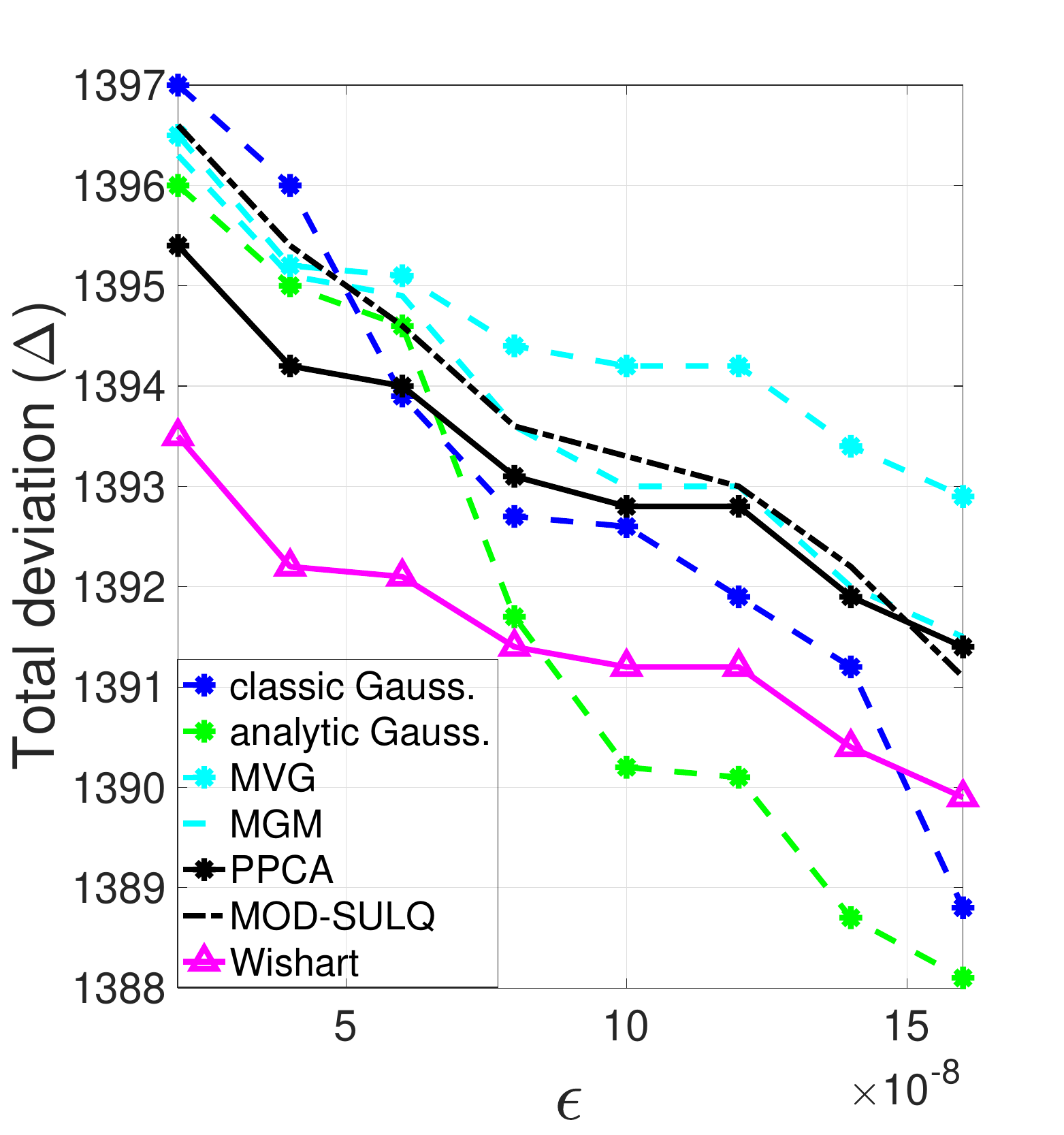}\\
       \end{center}
  \caption{\label{fig:pca-small-eps-others}The classification accuracy measured on the projected testing dataset when $\epsilon$ varies from $0.2\times 10^{-7}$ to $1.6\times 10^{-7}$.}
\end{figure}